\documentclass[english, a4paper]{lipics_hacked}
\usepackage[utf8]{inputenc}
\usepackage{microtype}
\usepackage{amsmath,longtable,float,etoolbox,tikz,algorithm2e}
\usetikzlibrary{graphs,quotes,decorations.shapes,decorations.pathreplacing}
\title{Towards Completely Characterizing the Complexity of Boolean Nets Synthesis}

\author{Ronny Tredup}
\author{Christian Rosenke}
\authorrunning{R. Tredup and C. Rosenke}
\Copyright{Ronny Tredup and Christian Rosenke}
\subjclass{Software system structures $\rightarrow$ Petri nets, Theory of computation $\rightarrow$ Problems, reductions and completeness}
%mandatory

\newcommand{\edgeScale}{0.75}
\newcommand{\nodeScale}{0.9}
\newcommand{\myEdge}[2]{ \tikz[baseline=-1pt]{
\draw[#2,line width=0.3pt] (0,0) -- ++(0.6,0) node[anchor=base, yshift=2pt, pos=0.5] {\scalebox{\edgeScale}{$#1$}};
}}
\newcommand{\edge}[1]{\myEdge{#1}{->}}
\newcommand{\bedge}[1]{\myEdge{#1}{<-}}
\newcommand{\fbedge}[1]{\myEdge{#1}{<->}}

\newcommand{\nop}{\ensuremath{\textsf{nop}}}
\newcommand{\inp}{\ensuremath{\textsf{inp}}}
\newcommand{\out}{\ensuremath{\textsf{out}}}
\newcommand{\set}{\ensuremath{\textsf{set}}}
\newcommand{\res}{\ensuremath{\textsf{res}}}
\newcommand{\swap}{\ensuremath{\textsf{swap}}}
\newcommand{\free}{\ensuremath{\textsf{free}}}
\newcommand{\used}{\ensuremath{\textsf{used}}}
\newcommand{\test}{\ensuremath{\textsf{test}}}
\newcommand{\exit}{\ensuremath{\textsf{exit}}}
\newcommand{\enter}{\ensuremath{\textsf{enter}}}
\newcommand{\keepo}{\ensuremath{\textsf{keep}^+}}
\newcommand{\keepze}{\ensuremath{\textsf{keep}^-}}
\newcommand{\R}{\ensuremath{\mathcal{R}}}

\begin{document}

\maketitle

%%%%%%%%%%%%%%%%%%%%%%%%%%%%%%%%%%%%%%%%%%%%%%%%%%%%%%%%%%%%%%%%%%%%%%%%%%%%%%%

%TODO: shorten abstract more
\begin{abstract}
Boolean Petri nets, which tolerate at most one token per place, are widely regarded as a fundamental model for concurrent systems.
They are differentiated into types of nets by the variety of applications that define their individual interaction set between places and transitions.
Taking a step back, one observes that only the eight interactions \emph{no operation} (\nop), \emph{input} (\inp), \emph{output} (\out), \set, \emph{reset} (\res), \swap, \emph{test of occupation} (\used), and \emph{test of disposability} (\free) have been used so far.
This paper argues that there are no other interactions for (deterministic) boolean nets and, thus, restrains their family to 256 types by considering every possible interaction subset.
Yet, research has explicitly defined seven of them:
\emph{Elementary net systems} (\nop, \inp, \out), for instance, have a connection to \emph{prime event structures} \cite{NRT1990,RE1996} and applications in workflow management systems like \textsc{milano} \cite{AM2000}.
\emph{Contextual nets} (\nop, \inp, \out, \used, \free), as a second example, implement \emph{reading without consuming} common to database systems, concurrent constraint programming, and shared memory systems \cite{MR1995}.
Other known classes are \emph{event/condition nets} (\nop, \inp, \out, \used) \cite{T1986}, \emph{inhibitor nets} (\nop, \inp, \out, \free) \cite{PK1997}, \emph{set nets} (\nop, \inp, \set, \used) \cite{KKP2013}, \emph{trace nets} (\nop, \inp, \out, \set, \res, \used, \free) \cite{BD1995}, and \emph{flip flop nets} (\nop, \inp, \out, \swap) \cite{S1996}.

This paper is devoted to a computational complexity analysis of the boolean net synthesis problem subject to a target class $\tau$.
The challenge is to translate given finite automata $A$, called \emph{transition systems} (TSs, for short), into boolean $\tau$-nets having a state transition behavior as specified by $A$.
This problem's complexity has yet been shown only for elementary net systems, where it is NP-complete to decide if general TSs \cite{BBD1997} or even considerably restricted TSs \cite{TRW2018,TR2018} can be synthesized, and for flip flop nets, which can be synthesized in polynomial time \cite{S1996}.

Our main result is a generic reduction scheme for NP-hardness proofs of boolean net synthesis that works for 77 different net classes allowing \nop.
We significantly generalize preliminary methods used in \cite{TRW2018,TR2018} for the hardness proof of synthesizing elementary net systems from heavily restricted TSs.
Unlike these premature approaches, the present solution covers all supersets of (\nop, \inp, \out) that exclude \swap, all supersets of either (\nop, \inp, \set) or (\nop, \out, \res), and extensions of (\nop, \swap) by at least one interaction from both, (\set, \res) and (\inp, \out, \used, \free).
Meanwhile it keeps the strong properties like low degree bounds.

We also identify seven classes with hard synthesis where the reduction does not fit into the general scheme.
This comprises (\nop, \set, \res) extended with at least one of \used\ and \free\ as well as (\nop, \inp, \free, [\used]) and (\nop, \out, \used, [\free]), where [$\cdot$] marks optional interactions.

Aside from this, we find 36 tractable cases for boolean net synthesis.
Firstly, like flip flop nets, the 16 extensions of (\nop, \swap) with a subset of (\inp, \out, \used, \free) can be synthesized in polynomial time by a version of Schmitt's approach \cite{S1996}.
Secondly, we presents a new polynomial time synthesis algorithm that works for the 16 classes combining (\nop, \set) with a subset of (\out, \used, \free) and combining (\nop, \res) with a subset of (\inp, \used, \free).
Four rather simple cases of polynomial synthesis are found in extending (\nop) with subsets of (\used, \free).

Regarding the set of all classes allowing \nop, we leave synthesis complexity open for the eight remaining cases of (\nop, \inp, [\used]), (\nop, \out, [\free]), (\nop, \set, \res) and (\nop, \swap) extended with at least one of \set\ and \res.
\end{abstract}

%%%%%%%%%%%%%%%%%%%%%%%%%%%%%%%%%%%%%%%%%%%%%%%%%%%%%%%%%%%%%%%%%%%%%%%%%%%%%%%

\section{Introduction}\label{sec:introduction}

This paper contributes to the analysis of the computational complexity of boolean Petri net synthesis as a function of the specific net class.
While the efficient algorithms developed in this paper attack the synthesis problem itself, the proofs for intractable synthesis cases turn to  \emph{feasibility}, the corresponding decision version.
Rather than really computing a net $N$ with state graph isomorphic to a given TS, it is sufficient for feasibility to just decide if the target class contains $N$.
If this is NP-complete, than synthesis is an NP-hard problem with no obvious efficient solutions.

In NP-completeness proofs, we entirely detach ourselves from the notion of Petri nets.
In particular, we use the well known equality between feasibility and the conjunction of the state separation property (SSP) and the event state separation property (ESSP) \cite{BBD2015}, which are solely defined on input TSs.
ESSP is also known for its connection to language viability \cite{BBD2015}, meaning that a TS $A$ has the ESSP if and only if there is a $\tau$-net having the same transitional behavior but not necessarily the same states as $A$.
The presented polynomial time reductions translate the NP-complete cubic monotone one-in-three $3$-SAT problem \cite{MR2001} into the ESSP for each of the considered 84 boolean net classes.
Hence, deciding language viability is NP-complete in all these cases.
As we also make sure that given boolean expressions $\varphi$ are transformed to TSs $A(\varphi)$ where the ESSP relative to the considered class implies the SSP, we always show the NP-completeness of the ESSP and feasibility at the same time.
Instead of 84 individual proofs, we present a scheme that covers 77 cases by just six reductions following a common pattern.
The remaining seven classes are covered by one additional reduction that follows a different pattern.

While this paper ignores the 128 practically less relevant \nop-free classes, it does turn towards the complexity analysis of 36 of the remaining 44 types of nets allowing \nop.
For the 16 extensions of (\nop, \swap) with a subset of (\inp, \out, \used, \free) we sketch how a generalization of Schmitt's approach \cite{S1996} leads to a polynomial time synthesis algorithm.
For the other 20 classes we provide our own polynomial time synthesis algorithm.

Although we have to leave synthesis complexity open for eight \nop-afflicted classes, we nevertheless discuss some of their properties and the consequent difficulties in the conclusions.

For the sake of readability, we have moved all technical proofs to separate sections at the end of this paper.

%%%%%%%%%%%%%%%%%%%%%%%%%%%%%%%%%%%%%%%%%%%%%%%%%%%%%%%%%%%%%%%%%%%%%%%%%%%%%%%

\section{Preliminary Notions}\label{sec:preliminaries}

This section provides short formal definitions of all preliminary notions used in the paper.
For a detailed introduction into the field of Petri net synthesis, we propose the excellent monograph of Badouel, Bernardinello and Darondeau \cite{BBD2015}.
Here, a boolean Petri net $N = (P, T, M_0, f)$ is given by finite and disjoint sets $P$ of places and $T$ of transitions, an initial marking $M_0 \subseteq P$, and a flow function $f: P \times T \rightarrow I$ assigning an interaction $f(p,t)$ of 
$I = \{\nop, \inp, \out, \set, \res, \swap, \used, \free\}$
to every pair of place $p$ and transition $t$.
The interactions $i \in I$ are binary partial functions $i:\{0,1\} \rightarrow \{0,1\}$ as defined in the listing of Figure~\ref{fig:interactions}.
For readability, we group interactions by $\enter = \{\out,\set,\swap\}$, $\exit = \{\inp,\res,\swap\}$, $\keepo = \{\nop,\set,\used\}$, and $\keepze = \{\nop,\res,\free\}$.
\begin{figure}[b]\centering
\begin{tabular}{c|c|c|c|c|c|c|c|c}
$x$ & $\nop(x)$ & $\inp(x)$ & $\out(x)$ & $\set(x)$ & $\res(x)$ & $\swap(x)$ & $\used(x)$ & $\free(x)$\\ \hline
$0$ & $0$ & & $1$ & $1$ & $0$ & $1$ & & $0$\\
$1$ & $1$ & $0$ & & $1$ & $0$ & $0$ & $1$ & \\
\end{tabular}
\caption{
All interactions in $I$.
An empty cell means that the column's function is undefined on the respective $x$.
The entirely undefined function is missing in $I$.
}\label{fig:interactions}
\end{figure}

The meaning of a boolean net is to realize a certain behavior by cascades of firing transitions. 
In particular, a transition $t \in T$ can fire at a marking $M \subseteq P$ if interaction $f(p,t)$ is defined on $1$ for all $p \in M$ and $f(q,t)$ is defined on $0$ for all $q \in P \setminus M$.
By firing, $t$ produces the next marking $M' \subseteq P$ that exactly consists of all $p \in M$ with $f(p,t)(1) = 1$ and all $q \in P \setminus M$ with $f(q,t)(0) = 1$.
This is denoted by $M \edge{t} M'$.

Given a boolean net $N=(P, T, M_0, f)$, its behavior is captured by a finite automaton $A(N)$, called the transition system (TS, for short) of $N$.
The state set of $A(N)$ consists of all markings that, starting from initial state $M_0$, can be reached by a cascade of firing transitions.
For every reachable marking $M$ and transition $t \in T$ with $M \edge{t} M'$ the state transition function $\delta$ of $A$ is defined as $\delta(M,t) = M'$.

Subsets $\tau \subseteq I$ define \emph{types of nets}, subclasses of boolean nets limited to the respective interactions.
Hence, in a $\tau$-net, $f(p,t) \in \tau$ for all contained places $p$ and transitions $t$.
It is clear for $\tau \subseteq \tau' \subseteq I$ that the class of $\tau$-nets is a subset of the $\tau'$-nets.
Notice that $I$ contains all possible binary partial functions $\{0,1\} \rightarrow \{0,1\}$ except for the entirely undefined function. 
This to include would be futile as it makes incident transitions unable to ever fire.
Hence, $I$ is complete for deterministic nets and so is the family of its 256 subclasses.

Boolean net synthesis for a class $\tau$ is going backwards from input TS $A=(S, E, \delta, s_0)$ to the computation of a $\tau$-net $N$ with $A(N)$ isomorphic to $A$, if such a net exists.
In contrast to $A(N)$, the abstract states $S$ of $A$ miss any information about markings they stand for.
Accordingly, the events $E$ are an abstraction of $N$'s transitions $T$ as they relate to state changes only globally without giving the information about the local changes to places.
After all, the transition function $\delta: S\times E\rightarrow S$ still tells us how states are affected by events.

To prove net synthesis of $\tau$-nets NP-hard, we show the NP-completeness of the corresponding decision version: $\tau$-feasibility is the problem to decide the existence of a $\tau$-net $N$ with $A(N)$ isomorphic to the given TS $A$.
On that account, an input TS $A$ is considered as a directed labeled graph on nodes $S$ and transition arcs $s\edge{e}s'$ for every $\delta(s,e)=s'$. 
An event $e$ \emph{occurs} at a state $s$, denoted by $s\edge{e}$, if $\delta(s,e)$ is defined.

TSs in this paper are \emph{deterministic} by design as their state transition behavior is given by a function.
TSs are also required to make every state \emph{reachable} from $s_0$ by a directed path.
Aside from that, a TS $A$ can be \emph{simple}, which prohibits $s \edge{e} s'$ and $s \edge{e'} s'$ for $e\not=e' \in E$, \emph{loop-free}, banning $s \edge{e} s$ for all $s \in S$, and \emph{reduced}, eliminating unused events.
We say that $A$ is \emph{modest} if, beside the properties of determinism and reachability, $A$ is simple, loop-free, and reduced.
Although a TS does not have to be modest to be $\tau$-feasible, we make sure that our reductions produce modest TSs only, thus, showing the hardness of $\tau$-feasibility even for this input restriction.

To describe feasibility without referencing the sought net $N$, we subsequently introduce the \emph{state separation property} (\emph{SSP}, for short) and the \emph{event state separation property} (\emph{ESSP}, for short) for TSs, which in conjunction are equivalent to feasibility.
These notions require to follow the interpretation of \cite{BBD2015}, which sees a type of nets $\tau$ as a template TS $(S_\tau, E_\tau,\delta_\tau)$ for all synthesizable TSs of that class.
Leaving out an initial state, they define $S_\tau = \{0,1\}$, $E_\tau = \tau$ and $\delta_\tau(s, i) = i(s)$ for all $s \in S_\tau$ and all $i \in E_\tau$.
Based on this, a $\tau$-region of given $A=(S, E, \delta, s_0)$ is a pair $(sup, sig)$ of the \emph{support} $sup: S \rightarrow S_\tau = \{0,1\}$ and the \emph{signature} $sig: E\rightarrow E_\tau = \tau$ where every transition $s \edge{e} s'$ of $A$ leads to a transition $sup(s) \edge{sig(e)} sup(s')$ of $\tau$.
While a region divides $S$ into the two sets $sup^{-1}(b) = \{s \in S \mid sup(s) = b\}$ for $b \in \{0,1\}$, the events are cumulated by $sig^{-1}(i) = \{e \in E \mid sig(e) = i\}$ for all available interactions $i \in \tau$.
We also use $sig^{-1}(\tau') = \{e \in E \mid sig(e) \in \tau'\}$ for $\tau' \subseteq \tau$.

For a TS $A=(S, E, \delta, s_0)$ and a type of nets $\tau$, a pair of states $s \not= s' \in S$ is \emph{separable} for $\tau$ if there is a $\tau$-region $(sup, sig)$ such that $sup(s) \not= \sup(s')$.
Accordingly, $A$ has the SSP for $\tau$ if all pairs of distinct states from $A$ are separable. 
Secondly, an event $e \in E$ is called \emph{inhibitable} at a state $s \in S$ if there is a $\tau$-region $(sup, sig)$ where $sup(s) \edge{sig(e)}$ does not hold, that is, if the interaction $sig(e) \in \tau$ is not defined on input $sup(s) \in \{0,1\}$.
Then $A$ has the ESSP for $\tau$ if for all states $s \in S$ it is true that all events $e \in E$ that do not occur at $s$, meaning $\neg s \edge{e}$, are inhibitable at $s$.
It is well known from \cite{BBD2015} that a TS $A$ is $\tau$-feasible, that is, there exists a $\tau$-net $N$ with $A(N)$ isomorphic to $A$, if and only if $A$ has both, the SSP and the ESSP for $\tau$.
Moreover, \cite{BBD2015} also states that a TS $A$ has the ESSP for $\tau$ if and only if $A$ is \emph{$\tau$-language viable}.
This means, there a $\tau$-net $N$ with $A(N)$ language equivalent to $A$ where every event sequence in $E^*$ traverses a $s_0$-rooted directed path in $A(N)$ if and only if it does in $A$.

SSP and ESSP can also be seen as decision problems.
Moreover, an SSP atom for $\tau$ is to decide for given $(A, s, s')$ whether the states $s, s'$ of $A$ are $\tau$-separable.
Similarly, an ESSP atom for $\tau$ is to answer for $(A, e, s)$ if event $e$ is $\tau$-inhibitable at state $s$ of $A$ .

While being introduced to assist in proofs for hard synthesis, regions, SSP and ESSP are also construction tools for boolean nets.
In fact, having a region set $\R$ for $A$ that solves all its SSP and ESSP atoms with respect to some type of nets $\tau$, one can construct a $\tau$-net $N(A, \R) = (\R, E(A), f, M_0)$ on place set $\R$, transition set $E(A)$, flow function $f(R, e) = sig(e)$ for all $R = (sup, sig) \in \R$ and all $e \in E(A)$, and initial marking $M_0 = \{R = (sup, sig) \in \R \mid s_0 \in sup\}$.
By \cite{BBD2015}, $A$ is isomorphic to the state graph of $N(A, \R)$.
Hence, if we can efficiently compute $\R$ then $A$ is synthesizable in polynomial time.

Types of nets $\tau$ and $\tilde{\tau}$ have an isomorphism $\phi$ if $s \edge{i} s'$ is a transition in the template TS $\tau$ if and only if $\phi(s) \edge{\phi(i)} \phi(s')$ is one in the template TS $\tilde{\tau}$.
We benefit from isomorphisms mapping \nop\ to \nop, \swap\ to \swap, \inp\ to \out, \set\ to \res, \used\ to \free, and vice versa:
\begin{lemma}[Without proof]
\label{lem:isomorphic_types}
If $\tau$ and $\tilde{\tau}$ are isomorphic types of nets then a TS $A$ has the (E)SSP for $\tau$ if and only if $A$ has the (E)SSP for $\tilde{\tau}$.
\end{lemma}

%%%%%%%%%%%%%%%%%%%%%%%%%%%%%%%%%%%%%%%%%%%%%%%%%%%%%%%%%%%%%%%%%%%%%%%%%%%%%%%%%%%%%%%%%%%%%%%%

\section{A Reduction Scheme yields the NP-completeness of Feasibility for 77 Boolean Petri Net Classes}\label{sec:main_result}

This section presents our main result:
\begin{theorem}\label{the:main_result} 
Let $\tau_1 = \{\nop, \inp, \out\}$, $\tau_2 = \{\nop, \inp, \res, \swap\}$, $\tilde{\tau}_2 = \{\nop, \out, \set, \swap\}$, $\tau_3 = \{\nop, \inp, \set\}$, $\tilde{\tau}_3 = \{\nop, \out, \res\}$, $\tau_4 = \{\nop, \set, \swap\}$ and $\tilde{\tau}_4 = \{\nop, \res, \swap\}$.
Deciding $\tau$-feasibility as well as $\tau$-language viability for modest transition systems is NP-complete if
\begin{enumerate}
\item $\tau = \tau' \cup \omega$ for $\tau' \in \{\tau_1, \tau_2, \tilde{\tau}_2\}$ and $\omega \subseteq \{\used, \free\}$,
\item $\tau \supseteq \tau_3$ or $\tau \supseteq \tilde{\tau}_3$, or
\item $\tau = \tau' \cup \omega$ for $\tau' \in \{\tau_4, \tilde{\tau}_4, \tau_4\cup  \tilde{\tau}_4\}$ and non empty $\omega \subseteq \{\used, \free\}$. 
\end{enumerate}
\end{theorem}

The remainder of this section is devoted to the proof of the main theorem.
That input is restricted to modest TSs shows that the problem is intrinsically difficult and that the hardness is not hidden in special structures as loops and multi-edges.
In fact, the TSs generated by our reductions have other interesting properties leading beyond the scope of this paper when analyzed in detail.
For instance, they are planar graphs and have a low maximum degree, that is, every state is incident to at most four other states.

In total, Theorem~\ref{the:main_result} covers 77 classes.
The first condition hits four classes for every set $\tau_1$, $\tau_2$, and $\tilde{\tau}_2$.
The two cases of the second one describe 32 classes each, but they intersect in the eight supersets of $\{\nop, \inp, \out, \set, \res\}$.
Condition three brings nine classes.
All three conditions cover different classes.
Although this demands for 77 NP-completeness proofs, executing them individually does not teach us a lot about the problem structure.
On the one hand, it is straight forward that $\tau$-feasibility is a member of NP for all considered type of nets $\tau$ and we do not have to explicitly prove this here.
In a non-deterministic computation, one can simply guess and check in polynomial time for all pairs $s, s'$ of states, respectively for all required pairs $s,e$ of state and event, the region that separates $s$ and $s'$, respectively inhibits $e$ at $s$, or refuse the input if such a region does not exist.
A similar argumentation is used in \cite{BBD1997} to show the hardness of feasibility in NP for elementary net systems.

On the other hand, it is probably impossible to show hardness in NP for all considered classes $\tau$ at the same time.
Here, we manage to boil it down to six reductions that are all based on one scheme using the NP-complete cubic monotone one-in-three-$3$-SAT problem \cite{MR2001}.
Starting from the common construction principle, we can choose one of our six reductions by a turn-switch $\sigma$.
In every switch position $\sigma_1, \dots, \sigma_6$, the chosen reduction works for multiple interaction sets based on mutually shared interactions and isomorphisms.

Before we can set out the details of our concept, the following subsection introduces our way of easily generating and combining gadget TSs for our NP-completeness proofs.

\subsection{Unions of Transition Systems}\label{sec:unions}

If $A_0=(S_0,E_0,\delta_0,s_0^0), \dots ,A_n=(S_n,E_n,\delta_n,s_0^n)$ are TSs with pairwise disjoint states (but not necessarily disjoint events) we say that $U(A_0, \dots, A_n)$ is their \emph{union}.
By $S(U)$, we denote the entirety of all states in $A_0, \dots, A_n$ and $E(U)$ is the aggregation of all events.
For a flexible formalism, we allow to build unions recursively:
Firstly, we allow empty unions and identify every TS $A$ with the union containing only $A$, that is, $A = U(A)$.   
Next, if $U_1= U(A^1_0,\dots,A^1_{n_1}), \dots, U_m=(A^m_0,\dots,A^n_{n_m})$ are unions (possibly with $U_i =U()$ or $U_i=A_i$) then $U(U_1, \dots, U_m)$ is the flattened union $U(A^1_0, \dots, A^1_{n_1},\dots, A^m_0, \dots, A^n_{n_m})$.

We lift the concepts of regions, SSP, and ESSP to unions $U = U(A_0, \dots, A_n)$ as follows:
A $\tau$-region $(sup, sig)$ of $U$ consists of $sup: S(U) \rightarrow S_\tau$ and $sig: E(U) \rightarrow E_\tau$ such that, for all $i \in \{0, \dots, n\}$, the projections $sup_i(s) = sup(s), s \in S_i$ and $sig_i(e) = sig(e), e \in E_i$ provide a region $(sup_i, sig_i)$ of $A_i$.
Then, $U$ has the SSP for $\tau$ if for all different states $s, s' \in S(U)$ of the same TS $A_i$ there is a $\tau$-region $(sup,sig)$ of $U$ with $sup(s) \not= sup(s')$.
Moreover, $U$ has the ESSP for $\tau$ if for all events $e \in E(U)$ and all states $s \in S(U)$ where $s\edge{e}$ does not hold there is a $\tau$-region $(sup,sig)$ of $U$ where $sup(s) \edge{sig(e)}$ does not hold.
Naturally, $U$ is feasible for $\tau$ if it has both, the SSP and the ESSP for $\tau$.
In the same way, atoms of SSP and ESSP are translated to the state and event sets $S(U)$ and $E(U)$.

To merge a union $U = U(A_0, \dots, A_n)$ into a single TS, we define the joining $A_\tau(U)$, which depends on the type of nets $\tau$.
For our NP-completeness scheme, we require one basic construction $A(U)$ and an enhanced construction $A^+(U)$.
If $s^0_0, \dots, s^n_0$ are the initial states of $U$'s TSs then $A(U) = (S(U) \cup \bot, E(U) \cup \odot \cup \ominus, \delta, \bot_0)$ and $A^+(U) = (S(U) \cup \bot, E(U) \cup \odot \cup \ominus, \delta^+, \bot_0)$ are TSs with additional connector states $\bot=\{\bot_0, \dots, \bot_n\}$ and fresh events $\odot=\{\odot_0, \dots, \odot_n\}$, $\ominus=\{\ominus_1, \dots, \ominus_n\}$ joining the individual TSs of $U$ by
\[\delta(s,e) = 
\begin{cases}
s^i_0, & \text{if } s = \bot_i \text{ and } e=\odot_i,\\
\bot_{i+1}, & \text{if } s = \bot_i \text{ and } e=\ominus_{i+1},\\
\delta_i(s,e), & \text{if } s \in S_i \text{ and } e \in E_i,

\end{cases}
\text{\hfill}
\delta^+(s,e) = 
\begin{cases}
\bot_i, & \text{if } s = s^i_0 \text{ and } e=\odot_i,\\
\bot_i, & \text{if } s = \bot_{i+1} \text{ and } e=\ominus_{i+1}\\
\delta(s,e), & \text{otherwise.}
\end{cases}
\]
Hence, $A(U)$ puts the connector states into a chain of the events from $\ominus$ and links the initial states of TSs from $U$ to this chain using events from $\odot$.
The enhancement $A^+(U)$ is obtained from $A(U)$ by extending $\delta$ with additional reverse transitions.
Notice that $A(U)$ and $A^+(U)$ are modest if every TS of $U$ is modest.
The following lemma certifies the validity of the joining operation for the unions and the types of nets that occur in our reduction scheme.
\begin{lemma}\label{lem:union_validity}
Let $\tau$ be a type of nets and $U = U(A_0, \dots, A_n)$ be a union of TSs $A_0, \dots, A_n$ where, for every event $e$ in $E(U)$, there is at least one state $s$ in $S(U)$ with $\neg (s \edge{e})$.
Moreover, define the joining
\begin{enumerate}
\item $A_\tau(U) = A(U)$ if $\inp \in \tau$ and $\{\out, \set, \swap\}  \cap \tau \not= \emptyset$ and, otherwise,
\item $A_\tau(U) = A^+(U)$ if $\swap \in \tau$, $\{\used, \free\}  \cap \tau \not= \emptyset$ and for all $i \in \{0, \dots, n\}$ there is exactly one outgoing and one incoming arc at the initial state $s_0^i$ of the TS $A_i$, both labeled with the same event $u_i$ occurring at no further arcs in $U$.
\end{enumerate}
If $A_\tau(U)$ is defined then $U$ has the $\tau$-(E)SSP if and only if $A_\tau(U)$ has the $\tau$-(E)SSP.
\end{lemma}
\begin{proof}
\emph{If}:
Projecting a $\tau$-region separating $s$ and $s'$, respectively inhibiting $e$ at $s$, in $A_\tau(U)$ to the component TSs yields a $\tau$-region separating $s$ and $s'$, respectively inhibiting $e$ at $s$ in $U$.
Hence, the $\tau$-(E)SSP of $A_\tau(U)$ trivially implies the $\tau$-(E)SSP of $U$.

\emph{Only if}:
In the following, if $\tau$ has \inp\ and at least one of $\{\out,\set,\swap\}$ let $\textsf{exit}=\inp$ and \textsf{enter} be any of the available interactions from $\tau \cap \{\out,\set,\swap\}$.
Otherwise, if the second condition holds, define $\textsf{enter}= \textsf{exit}=\swap$ and let \textsf{test} be any interaction of $\tau \cap \{\used,\free\}$.
A $\tau$-region $R$ of $U$ separating $s$ and $s'$, respectively inhibiting $e$ at $s$, can be completed to become an equivalent $\tau$-region $R'$ of $A(U)$ by setting
\begin{align*}
sup_{R'}(s'') &= \begin{cases}
sup_R(s''), & \text{if } s'' \in S(U),\\
sup_R(s), & \text{otherwise, that is, } s'' \in \bot \text{ and}
\end{cases}\\
sig_{R'}(e') &= \begin{cases}
sig_R(e'), & \text{if } e' \in E(U),\\
\nop, & \text{if } e' = \ominus_j, j \in \{1, \dots, n\},\\
\nop, & \text{if } e' = \odot_i \text{ and } sup_R(s^i_0) = sup_R(s), i \in \{0, \dots, n\}\\
\textsf{enter} , & \text{if } e' = \odot_i \text{ and } sup_R(s^i_0) - sup_R(s) = 1, i \in \{0, \dots, n\},\\
\textsf{exit} , & \text{if } e' = \odot_i \text{ and } sup_R(s^i_0) - sup_R(s) = -1, i \in \{0, \dots, n\}.
\end{cases}
\end{align*}
Notice that a $\tau$-region $R'$ like this, which inherits the property of inhibiting $e$ at $s$ from $R$, do also inhibit $e$ at all connector states, since $sup_{R'}(s) = sup_{R'}(\bot_i)$.
This has the following consequence: As every event $e\in E(U)$ has at least one state $s\in S(U)$ with $\neg s\edge{e}$, the ESSP of $U$ implies that $U$ has at least one inhibiting region $R$ for every event $e$.
Hence, for every event $e$ we can use the respective region to create $R'$ as defined above, inhibiting $e$ at every connector state of the TS $A_\tau(U)$.

For the (E)SSP of $A_\tau(U)$ it is subsequently sufficient to analyze (event) state separation concerning just the connector states and events.
To separate the state $\bot_i$ from all the other states of $(S(U) \cup \bot)\setminus \{\bot_i\}$ we simply define the $\tau$-region $R_i$ where only $sup_{R_i}(\bot_i) = 1$ and where the signature of all events is \nop\ except for $\odot_i, \ominus_i, \ominus_{i+1}$. 
For these events (if they exist), we let $sig_{R_i}(\odot_i) = sig_{R_i}(\ominus_{i+1}) =\textsf{exit}$ and $sig_{R_i}(\ominus_i) = \textsf{enter}$.

Hence, taking $R_i$ over all $i \in \{0, \dots, m\}$ solves the remaining SSP atoms.

Moreover, for $A_\tau(U)=A(U)$, the regions $R_i, i \in \{0, \dots, n\}$ also inhibit $\odot_i$ and $\ominus_{i+1}$ at all states, which solves the rest of the ESSP atoms, too.

Hence, it remains to inhibit the $\odot$- and $\ominus$-events in $A^+(U)$.
For $i\in\{0,\dots, n\}$ the inhibition of $\odot_i$ at all relevant states of $A^+(U)$ can be done as follows: 
If $\textsf{test}=\used$ ($\textsf{test}=\free$), we define the region $R^+_i$ that includes (excludes) exactly the states $\bot_i,s^i_0$ and defines a \textsf{test} signature for $\odot_i$, a \swap\ signature for $u_i, \ominus_i, \ominus_{i+1}$ (if they exists) and a \nop\ signature for the remaining events of $A^+(U)$.
Similarly, for $i\in \{1,\dots, n\}$ we inhibit $\ominus_i$ in $A^+(U)$, by the region $R^+_i$ that includes (excludes) only the states $\bot_{i-1},\bot_i$ and defines a \textsf{test} signature for $\ominus_i$, a \swap\ signature for $\odot_i,\ominus_{i-1},\ominus_{i+1}$ (if they exist) and a \nop\ signature for the remaining events of $A^+(U)$.
\end{proof}

%%%%%%%%%%%%%%%%%%%%%%%%%%%%%%%%%%%%%%%%%%%%%%%%%%%%%%%%%%%%%%%%%%%%

\subsection{The General Reduction Scheme}\label{sec:general_scheme}

Our general scheme can be set up to a specific reduction by the turn switch $\sigma$.
In each of its six positions, $\sigma$ covers a whole collection of net classes.
Therefore, we simply understand the positions $\sigma_1, \dots, \sigma_6$ as the type sets managed by the respective reductions: 
\[\begin{array}{rcl@{\quad}rcl}
\sigma_1 &=& \{\tau_1 \cup \omega \mid \omega \subseteq \{\used, \free\} \} &
\sigma_2 &=& \{\tau_3 \cup \omega \mid \omega \subseteq \{\out, \res, \used, \free\}\\
\sigma_3 &=& \{\tau_2 \cup \omega \mid \omega \subseteq \{\used, \free\} \} &
\sigma_4 &=& \{\tau_3 \cup \{\swap\} \cup \omega \mid \omega \subseteq \{\out, \res, \used, \free\} \}\\
\sigma_5 &=& \{\tau_4 \cup \{\free\} \} &
\sigma_6 &=& \{\tau_4 \cup \{\used\} \cup \omega \mid \omega \subseteq \{\res, \free\} \}
\end{array}\]
The input to our scheme is the switch position $\sigma \in \{\sigma_1, \dots, \sigma_6\}$ and a cubic monotone boolean $3$-CNF $\varphi = \{\zeta_0, \dots, \zeta_{m-1}\}$, a set of negation-free $3$-clauses over the variables $V(\varphi)$ such that every variable is a member of exactly three clauses.
According to \cite{MR2001}, it is NP-complete to decide if $\varphi$ has a one-in-three model, that is, a subset $M \subseteq V(\varphi)$ of variables that hit every clause exactly once, which means $\vert M \cap \zeta_i\vert  = 1$ for all $i \in \{0, \dots, m-1\}$.
The result is a union $U^\sigma_\varphi$ of modest gadget TSs with the following properties:

\begin{enumerate}
\item
The variables $V(\varphi)$ are a subset of $E(U^\sigma_\varphi)$, the union events.
\item
There is a key state $s_{key} \in S(U^\sigma_\varphi)$ and a key event $k \in E(U^\sigma_\varphi)$ with $\neg s_{key} \edge{k}$.
\item\label{condition_three}
For every $\tau \in \sigma$, there is a $\tau$-region inhibiting $k$ at $s_{key}$ if and only if $\varphi$ has a one-in-three model $M$.
\item\label{condition_four}
For every $\tau \in \sigma$, the $\tau$-inhibitability of $k$ at $s_{key}$ implies that all ESSP atoms and all SSP atoms of $U^\sigma_\varphi$ are solvable.
\end{enumerate}

A polynomial time reduction scheme with these properties proves Theorem~\ref{the:main_result} as follows:
Condition \ref{condition_four} makes $\tau$-ESSP and $\tau$-feasibility the same problem for $U^\sigma_\varphi$.
Thus, feasibility is reduced to language viability and we subsequently concentrate on the NP-completeness proof for this problem.
In fact, by a one-in-three model for $\varphi$ Condition \ref{condition_three} makes $k$ inhibitable at the key state and Condition \ref{condition_four} leads to the ESSP, the SSP, and thus, feasibility, of $U^\sigma_\varphi$.
Reversely, a feasible $U^\sigma_\varphi$ has the ESSP by definition, which inhibits $k$ at the key state and, thus, leads to a one-in-three model of $\varphi$ by Condition \ref{condition_three}.
Lemma~\ref{lem:union_validity} transfers the whole argumentation to the joined TS $A(U^\sigma_\varphi)$ proving the NP-completeness of $\tau$-feasibility for all $\tau$ in the positions $\sigma_1, \dots, \sigma_6$.
Every remaining class $\tilde{\tau}$ of Theorem~\ref{the:main_result} is isomorphic to one of the already covered cases $\tau$ which makes $\tilde{\tau}$-feasibility NP-complete by Lemma \ref{lem:isomorphic_types}.

To present an example of our reduction, Figure \ref{fig:example} shows $A(U^{\sigma_4}_\varphi)$ for the $3$-CNF $\varphi=\{\zeta_0,\dots, \zeta_{5}\}$ built of the six clauses $\zeta_0=\{X_0,X_1,X_2\}$, $\zeta_1= \{X_2,X_0,X_3\}$, $\zeta_2= \{X_1,X_3,X_0\}$, $\zeta_3= \{X_2,X_4,X_5\}$, $\zeta_4=\{X_1,X_5,X_4\}$, and $\zeta_5= \{X_4,X_3,X_5\}$.
Ignoring the connector states and transitions, the figure also shows the complete union $U^{\sigma_4}_{\varphi}$ together with a $\tau$-region that inhibits the key event $k$ at the key state $h_{0,6}$ for all $\tau \in \sigma_4$.
The support $sup$ includes exactly the red emphasized states and the signature $sig$ is assumed to be defined in accordance to Lemmas~\ref{lem:indicator_regions} and~\ref{lem:key_union} from Sections~\ref{sec:translators} and~\ref{sec:key_unions}, respectively.
The example can be used to comprehend all the steps of the reduction scheme laid out in this section.
\newcommand{\nscale}[1]{\ensuremath{\escale{#1}}}
\newcommand{\escale}[1]{\ensuremath{\textbf{\scalebox{0.6}{#1}}}}
\newcommand{\transExample}[9]{
\ifstrequal{#8}{0}{
\begin{scope}[nodes={set=import nodes}, yshift=-#9 cm, ]% make all nodes part of this set
		\coordinate(0) at (0,0);
		\coordinate(1) at (5.5,0);
		\coordinate(2) at (11,0);
		\foreach \i in {0} {\fill[red!40, rounded corners] (\i) +(-0.5,-0.25) rectangle +(0.6,0.35);}
		\foreach \i in {1} {\fill[red!40, rounded corners] (\i) +(-0.5,-0.25) rectangle +(3.6,0.35);}
		\foreach \i in {2} {\fill[red!40, rounded corners] (\i) +(-0.5,-0.25) rectangle +(2.1,0.35);}
		\coordinate(3) at (1,1.5);
		\coordinate(4) at (6.5,1.5);
		\coordinate(5) at (12,1.5);
		\foreach \i in {3,4,5} {\fill[red!40, rounded corners] (\i) +(-0.5,-0.25) rectangle +(0.6,0.35);}
 	%Translatorpart T_{i,0}
		\node (00) at (1,1.5) {\nscale{$t_{#1,0,0}$}};
		\node (10) at (2.25,1) {\nscale{$t_{#1,0,1}$}};	
		\node (20) at (0,0) {\nscale{$t_{#1,0,2}$}};
		\node (30) at (1.5,0) {\nscale{$t_{#1,0,3}$}};
		\node (40) at (3,0) {\nscale{$t_{#1,0,4}$}};
		\node (50) at (4.5,0) {\nscale{$t_{#1,0,5}$}};
	
	%Translatorpart T_{i,1}
		\node (01) at (6.5,1.5) {\nscale{$t_{#1,1,0}$}};
		\node (11) at (7.75,1) {\nscale{$t_{#1,1,1}$}};	
		\node (21) at (5.5,0) {\nscale{$t_{#1,1,2}$}};
		\node (31) at (7,0) {\nscale{$t_{#1,1,3}$}};
		\node (41) at (8.5,0) {\nscale{$t_{#1,1,4}$}};
		\node (51) at (10,0) {\nscale{$t_{#1,1,5}$}};
		
	%Translatorpart T_{i,2}
		\node (02) at (12,1.5) {\nscale{$t_{#1,2,0}$}};
		\node (12) at (13.25,1) {\nscale{$t_{#1,2,1}$}};
		\node (22) at (11,0) {\nscale{$t_{#1,2,2}$}};
		\node (32) at (12.5,0) {\nscale{$t_{#1,2,3}$}};
		\node (42) at (14,0) {\nscale{$t_{#1,2,4}$}};
		\node (52) at (15.5,0) {\nscale{$t_{#1,2,5}$}};
		
\end{scope}
\graph {
	(import nodes);
		%Translatorpart T_{i,0}
			00 ->["$\escale{$k$}$"]10;
			10 ->[swap, bend right =20, "\pgfmathparse{int(3*#1)} \escale{$v_{\pgfmathresult}$}"]20;
			10 ->[bend left =20, "\pgfmathparse{int(3*#1)}  \escale{$w_{\pgfmathresult}$}"]50;
			20->[ bend left =20, "\escale{#2}"]30;
			30 ->[bend left =20, "\escale{#3}"]40;
			40 ->[bend left =20, "\escale{#4}"]50;
			30 ->[bend left =20, "\escale{#5}"]20;
			40 ->[bend left =20, "\escale{#6}"]30;
			50->[bend left =20,  "\escale{#7}"]40;
			
	%Translatorpart T_{i,1}
			01 ->["$\escale{$k$}$"]11;
			11 ->[swap, bend right =20, "\pgfmathparse{int(3*#1+1)}  \escale{$v_{\pgfmathresult}$}"]21;
			11 ->[bend left =20,  "\pgfmathparse{int(3*#1+1)}  \escale{$w_{\pgfmathresult}$}"]51;
			21->[ bend left =20, "\escale{#3}"]31;
			31 ->[bend left =20, "\escale{#4}"]41;
			41 ->[bend left =20, "\escale{#2}"]51;
			31 ->[bend left =20, "\escale{#6}"]21;
			41 ->[bend left =20, "\escale{#7}"]31;
			51->[bend left =20,  "\escale{#5}"]41;
		
		%Translatorpart T_{i,2}
			02 ->["$\escale{$k$}$"]12;
			12 ->[swap, bend right =20,  "\pgfmathparse{int(3*#1+2)}  \escale{$v_{\pgfmathresult}$}"]22;
			12 ->[bend left=20,  "\pgfmathparse{int(3*#1+2)}  \escale{$w_{\pgfmathresult}$}"]52;
			22->[bend left =20, "\escale{#4}"]32;
			32 ->[bend left =20, "\escale{#2}"]42;
			42 ->[bend left =20, "\escale{#3}"]52;
			32 ->[bend left =20,  "\escale{#7}"]22;
			42->[bend left =20,  "\escale{#5}"]32;
			52 ->[bend left =20,  "\escale{#6}"]42;

	};

}{
\ifstrequal{#8}{1}{
\begin{scope}[nodes={set=import nodes}, yshift=-#9 cm, ]% make all nodes part of this set
		\coordinate(0) at (0,0);
		\coordinate(1) at (5.5,0);
		\coordinate(2) at (11,0);
		\foreach \i in {0} {\fill[red!40, rounded corners] (\i) +(-0.5,-0.25) rectangle +(2.1,0.35);}
		\foreach \i in {1} {\fill[red!40, rounded corners] (\i) +(-0.5,-0.25) rectangle +(0.6,0.35);}
		\foreach \i in {2} {\fill[red!40, rounded corners] (\i) +(-0.5,-0.25) rectangle +(3.6,0.35);}
		\coordinate(3) at (1,1.5);
		\coordinate(4) at (6.5,1.5);
		\coordinate(5) at (12,1.5);
		\foreach \i in {3,4,5} {\fill[red!40, rounded corners] (\i) +(-0.5,-0.25) rectangle +(0.6,0.35);}
 	%Translatorpart T_{i,0}
		\node (00) at (1,1.5) {\nscale{$t_{#1,0,0}$}};
		\node (10) at (2.25,1) {\nscale{$t_{#1,0,1}$}};	
		\node (20) at (0,0) {\nscale{$t_{#1,0,2}$}};
		\node (30) at (1.5,0) {\nscale{$t_{#1,0,3}$}};
		\node (40) at (3,0) {\nscale{$t_{#1,0,4}$}};
		\node (50) at (4.5,0) {\nscale{$t_{#1,0,5}$}};
	
	%Translatorpart T_{i,1}
		\node (01) at (6.5,1.5) {\nscale{$t_{#1,1,0}$}};
		\node (11) at (7.75,1) {\nscale{$t_{#1,1,1}$}};	
		\node (21) at (5.5,0) {\nscale{$t_{#1,1,2}$}};
		\node (31) at (7,0) {\nscale{$t_{#1,1,3}$}};
		\node (41) at (8.5,0) {\nscale{$t_{#1,1,4}$}};
		\node (51) at (10,0) {\nscale{$t_{#1,1,5}$}};
		
	%Translatorpart T_{i,2}
		\node (02) at (12,1.5) {\nscale{$t_{#1,2,0}$}};
		\node (12) at (13.25,1) {\nscale{$t_{#1,2,1}$}};
		\node (22) at (11,0) {\nscale{$t_{#1,2,2}$}};
		\node (32) at (12.5,0) {\nscale{$t_{#1,2,3}$}};
		\node (42) at (14,0) {\nscale{$t_{#1,2,4}$}};
		\node (52) at (15.5,0) {\nscale{$t_{#1,2,5}$}};
		
\end{scope}
\graph {
	(import nodes);
		%Translatorpart T_{i,0}
			00 ->["$\escale{$k$}$"]10;
			10 ->[swap, bend right =20, "\pgfmathparse{int(3*#1)} \escale{$v_{\pgfmathresult}$}"]20;
			10 ->[bend left =20, "\pgfmathparse{int(3*#1)}  \escale{$w_{\pgfmathresult}$}"]50;
			20->[ bend left =20, "\escale{#2}"]30;
			30 ->[bend left =20, "\escale{#3}"]40;
			40 ->[bend left =20, "\escale{#4}"]50;
			30 ->[bend left =20, "\escale{#5}"]20;
			40 ->[bend left =20, "\escale{#6}"]30;
			50->[bend left =20,  "\escale{#7}"]40;
			
	%Translatorpart T_{i,1}
			01 ->["$\escale{$k$}$"]11;
			11 ->[swap, bend right =20, "\pgfmathparse{int(3*#1+1)}  \escale{$v_{\pgfmathresult}$}"]21;
			11 ->[bend left =20,  "\pgfmathparse{int(3*#1+1)}  \escale{$w_{\pgfmathresult}$}"]51;
			21->[ bend left =20, "\escale{#3}"]31;
			31 ->[bend left =20, "\escale{#4}"]41;
			41 ->[bend left =20, "\escale{#2}"]51;
			31 ->[bend left =20, "\escale{#6}"]21;
			41 ->[bend left =20, "\escale{#7}"]31;
			51->[bend left =20,  "\escale{#5}"]41;
		
		%Translatorpart T_{i,2}
			02 ->["$\escale{$k$}$"]12;
			12 ->[swap, bend right =20,  "\pgfmathparse{int(3*#1+2)}  \escale{$v_{\pgfmathresult}$}"]22;
			12 ->[bend left=20,  "\pgfmathparse{int(3*#1+2)}  \escale{$w_{\pgfmathresult}$}"]52;
			22->[bend left =20, "\escale{#4}"]32;
			32 ->[bend left =20, "\escale{#2}"]42;
			42 ->[bend left =20, "\escale{#3}"]52;
			32 ->[bend left =20,  "\escale{#7}"]22;
			42->[bend left =20,  "\escale{#5}"]32;
			52 ->[bend left =20,  "\escale{#6}"]42;

	};
}{
\begin{scope}[nodes={set=import nodes}, yshift=-#9 cm, ]% make all nodes part of this set

		\coordinate(0) at (0,0);
		\coordinate(1) at (5.5,0);
		\coordinate(2) at (11,0);
		\foreach \i in {0} {\fill[red!40, rounded corners] (\i) +(-0.5,-0.25) rectangle +(3.6,0.35);}
		\foreach \i in {1} {\fill[red!40, rounded corners] (\i) +(-0.5,-0.25) rectangle +(2.1,0.35);}
		\foreach \i in {2} {\fill[red!40, rounded corners] (\i) +(-0.5,-0.25) rectangle +(0.6,0.35);}
		\coordinate(3) at (1,1.5);
		\coordinate(4) at (6.5,1.5);
		\coordinate(5) at (12,1.5);
		\foreach \i in {3,4,5} {\fill[red!40, rounded corners] (\i) +(-0.5,-0.25) rectangle +(0.6,0.35);}
 	%Translatorpart T_{i,0}
		\node (00) at (1,1.5) {\nscale{$t_{#1,0,0}$}};
		\node (10) at (2.25,1) {\nscale{$t_{#1,0,1}$}};	
		\node (20) at (0,0) {\nscale{$t_{#1,0,2}$}};
		\node (30) at (1.5,0) {\nscale{$t_{#1,0,3}$}};
		\node (40) at (3,0) {\nscale{$t_{#1,0,4}$}};
		\node (50) at (4.5,0) {\nscale{$t_{#1,0,5}$}};
	
	%Translatorpart T_{i,1}
		\node (01) at (6.5,1.5) {\nscale{$t_{#1,1,0}$}};
		\node (11) at (7.75,1) {\nscale{$t_{#1,1,1}$}};	
		\node (21) at (5.5,0) {\nscale{$t_{#1,1,2}$}};
		\node (31) at (7,0) {\nscale{$t_{#1,1,3}$}};
		\node (41) at (8.5,0) {\nscale{$t_{#1,1,4}$}};
		\node (51) at (10,0) {\nscale{$t_{#1,1,5}$}};
		
	%Translatorpart T_{i,2}
		\node (02) at (12,1.5) {\nscale{$t_{#1,2,0}$}};
		\node (12) at (13.25,1) {\nscale{$t_{#1,2,1}$}};
		\node (22) at (11,0) {\nscale{$t_{#1,2,2}$}};
		\node (32) at (12.5,0) {\nscale{$t_{#1,2,3}$}};
		\node (42) at (14,0) {\nscale{$t_{#1,2,4}$}};
		\node (52) at (15.5,0) {\nscale{$t_{#1,2,5}$}};
		
\end{scope}
\graph {
	(import nodes);
		%Translatorpart T_{i,0}
			00 ->["$\escale{$k$}$"]10;
			10 ->[swap, bend right =20, "\pgfmathparse{int(3*#1)} \escale{$v_{\pgfmathresult}$}"]20;
			10 ->[bend left =20, "\pgfmathparse{int(3*#1)}  \escale{$w_{\pgfmathresult}$}"]50;
			20->[ bend left =20, "\escale{#2}"]30;
			30 ->[bend left =20, "\escale{#3}"]40;
			40 ->[bend left =20, "\escale{#4}"]50;
			30 ->[bend left =20, "\escale{#5}"]20;
			40 ->[bend left =20, "\escale{#6}"]30;
			50->[bend left =20,  "\escale{#7}"]40;
			
	%Translatorpart T_{i,1}
			01 ->["$\escale{$k$}$"]11;
			11 ->[swap, bend right =20, "\pgfmathparse{int(3*#1+1)}  \escale{$v_{\pgfmathresult}$}"]21;
			11 ->[bend left =20,  "\pgfmathparse{int(3*#1+1)}  \escale{$w_{\pgfmathresult}$}"]51;
			21->[ bend left =20, "\escale{#3}"]31;
			31 ->[bend left =20, "\escale{#4}"]41;
			41 ->[bend left =20, "\escale{#2}"]51;
			31 ->[bend left =20, "\escale{#6}"]21;
			41 ->[bend left =20, "\escale{#7}"]31;
			51->[bend left =20,  "\escale{#5}"]41;
		
		%Translatorpart T_{i,2}
			02 ->["$\escale{$k$}$"]12;
			12 ->[swap, bend right =20,  "\pgfmathparse{int(3*#1+2)}  \escale{$v_{\pgfmathresult}$}"]22;
			12 ->[bend left=20,  "\pgfmathparse{int(3*#1+2)}  \escale{$w_{\pgfmathresult}$}"]52;
			22->[bend left =20, "\escale{#4}"]32;
			32 ->[bend left =20, "\escale{#2}"]42;
			42 ->[bend left =20, "\escale{#3}"]52;
			32 ->[bend left =20,  "\escale{#7}"]22;
			42->[bend left =20,  "\escale{#5}"]32;
			52 ->[bend left =20,  "\escale{#6}"]42;		
	};
}}
}
\newcommand{\generator}[6]{
\begin{scope}[nodes={set=import nodes}, yshift=#5 cm, xshift=#6 cm]
 	%generator template
		
		\coordinate (1) at (0,0) ;
		\foreach \i in {1} {\fill[red!40, rounded corners] (\i) +(-0.4,-0.3) rectangle +(1.8,0.4);}
		\node (00) at (0,0) {\nscale{$g^{#1,#2}_{#3,0}$}};
		\node (10) at (1.5,0) {\nscale{$g^{#1,#2}_{#3,1}$}};	
		\node (20) at (0,-1) {\nscale{$g^{#1,#2}_{#3,2}$}};
		\node (30) at (1.5,-1) {\nscale{$g^{#1,#2}_{#3,3}$}};
		\node (dots) at (3,-0.6) {$\dots$};
		\coordinate (2) at (4.5cm,0) ;
		\foreach \i in {2} {\fill[red!40, rounded corners] (\i) +(-0.4,-0.3) rectangle +(1.9,0.4);}
		\node (01) at (4.5,0) {\nscale{$g^{#1,#2}_{#4,0}$}};
		\node (11) at (6,0) {\nscale{$g^{#1,#2}_{#4,1}$}};	
		\node (21) at (4.5,-1) {\nscale{$g^{#1,#2}_{#4,2}$}};
		\node (31) at (6,-1) {\nscale{$g^{#1,#2}_{#4,3}$}};
\end{scope}
\graph {
	(import nodes);
			00 ->["$\escale{$#1_{#3}$}$"]10;
			00 ->[swap, "$\escale{$k$}$"]20;
			20 ->[swap,"$\escale{$#2_{#3}$}$"]30;
			10 ->["$\escale{$k$}$"]30;
			01 ->["$\escale{$#1_{#4}$}$"]11;
			01 ->[swap, "$\escale{$k$}$"]21;
			21 ->[swap,"$\escale{$#2_{#4}$}$"]31;
			11 ->["$\escale{$k$}$"]31;
			};

}
\def\myFigureScale{0.75}
\begin{figure}[b!]
\centering
\begin{tikzpicture}[new set = import nodes, scale=\myFigureScale, every node/.style={scale=\myFigureScale}]

\transExample{0}{$X_0$}{$X_1$}{$X_2$}{$x_0$}{$x_1$}{$x_2$}{0}{0}
%brace T_0
\draw [decorate, decoration={brace, amplitude=3pt}]
 (-0.75,-0.5)-- (-0.75,2) node [midway, left,   xshift=-0.25cm] {$T^{\sigma_4}_0$};

\node (b_0) at (1,2.5) {\nscale{$\bot_0$}};
\draw[->] (b_0)--(00)node [pos=0.4, left ] {\escale{$\odot_0$}};
\node (b_1) at (6.5,2.5) {\nscale{$\bot_1$}};
\draw[->] (b_1)--(01)node [pos=0.4, left ] {\escale{$\odot_1$}};
\node (b_2) at (12,2.5) {\nscale{$\bot_2$}};
\draw[->] (b_2)--(02)node [pos=0.4, left ] {\escale{$\odot_2$}};
\draw[->] (b_0)--(b_1)node [pos=0.4, above] {\escale{$\ominus_1$}};
\draw[->] (b_1)--(b_2)node [pos=0.4, above ] {\escale{$\ominus_2$}};
\coordinate (x_0) at (16,2.5) {};
\coordinate (y_0) at (16,-1) {};
\draw (b_2)--(x_0)--(y_0)node [pos=0.4, right ] {\escale{$\ominus_3$}};

%%%%%%%%%%%%%%%%%%%%%%%%%%%%%%%%%%%
\transExample{1}{$X_2$}{$X_0$}{$X_3$}{$x_2$}{$x_0$}{$x_3$}{1}{3.5}
%brace T_1
\draw [decorate, decoration={brace, amplitude=3pt}]
 (16,-1.5) --(16,-4) node [midway, right,   xshift=0.25cm] {$T^{\sigma_4}_1$};
\node (b_3) at (12,-1) {\nscale{$\bot_3$}};
\draw[->] (b_3)--(02)node [pos=0.4, left ] {\escale{$\odot_3$}};
\node (b_4) at (6.5,-1) {\nscale{$\bot_4$}};
\draw[->] (b_4)--(01)node [pos=0.4, left ] {\escale{$\odot_4$}};
\node (b_5) at (1,-1) {\nscale{$\bot_5$}};
\draw[->] (b_5)--(00)node [pos=0.4, left ] {\escale{$\odot_5$}};
\draw[->] (y_0)--(b_3);
\draw[->] (b_3)--(b_4)node [pos=0.4, below ] {\escale{$\ominus_4$}};
\draw[->] (b_4)--(b_5)node [pos=0.4, below  ] {\escale{$\ominus_5$}};
\coordinate (x_1) at (-1,-1) ;
\coordinate (y_1) at (-1,-4.5);
\draw (b_5)--(x_1)--(y_1)node [pos=0.4, left ] {\escale{$\ominus_6$}};

%%%%%%%%%%%%%%%%%%%%%%%%%%%%%%%%%%%
\transExample{2}{$X_1$}{$X_3$}{$X_0$}{$x_1$}{$x_3$}{$x_0$}{2}{7}
%brace T_2
\draw [decorate, decoration={brace, amplitude=3pt}]
 (-0.75,-7.5)--(-0.75,-5) node [midway, left,   xshift=-0.25cm] {$T^{\sigma_4}_2$};
\node (b_8) at (12,-4.5) {\nscale{$\bot_8$}};
\draw[->] (b_8)--(02)node [pos=0.4, left ] {\escale{$\odot_8$}};
\node (b_7) at (6.5,-4.5) {\nscale{$\bot_7$}};
\draw[->] (b_7)--(01)node [pos=0.4, left ] {\escale{$\odot_7$}};
\node (b_6) at (1,-4.5) {\nscale{$\bot_6$}};
\draw[->] (b_6)--(00)node [pos=0.4, left ] {\escale{$\odot_6$}};
\draw[->] (y_1)--(b_6);
\draw[->] (b_6)--(b_7)node [pos=0.4, below  ] {\escale{$\ominus_7$}};
\draw[->] (b_7)--(b_8)node [pos=0.4, below  ] {\escale{$\ominus_8$}};
\coordinate (x_2) at (16,-4.5) {};
\coordinate (y_2) at (16,-8) {};
\draw (b_8)--(x_2)--(y_2)node [pos=0.4, right ] {\escale{$\ominus_9$}};
%%%%%%%%%%%%%%%%%%%%%%%%%%%%%%%%%%%%%%
\transExample{3}{$X_2$}{$X_4$}{$X_5$}{$x_2$}{$x_4$}{$x_5$}{1}{10.5}
%brace T_3
\draw [decorate, decoration={brace, amplitude=3pt}]
 (16,-8.5) --(16,-11) node [midway, right,   xshift=0.25cm] {$T^{\sigma_4}_3$};
\node (b_9) at (12,-8) {\nscale{$\bot_9$}};
\draw[->] (b_9)--(02)node [pos=0.4, left ] {\escale{$\odot_9$}};
\node (b_10) at (6.5,-8) {\nscale{$\bot_{10}$}};
\draw[->] (b_10)--(01)node [pos=0.4, left ] {\escale{$\odot_{10}$}};
\node (b_11) at (1,-8) {\nscale{$\bot_{11}$}};
\draw[->] (b_11)--(00)node [pos=0.4, left ] {\escale{$\odot_{11}$}};
\draw[->] (y_2)--(b_9);
\draw[->] (b_9)--(b_10)node [pos=0.4, below  ] {\escale{$\ominus_{10}$}};
\draw[->] (b_10)--(b_11)node [pos=0.4, below  ] {\escale{$\ominus_{11}$}};
\coordinate (x_3) at (-1,-8) ;
\coordinate (y_3) at (-1,-11.5);
\draw (b_11)--(x_3)--(y_3)node [pos=0.4, left ] {\escale{$\ominus_{12}$}};

%%%%%%%%%%%%%%%%%%%%%%%%%%%%%%%%%%%%%%
\transExample{4}{$X_1$}{$X_5$}{$X_4$}{$x_1$}{$x_5$}{$x_4$}{2}{14}
%brace T_4
\draw [decorate, decoration={brace, amplitude=3pt}]
(-0.75,-14.5)-- (-0.75,-12) node [midway, left,   xshift=-0.25cm] {$T^{\sigma_4}_4$};
\node (b_14) at (12,-11.5) {\nscale{$\bot_{14}$}};					
\draw[->] (b_14)--(02)node [pos=0.4, left ] {\escale{$\odot_{14}$}};		
\node (b_13) at (6.5,-11.5) {\nscale{$\bot_{13}$}};					
\draw[->] (b_13)--(01)node [pos=0.4, left ] {\escale{$\odot_{13}$}};		
\node (b_12) at (1,-11.5) {\nscale{$\bot_{12}$}};						
\draw[->] (b_12)--(00)node [pos=0.4, left ] {\escale{$\odot_{12}$}};		
\draw[->] (y_3)--(b_12);										
\draw[->] (b_12)--(b_13)node [pos=0.4, below  ] {\escale{$\ominus_{13}$}};
\draw[->] (b_13)--(b_14)node [pos=0.4, below ] {\escale{$\ominus_{14}$}};
\coordinate (x_4) at (16,-11.5) {};
\coordinate (y_4) at (16,-15) {};
\draw (b_14)--(x_4)--(y_4)node [pos=0.4, right ] {\escale{$\ominus_{15}$}};

%%%%%%%%%%%%%%%%%%%%%%%%%%%%%%%%%%%%%%
\transExample{5}{$X_4$}{$X_3$}{$X_5$}{$x_4$}{$x_3$}{$x_5$}{0}{17.5}
%brace T_5
\draw [decorate, decoration={brace, amplitude=3pt}]
 (16,-15.5) --(16,-18) node [midway, right,   xshift=0.25cm] {$T^{\sigma_4}_5$};
\node (b_15) at (12,-15) {\nscale{$\bot_{15}$}};
\draw[->] (b_15)--(02)node [pos=0.4, left ] {\escale{$\odot_{15}$}};
\node (b_16) at (6.5,-15) {\nscale{$\bot_{16}$}};
\draw[->] (b_16)--(01)node [pos=0.4, left ] {\escale{$\odot_{16}$}};
\node (b_17) at (1,-15) {\nscale{$\bot_{17}$}};
\draw[->] (b_17)--(00)node [pos=0.4, left ] {\escale{$\odot_{17}$}};
\draw[->] (y_4)--(b_15);
\draw[->] (b_15)--(b_16)node [pos=0.4, below ] {\escale{$\ominus_{16}$}};
\draw[->] (b_16)--(b_17)node [pos=0.4, below ] {\escale{$\ominus_{17}$}};
\coordinate (x_5) at (-1,-15) ;
\coordinate (y_5) at (-1,-18);
\node(dots) at (-1,-18.5){$\vdots$};
\draw (b_17)--(x_5)--(y_5)node [pos=0.4, left ] {\escale{$\ominus_{18}$}};
\end{tikzpicture}
\caption{
The result TS $A(U^{\sigma_4}_\varphi)$ of our reduction under switch position $\sigma_4$ for input $\varphi=\{\zeta_0,\dots, \zeta_{5}\}$.
The red marked states define the support $sup$ of a key region $(sup,sig)$.
}
\end{figure}
%%%%%%%%%%%%%%%%%%%%%%%%%%%%%%%%%%%%%%%%%%%%%%%%%%%%%%%%%%%%%%%%%%%%%%%%%%%%%%%
\begin{figure}[t!]\ContinuedFloat
\centering
\begin{tikzpicture}[new set = import nodes, scale=\myFigureScale, every node/.style={scale=\myFigureScale}]

\generator{u}{q}{0}{17}{0}{0}
%brace F_q, F_y
\draw [decorate, decoration={brace, amplitude=3pt}]
(-0.75,-1)-- (-0.75,1) node [midway, above,  xshift=-0.25cm, rotate=90] {\scalebox{1.2}{$\substack{G^{\_,q}_0,\dots, G^{\_,q}_{17} \\ G^{\_,y}_0,\dots, G^{\_,y}_{17}}$}};
\node (18) at (0,3) {\nscale{$\vdots$}};
\node (b_18) at (0,1.5) {\nscale{$\bot_{18}$}};
\draw[->] (18)--(b_18)node [pos=0.4, left ] {\escale{$\ominus_{18}$}};
\draw[->] (b_18)--(00)node [pos=0.4, right] {\escale{$\odot_{18}$}};
\node (b_35) at (4.5,1.5) {\nscale{$\bot_{35}$}};
\draw[->] (b_35)--(01)node [pos=0.4, right] {\escale{$\odot_{35}$}};
\coordinate (x_0) at (1.75,1.5) ;
\node (dots) at (2.5,1.5){\dots};
\coordinate (y_0) at (3,1.5);
%\draw[->] (y_4)--(b_15);
\draw[->] (b_18)--(x_0)node [pos=0.4, above ] {\escale{$\ominus_{19}$}};
\draw[->] (y_0)--(b_35)node [pos=0.4, above ] {\escale{$\ominus_{35}$}};

%%%%%%%%%%%%%%%%%%%%%%%%%%%%%%%
\generator{u'}{y}{0}{17}{0}{7.5}
\node (b_36) at (7.5,1.5) {\nscale{$\bot_{36}$}};
\draw[->] (b_36)--(00)node [pos=0.4, right ] {\escale{$\odot_{36}$}};
\node (b_53) at (12,1.5) {\nscale{$\bot_{53}$}};
\draw[->] (b_53)--(01)node [pos=0.4, right ] {\escale{$\odot_{53}$}};
\coordinate (x_1) at (9,1.5) ;
\node (dots_2) at (10,1.5){\dots};
\coordinate (y_1) at (10.5,1.5);
\draw[->] (b_35)--(b_36)node [pos=0.4, above ] {\escale{$\ominus_{36}$}};
\draw[->] (b_36)--(x_1)node [pos=0.4, above ] {\escale{$\ominus_{37}$}};
\draw[->] (y_1)--(b_53)node [pos=0.4, above ] {\escale{$\ominus_{53}$}};

\coordinate (x_2) at (14.25,1.5) ;
\coordinate (y_2) at (14.25,-2) ;
\draw(b_53)--(x_2) node [pos=0.4, above ] {};
\draw (x_2)--(y_2) node[pos=0.5, right] {\escale{$\ominus_{54}$}};
%%%%%%%%%%%%%%%%%%%%%%%%%%%%%%%%%%%%%

\generator{c}{c}{0}{34}{-3.5}{7.5}

\draw [decorate, decoration={brace, amplitude=3pt}]
(14.25,-2.5)-- (14.25,-4.5) node [midway, below,  xshift=0.25cm, rotate=90] {\scalebox{1.2}{$\substack{G^{x,\_}_0,\dots, G^{x,\_}_{5} \\ G^{c,c}_0,\dots, G^{c,c}_{34}}$}};
\node (b_88) at (7.5,-2) {\nscale{$\bot_{88}$}};
\draw[->] (b_88)--(00)node [pos=0.4, right] {\escale{$\odot_{88}$}};
\node (b_54) at (12,-2) {\nscale{$\bot_{54}$}};
\draw[->] (b_54)--(01)node [pos=0.4, right] {\escale{$\odot_{54}$}};
\draw[->] (y_2)--(b_54) node[pos=0.4, above] {};
\coordinate (x_3) at (9,-2) ;
\node (dots_2) at (9.8,-2){\dots};
\coordinate (y_3) at (10.5,-2);
\draw[->](b_54)--(y_3) node [pos=0.4, above ] {\escale{$\ominus_{55}$}};
\draw[->] (x_3)--(b_88) node[pos=0.4, above] {\escale{$\ominus_{88}$}};

%%%%%%%%%%%%%%%%%%%%%%%%%%%%%%%%%%%%%%
\generator{x}{u''}{0}{5}{-3.5}{0}

\node (b_94) at (0,-2) {\nscale{$\bot_{94}$}};
\draw[->] (b_94)--(00)node [pos=0.4, right] {\escale{$\odot_{94}$}};
\node (b_89) at (4.5,-2) {\nscale{$\bot_{89}$}};
\draw[->] (b_89)--(01)node [pos=0.4, right ] {\escale{$\odot_{89}$}};
\coordinate (x_4) at (1.75,-2) ;
\node (dots) at (2.5,-2){\dots};
\coordinate (y_4) at (3,-2);
\draw[->] (b_88)--(b_89) node[pos=0.4, above]{\escale{$\ominus_{89}$}};
\draw[->] (x_4)--(b_94)node [pos=0.4, above ] {\escale{$\ominus_{94}$}};
\draw[->] (b_89)--(y_4)node [pos=0.4, above ] {\escale{$\ominus_{90}$}};

\coordinate (x_5) at (-0.75,-2) ;
\coordinate(y_5) at (-0.75,-5.25);
\draw (b_94)--(x_5)--(y_5)node [pos=0.9, left] {\escale{$\ominus_{95}$}};
%
%
%%%%%%%%%%%%%%%%%%%%%%%%%%%%%%%%%%%%%%%
\begin{scope}[nodes={set=import nodes},yshift=-6.5cm, ]% make all nodes part of this set
 	%freezer F^z_0 zeros 1 for sigma_1...sigma_4
	\foreach \i in {0,...,4} {\coordinate (\i) at (\i*1.5cm,0);}
	\foreach \i in {0} {\fill[red!40, rounded corners] (\i) +(-0.5,-0.25) rectangle +(0.5,0.35);}
	\foreach \i in {3} {\fill[red!40, rounded corners] (\i) +(-2,-0.25) rectangle +(0.5,0.35);}
	\foreach \i in {0,...,4} {\node (n\i) at (\i) {\nscale{$f_{0,\i}$}};}
		
\end{scope}
\graph {
	(import nodes);
			n0 ->["$\escale{$k$}$"]n1;
			n1 ->["$\escale{$n_0$}$"]n2;
			n2 ->["$\escale{$z_0$}$"]n3;
			n3 ->["$\escale{$k$}$"]n4;
			};

\begin{scope}[nodes={set=import nodes}, yshift=-6.5cm,xshift=7cm]
 	%generator template
		
		\coordinate (1) at (0,0) ;
		\foreach \i in {1} {\fill[red!40, rounded corners] (\i) +(-0.35,-0.35) rectangle +(1.85,0.4);}
		\node (00) at (0,0) {\nscale{$g^{n,u'''}_{0,0}$}};
		\node (10) at (1.5,0) {\nscale{$g^{n,u'''}_{0,1}$}};	
		\node (20) at (0,-1) {\nscale{$g^{n,u'''}_{0,2}$}};
		\node (30) at (1.5,-1) {\nscale{$g^{n,u'''}_{0,3}$}};
		
\end{scope}
\graph {
	(import nodes);
			00 ->["$\escale{$n_{0}$}$"]10;
			00 ->[swap, "$\escale{$k$}$"]20;
			20 ->[swap,"$\escale{$u'''_{0}$}$"]30;
			10 ->["$\escale{$k$}$"]30;
			
			};
%brace F_0, G^{n,\_}_0%%%%%%%%%%%%%%%%%%%%%%%%%%%%%%%%%%%%%
\draw [decorate, decoration={brace, amplitude=3pt}]
(-0.75,-7.5)-- (-0.75,-6) node [midway, above,  xshift=-0.25cm, rotate=90] {$F_0,G^{n,\_}_0$};
\node(b_95) at (0,-5.25){\nscale{$\bot_{95}$}};
\draw[->] (y_5)--(b_95);
\draw[->] (b_95)--(n0)node[pos=0.4, left]{\escale{$\odot_{95}$}};
\node(b_96) at (7,-5.25){\nscale{$\bot_{96}$}};
\draw[->] (b_96)--(00)node[pos=0.4, left]{\escale{$\odot_{96}$}};
\draw[->] (b_95)--(b_96)node[pos=0.4, above]{\escale{$\ominus_{96}$}};
\coordinate (x_6) at (9.5,-5.25) ;
\coordinate(y_6) at (9.5,-8.25);

%%%%%%%%%%%%%%%%%%%%%%%%%%%%%%%%%%%%%%%%%%%%%

\begin{scope}[nodes={set=import nodes}, yshift=-9.5cm]
%headmaster 1 for sigma_1...sigma_3
%States H_0 
\foreach \i in {0,...,6} {\coordinate (h0\i) at (\i*2cm,0);}
%States H_{3m-1}
\foreach \i in {0,...,6} {\coordinate (h3m_1\i) at (\i*2cm,-3);}
%States H_{3m}
\foreach \i in {0,...,6} {\coordinate (h3m\i) at (\i*2cm,-4.5);}
%States H_{6m-1}
\foreach \i in {0,...,6} {\coordinate (h6m_1\i) at (\i*2cm,-7.5);}
\foreach \i in {h00,h03,h3m_10,h3m_13,h3m0,h3m3,h6m_10,h6m_13} {\fill[red!40, rounded corners] (\i) +(-0.5,-0.25) rectangle +(0.6,0.25);}
%using the defined state coordinates
\foreach \i in {0,...,6} {\node (h0\i) at (h0\i) {\nscale{$h_{0,\i}$}};}
\foreach \i in {0,...,6} {\node (h3m_1\i) at (h3m_1\i) {\nscale{$h_{17,\i}$}};}
\foreach \i in {0,...,6} {\node (h3m\i) at (h3m\i) {\nscale{$h_{18,\i}$}};}
\foreach \i in {0,...,6} {\node (h6m_1\i) at (h6m_1\i) {\nscale{$h_{35,\i}$}};}

%REACHABILITY STUFF
\coordinate (d0) at(0,-1.5);
\coordinate (d1) at(12cm,-1.5cm);
\coordinate (r0_0) at(11,-0.75);
\coordinate (r0_1) at(6cm,-0.75cm);
\coordinate (r0_2) at(0.5cm,-0.75cm);
\node (d0) at (d0) {\nscale{$\vdots$}};
\node (d1) at (d1) {\nscale{$\vdots$}};
\draw[->, dashed, rounded corners] (h06) -- (r0_0)--(r0_1)--node[above] { \escale{$r_{0}$}}(r0_2)-- (d0);
\coordinate (r3m_2_0) at(11cm,-2.25cm);
\coordinate (r3m_2_1) at(6cm,-2.25cm);
\coordinate (r3m_2_2) at(0.5cm,-2.25cm);
\draw[->, dashed, rounded corners] (d1) -- (r3m_2_0)--(r3m_2_1)--node[above] { \escale{$r_{16}$}}(r3m_2_2)-- (h3m_10);
\coordinate (r3m_1_0) at(11cm,-3.75cm);
\coordinate (r3m_1_1) at(6cm,-3.75cm);
\coordinate (r3m_1_2) at(0.5cm,-3.75cm);
\draw[->, dashed, rounded corners] (h3m_16) -- (r3m_1_0)--(r3m_1_1)--node[above] { \escale{$r_{17}$}}(r3m_1_2)-- (h3m0);
\coordinate (d2) at(0,-6);
\coordinate (d3) at(12cm,-6cm);
\coordinate (r3m_0) at(11,-5.25);
\coordinate (r3m_1) at(6cm,-5.25cm);
\coordinate (r3m_2) at(0.5cm,-5.25cm);
\node (d2) at (d2) {\nscale{$\vdots$}};
\node (d3) at (d3) {\nscale{$\vdots$}};
\draw[->, dashed, rounded corners] (h3m6) -- (r3m_0)--(r3m_1)--node[above] { \escale{$r_{18}$}}(r3m_2)-- (d2);
\coordinate (r6m_2_0) at(11cm,-6.75cm);
\coordinate (r6m_2_1) at(6cm,-6.75cm);
\coordinate (r6m_2_2) at(0.5cm,-6.75cm);
\draw[->, dashed, rounded corners] (d3) -- (r6m_2_0)--(r6m_2_1)--node[above] { \escale{$r_{34}$}}(r6m_2_2)-- (h6m_10);

\graph { 
%H_0
%
(h00) ->["\escale{$k$}"] (h01) ->["\escale{$z_0$}"] (h02) ->["\escale{$v_0$}"] (h03) ->["\escale{$k$}"] (h04) ->["\escale{$q_0$}"] (h05)->["\escale{$z_0$}"](h06);
%H_{3m-1}
%
(h3m_10) ->["\escale{$k$}"] (h3m_11) ->["\escale{$z_{17}$}"] (h3m_12) ->["\escale{$v_{17}$}"] (h3m_13) ->["\escale{$k$}"] (h3m_14) ->["\escale{$q_{17}$}"] (h3m_15)->["\escale{$z_{17}$}"](h3m_16);
%H_{3m}
%
(h3m0) ->["\escale{$k$}"] (h3m1) ->["\escale{$w_0$}"] (h3m2) ->["\escale{$p_{0}$}"] (h3m3) ->["\escale{$k$}"] (h3m4) ->["\escale{$y_{0}$}"] (h3m5)->["\escale{$w_0$}"](h3m6);
%H_{6m-1}
%
(h6m_10) ->["\escale{$k$}"] (h6m_11) ->["\escale{$w_{17}$}"] (h6m_12) ->["\escale{$p_{17}$}"] (h6m_13) ->["\escale{$k$}"] (h6m_14) ->["\escale{$y_{17}$}"] (h6m_15)->["\escale{$w_{17}$}"](h6m_16);

%CONSISTENCY EDGES
(h06)->["\escale{$c_0$}"] (d1); 
(d1)->["\escale{$c_{16}$}"] (h3m_16); 
(h3m_16)->["\escale{$c_{17}$}"] (h3m6); 
(h3m6)->["\escale{$c_{18}$}"] (d3); 
(d3)->["\escale{$c_{34}$}"] (h6m_16); 
};
\end{scope}
%brace H
\draw [decorate, decoration={brace, amplitude=3pt}]
(14.25,-9.5)-- (14.25,-17) node [midway, right,  xshift=0.25cm] {$H$};
\node (z_6) at (0,-8.25){\nscale{$\bot_{97}$}};
\draw (b_96)--(x_6)--(y_6)node[pos=0.5, right]{\escale{$\ominus_{97}$}}--(z_6);
\draw[->](z_6)--(h00)node[pos=0.4, left]{\escale{$\odot_{97}$}};
\end{tikzpicture}
\caption{
The result TS $A(U^{\sigma_4}_\varphi)$ of our reduction (continued).
}
\label{fig:example}
\end{figure}
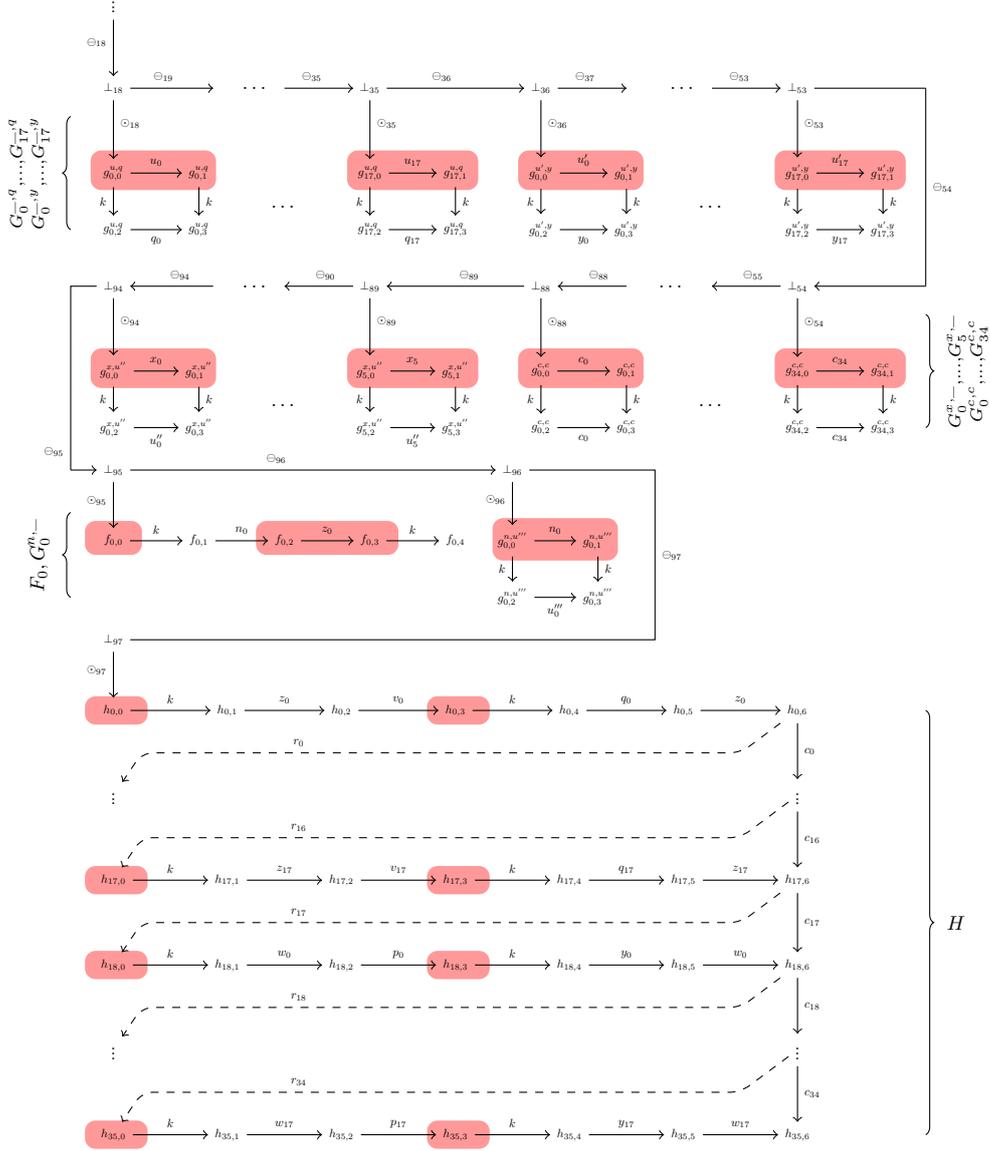

In the following, let $\sigma$ be turned to a position in $\{\sigma_1,\dots, \sigma_6\}$ and $\tau$ be a type of nets from $\sigma$.
To refer to events and states, the generic description of the $U^\sigma_\varphi$ uses lowercase English letters for states and regular events, uppercase $X$ for events that represent the variables of $\varphi$ and lowercase Greek letters for event placeholders, where the actual event depends on the switch position.
Nevertheless, working with many different objects, we cannot refrain from using Greek letters for other purposes, too.

For structure, $U^\sigma_\varphi$ is subdivided into a \emph{key union} $K^\sigma_m$, which depends on $\sigma$ and less on $\varphi$ (in fact, only on the number $m$ of clauses), and a \emph{translator union} $T^\sigma_\varphi$, which depends on $\varphi$ and less on $\sigma$. 
While the key union provides $k$ and $s_{key}$ together with other helpful events, the translator union represents $\varphi$.
The sense in splitting $U^\sigma_\varphi = U(K^\sigma_m, T^\sigma_\varphi)$ is the following:
If we have a $\tau$-region of $U^\sigma_\varphi$ that inhibits $k$ at $s_{key}$, it is by definition decomposed into a $\tau$-region $(sup_K, sig_K)$ of $K^\sigma_m$ and a compatible $\tau$-region $(sup, sig)$ of $T^\sigma_\varphi$.
We call $(sup_K,sig_K)$ a key region and $(sup, sig)$ an indicator region.
The relation between these two regions is only given by the few events shared among $K^\sigma_m$ and $T^\sigma_\varphi$, subsequently called the \emph{interface}.
Then, for a key region, the sole purpose of $K^\sigma_m$ is to somehow regulate the signature of events in the interface.
Using the interface conditions installed by $K^\sigma_m$, a compatible indicator region just makes sure that the variable events describe a one-in-three model of $\varphi$.
Reversely, if such a model exists, we can construct a $\tau$-region for $U^\sigma_\varphi$ inhibiting $k$ at $s_{key}$.

In the next step, we give an abstract description of $T^\sigma_\varphi$.
There, we use the variables $V(\varphi)=\{X_0,\dots, X_{m-1}\}$ as events.
For each $X_j$, we also add a corresponding helper event $x_j$.
Then, every clause $\zeta_i = \{X_{i,0}, X_{i,1}, X_{i,2}\} \subseteq V(\varphi)$ is implemented as a \emph{translator} $T^\sigma_i=U(T^\sigma_{i,0},T^\sigma_{i,1},T^\sigma_{i,2})$, a subunion of $T^\sigma_\varphi$. 
The TS $T^\sigma_{i,\alpha}$ that builds $T^\sigma_i$ with its three copies for $\alpha \in \{0,1,2\}$ is shown in Figure \ref{fig:translators}.1 for $\sigma \in \{\sigma_1, \dots, \sigma_4\}$ and in Figure \ref{fig:translators}.5 for $\sigma \in \{\sigma_5, \sigma_6\}$.
We let any choice of $\alpha \in \{0, 1, 2\}$ select one TS $T^\sigma_{i,\alpha}$ and also define $\beta = \alpha+1 \mod 3$ and $\gamma = \alpha+2 \mod 3$ to address the other two TSs in a specific consecutive manner.

The foundation of translation is to make sure for the three variable events $X_{i,0}, X_{i,1}, X_{i,2}$ of $\zeta_i$ that in an indicator region exactly one of them can get a signature different from \nop.
Taken across all translators, this implements the requirements of a one-in-three model $M =\{X \in V(\varphi) \mid sig(X) \not = \nop\}$ within the union $T^\sigma_\varphi$.
To effect this behavior, the TS $T^\sigma_{i,\alpha}$ provides two paths
\begin{align*}
P_{i,\alpha} &= s^0_{i,\alpha} \dots s^1_{i,\alpha} \edge{X_{i,\alpha}} s^2_{i,\alpha} \dots s^3_{i,\alpha} \edge{X_{i,\beta}} s^4_{i,\alpha,4} \dots s^5_{i,\alpha,5} \edge{X_{i,\gamma}} s^6_{i,\alpha} \dots s^7_{i,\alpha},\\
\dot{P}_{i,\alpha} &= \dot{s}^0_{i,\alpha} \dots \dot{s}^1_{i,\alpha} \bedge{x_{i,\alpha}} \dot{s}^2_{i,\alpha} \dots \dot{s}^3_{i,\alpha} \bedge{x_{i,\beta}} \dot{s}^4_{i,\alpha,4} \dots \dot{s}^5_{i,\alpha} \bedge{x_{i,\gamma}} \dot{s}^6_{i,\alpha} \dots \dot{s}^7_{i,\alpha}
\end{align*}
on states $s^0_{i,\alpha}, \dots, s^7_{i,\alpha}$, respectively $\dot{s}^0_{i,\alpha}, \dots, \dot{s}^7_{i,\alpha}$, containing transitions labeled with $X_{i,\alpha}$, $X_{i,\beta}$, $X_{i,\gamma}$, respectively $x_{i,\alpha}$, $x_{i,\beta}$, $x_{i,\gamma}$, at the given positions.
Notice, while cross checking with Figure \ref{fig:translators}, that states can be the same if they are linked by dots, like $s^0_{i, \alpha}$ and $s^1_{i, \alpha}$ both represent $t_{i, \alpha, 2}$, or have the same superscript, like $s^0_{i, \alpha}$ and $\dot{s}^0_{i, \alpha}$ both represent $t_{i, \alpha, 2}$, too.

For an indicator region $(sup, sig)$ of $T^\sigma_\varphi$, the basis of our construction is a synchronization of certain states on these paths.
Firstly, this concerns for all $j \in \{0, \dots, 7\}$ the opposing states $s^j_{i,\alpha}$ and $\dot{s}^j_{i,\alpha}$ on the two paths, that is, $sup(s^j_{i,\alpha}) = sup(\dot{s}^j_{i,\alpha})$.
Secondly, across all translators, we synchronize all initial and all terminal states
\[sup(s^0_{0,\alpha}) = \dots = sup(s^0_{m-1,\alpha}) \not= sup(s^7_{0,\alpha}) = \dots = sup(s^7_{m-1,\alpha})\]
of the primal paths and make sure that the support of initial differs from terminal states.
As part of the interface, the placeholders $\xi^\sigma_{3i+\alpha}$ and $\theta^\sigma_{3i+\alpha}$, called \emph{materializers}, have a significant role in this second synchronization process.
Playing together, all materializers $\xi^\sigma_{3i}, \xi^\sigma_{3i+1}, \xi^\sigma_{3i+2}, \theta^\sigma_{3i}, \theta^\sigma_{3i+1}, \theta^\sigma_{3i+2}$ of $T^\sigma_i$ also synchronize the initial and terminal states
\[sup(s^0_{i,0}) = sup(s^0_{i,1}) = sup(s^0_{i,2}) \text{ and } sup(s^7_{i,0}) = sup(s^7_{i,1}) = sup(s^7_{i,2})\]
across the three TSs of $T^\sigma_i$.

Recall that the indicator region $(sup, sig)$ maps the TSs of $T^\sigma_i$ to the type of nets TS $\tau$.
This includes $P_{i,\alpha}$ and $\dot{P}_{i,\alpha}$ which become paths in TS $\tau$ traversing along the states $\{0, 1\}$.
By the previous synchronization of states, the events $E(P_{i,\alpha}) \setminus \zeta_i$ on the primal path, respectively $E(\dot{P}_{i,\alpha}) \setminus \{x_{i,0},x_{i,1},x_{i,2}\}$ on the secondary path, are prevented from taking a signature in $\{\inp, \out, \set, \res, \swap\}$.
Traversing the mapped paths in $\tau$, these interactions do not step from $0$ to $1$ or vice versa.
Since the mapped paths start and terminate at different states $\tau$, that is, $sup(s^0_{i,\alpha}) = sup(\dot{s}^0_{i,\alpha}) \not = sup(s^7_{i,\alpha}) = sup(\dot{s}^7_{i,\alpha})$, the remaining interactions $sig(X_{i,0}), sig(X_{i,1}), sig(X_{i,2})$, respectively, $sig(x_{i,0}), sig(x_{i,1}), sig(x_{i,2})$, have to perform an odd number of state changes.
By the synchronization of the paths, there are only eight possibilities for this behavior, where four start the mapped primal path at $1$ and the other four are simply their complements.
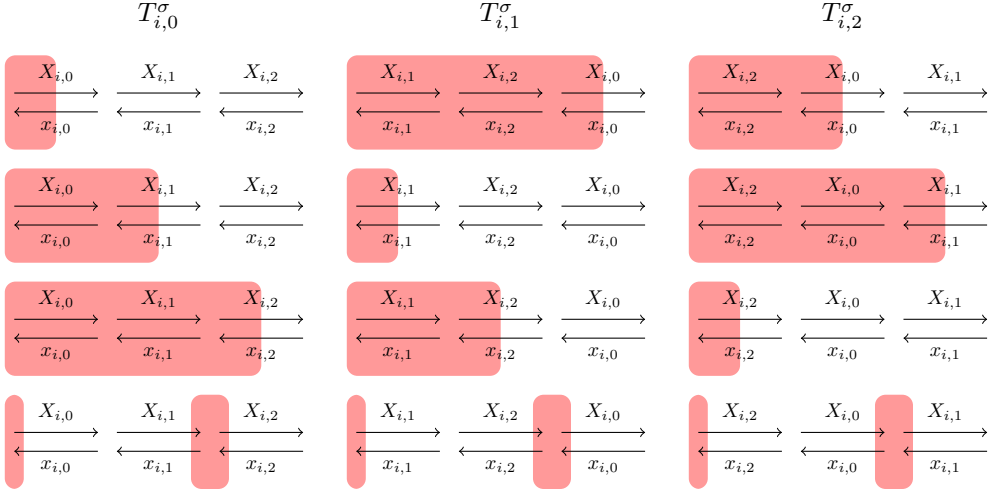
\begin{figure}
\newcommand{\pathlength}{1.1}
\newcommand{\pathspacing}{0.25}
\newcommand{\pathsketch}[3]{
\path (0,0) ++(0.5*\pathspacing,0) coordinate(a0)
-- node[pos=0.5,above] {\scalebox{\edgeScale}{$X_{i,#1}$}} ++(\pathlength,0) coordinate (b0) ++(\pathspacing,0) coordinate(a1)
-- node[pos=0.5,above] {\scalebox{\edgeScale}{$X_{i,#2}$}} ++(\pathlength,0) coordinate (b1) ++(\pathspacing,0) coordinate (a2)
-- node[pos=0.5,above] {\scalebox{\edgeScale}{$X_{i,#3}$}} ++(\pathlength,0) coordinate (b2);
\path (a0|-0,-\pathspacing) -- (b0|-0,-\pathspacing) node[pos=0.5,below] {\scalebox{\edgeScale}{$x_{i,#1}$}}
-- (a1|-0,-\pathspacing) -- (b1|-0,-\pathspacing) node[pos=0.5,below] {\scalebox{\edgeScale}{$x_{i,#2}$}}
-- (a2|-0,-\pathspacing) -- (b2|-0,-\pathspacing) node[pos=0.5,below] {\scalebox{\edgeScale}{$x_{i,#3}$}};
\foreach \i in {0,...,2}{
\draw[->] (a\i)--(b\i);
\draw[<-] (a\i|-0,-0.25) -- (b\i|-0,-0.25);
}}
\newcommand{\pathline}[3]{
\begin{scope}[xshift=0cm]
\fill[red!40, rounded corners] (0,0.5) rectangle (#1*\pathlength+#1*\pathspacing-0.5*\pathlength-0.5*\pathspacing, -0.5-\pathspacing);
\pathsketch{0}{1}{2}
\end{scope}
\begin{scope}[xshift=4.5cm]
\fill[red!40, rounded corners] (0,0.5) rectangle (#2*\pathlength+#2*\pathspacing-0.5*\pathlength-0.5*\pathspacing, -0.5-\pathspacing);
\pathsketch{1}{2}{0}
\end{scope}
\begin{scope}[xshift=9cm]
\fill[red!40, rounded corners] (0,0.5) rectangle (#3*\pathlength+#3*\pathspacing-0.5*\pathlength-0.5*\pathspacing, -0.5-\pathspacing);
\pathsketch{2}{0}{1}
\end{scope}
}
\newcommand{\swapregion}{
\fill[red!40, rounded corners] (0,0.5) rectangle (\pathspacing, -0.5-\pathspacing);
\fill[red!40, rounded corners] (2*\pathlength+\pathspacing,0.5) rectangle (2*\pathlength+3*\pathspacing, -0.5-\pathspacing);
}
\centering
\begin{tikzpicture}
\begin{scope}[yshift=-0.5cm]
\pathline{1}{3}{2}
\end{scope}
\begin{scope}[yshift=-2cm]
\pathline{2}{1}{3}
\end{scope}
\begin{scope}[yshift=-3.5cm]
\pathline{3}{2}{1}
\end{scope}
\begin{scope}[yshift=-5cm]
\begin{scope}[xshift=0cm]
\swapregion
\pathsketch{0}{1}{2}
\path (a1)--(b1) coordinate[pos=0.5] (a);
\end{scope}
\begin{scope}[xshift=4.5cm]
\swapregion
\pathsketch{1}{2}{0}
\path (a1)--(b1) coordinate[pos=0.5] (b);
\end{scope}
\begin{scope}[xshift=9cm]
\swapregion
\pathsketch{2}{0}{1}
\path (a1)--(b1) coordinate[pos=0.5] (c);
\end{scope}
\end{scope}
\node at (a|-0,0.5) {$T^\sigma_{i,0}$};
\node at (b|-0,0.5) {$T^\sigma_{i,1}$};
\node at (c|-0,0.5) {$T^\sigma_{i,2}$};
\end{tikzpicture}
\caption{
Every row shows abstractions of the paths $P_{i,\alpha}, \dot{P}_{i,\alpha}$ in the three translator TSs together with a red-marked support that contains the start states of the paths and excludes their terminals.
The first three rows apply one state change per path and the last row uses three.
Except for complement supports, the four rows demonstrate the only ways to realize the required state changes using just the variable events $X_{i,0}, X_{i,1}, X_{i,2}$, respectively, $x_{i,0}, x_{i,1}, x_{i,2}$, while other path events remain \nop.
In the first three rows, exactly two of the variable events, respectively, exactly two of their counterparts, are forced to \nop, as there are both, transitions entirely within the support and others entirely outside.
In the fourth row, the only possible signature for $X_{i,0}, X_{i,1}, X_{i,2}, x_{i,0}, x_{i,1}, x_{i,2}$ is \swap\ as every event has both, incoming and outgoing transitions relative to the given support.
}
\label{fig:all_border_crossing}
\end{figure}

Figure~\ref{fig:all_border_crossing} sketches the first four cases and teaches us that three state changes on the mapped paths come with $sig(x_{i,0})=sig(x_{i,1})=sig(x_{i,2})=\swap$.
Hence, if \swap\ is not available, like in $\sigma_1$ and $\sigma_2$, an indicator region implements the one-in-three behavior.
In the other four switch positions $\sigma_3, \dots, \sigma_6$, we simply have to prevent the indicator region from assigning \swap\ to any event of $x_{i,0}, x_{i,1}, x_{i,2}$.
In particular, the corresponding translator unions $T^\sigma_\varphi$ install an additional \emph{freezer} $F^\sigma_T$ to hinder these troubling \swap\ assignments.
See also Figure~\ref{fig:example} to get the idea.

The general reduction idea in mind, Subsection~\ref{sec:translators} introduces the details of the translator union $T^\sigma_\varphi$ for every $\sigma\in \{\sigma_1,\dots, \sigma_6\}$.
Subsection~\ref{sec:key_unions} does the same for $K^\sigma_m$.

\medskip
Before we start with our construction, we need some minor tools:
Firstly, we use so called \emph{generators} $G^{\eta,\varrho}_j$ in $K^\sigma_m$ and in $T^\sigma_\varphi$.
A template $G^{\eta,\varrho}_j$ serves as a blueprint for freezer gadget TSs as follows: 
For $j \in \mathbb{N}$ and symbols $a, b$ the template generates TS $G^{a,b}_{j}$ from Figure~\ref{fig:translators}.2 with states $g^{a,b}_{j,0}, \dots, g^{a,b}_{j,3}$ and placeholder $\eta_j$ substituted by event $a_j$ and placeholder $\varrho_j$ by event $b_j$.
Generated TSs are used according to the following lemma:
\begin{lemma}[Without proof]\label{lem:generator}
For $j \in \mathbb{N}$ and symbols $a,b$ let $G^{a,b}_{j}$ be the generated TS.
For a $\tau$-region $(sig,sup)$ of $G^{a,b}_{j}$ the following conditions hold:
\begin{enumerate}
\item 
If $sig(k)=\inp$ then $a_j \in sig^{-1}(\keepo)$ and $b_j \in sig^{-1}(\keepze)$.
\item 
If $sig(k)=\out$ then $a_j \in sig^{-1}(\keepze)$ and $b_j \in sig^{-1}(\keepo)$.
\item 
If $sig(k)=\used$ then $a_j, b_j \in sig^{-1}(\keepo)$.	
\item 
If $sig(k)=\free$ then $a_j,b_j \in sig^{-1}(\keepze)$.	
\end{enumerate}
\end{lemma}

As second minor notion, we introduce forward-backward transitions $s\fbedge{e}s'$ which simply express the presence of both, $s\edge{e}s'$ and $s'\edge{e}s$.
Third notion are \emph{blanc} events.
Some events occur only once in the whole construction with the sole purpose of making states reachable, assuring the solvability of secondary (E)SSP atoms, or satisfying the requirements of Lemma~\ref{lem:union_validity}.
They generally do not help understanding and, instead of introducing confusing names, we simply indicate them by an underscore $\_$.
Moreover, we use blanc events in $s\fbedge{\_}s'$ to say that there is an anonymous event $u$ that occurs exactly twice, namely at the transitions $s\edge{u}s'$ and $s'\edge{u}s$.
As blancs can always be fitted into any given support of the construction by assigning an appropriate signature, we never need to define this explicitly.

%%%%%%%%%%%%%%%%%%%%%%%%%%%%%%%%%%%%%%%%%%%%%%%%%%%%%%%%%%%%%%%%%%%%%%%%%%%%%%%

\subsection{Details of the Translator Union}\label{sec:translators}

This section defines the translator union $T^\sigma_\varphi$ for all cubic monotone boolean $3$-CNF $\varphi$ with $m$ clauses and every $\sigma\in \{\sigma_1,\dots, \sigma_6\}$.
The union $T^\sigma_\varphi=U(T^\sigma_0, \dots, T^\sigma_{m-1}, F^\sigma_T)$ consists of translator subunion $T^\sigma_i$ for $i \in \{0 ,\dots, m-1\}$ to cover all clauses in $\varphi$ and a freezer $F^\sigma_T$ to prevent unwanted \swap.
The translators $T^\sigma_i=U(T^\sigma_{i,0},T^\sigma_{i,1},T^\sigma_{i,2})$ are built from TSs $T^\sigma_{i,\alpha}$ for $\alpha\in \{0,1,2\}$.
First result of this section is Lemma~\ref{lem:translators} fixing the signature of the interface between $K^\sigma_m$ and $T^\sigma_\varphi$ as the basis for the anticipated translator functionality.
Secondly, Lemma~\ref{lem:indicator_regions} establishes the indicator regions of $T^\sigma_\varphi$ binding the existence of a one-in-three model for $\varphi$ to the inhibitability of $k$ at $s_{key}$ in the union $U^\sigma_\varphi$.

Figure~\ref{fig:translators}.1 defines the gadget TS $T^\sigma_{i,\alpha}$  with initial state $t_{i,\alpha,0}$ for $\sigma_1, \dots, \sigma_4$ and Figure~\ref{fig:translators}.5 for $\sigma_5, \sigma_6$ where the initial state is $t'_{i,\alpha,s}$.
The latter contains six events $a_{18i+6\alpha}, \dots, a_{36i+6\alpha+5}$, thus, 18 events for $T^\sigma_i$. 
Figure~\ref{fig:translators}.5 just uses $a_j$ for $a_{18i+6\alpha+j}, j \in \{0, \dots, 5\}$ to preserve clarity.
Moreover, for $\sigma \in \{\sigma_1, \sigma_2, \sigma_4, \sigma_5\}$, every placeholder $\xi^{\sigma}_{3i+\alpha}$ becomes $v_{3i+\alpha}$ and every placeholder $\theta^{\sigma}_{3i+\alpha}$ becomes $w_{3i+\alpha}$.
For $\sigma_3, \sigma_6$ we proceed just the other way around and let every $\xi^{\sigma}_{3i+\alpha}$ be taken by $w_{3i+\alpha}$ and every $\theta^{\sigma}_{3i+\alpha}$ by $v_{3i+\alpha}$.

Notice that Figure~\ref{fig:translators}.1 and~\ref{fig:translators}.5 have several colored areas.
They demonstrate the $T^\sigma_{i,\alpha}$-related fractions of the three possible indicator regions, later defined in detail by Lemma~\ref{lem:indicator_regions}.
The red region fraction stands for $X_{i,\alpha} \in M$, the green for $X_{i,\beta} \in M$ and blue for $X_{i,\gamma} \in M$.
In all three settings, the states taking part in the indicator support are exactly those within the colored area plus $t_{i, \alpha, 0}$ for $\sigma_1, \dots, \sigma_4$ or plus $t'_{i, \alpha, 0}, t'_{i, \alpha, 1}$ for $\sigma_6$.

Figure~\ref{fig:translators}.2 defines the generator template $G^{\eta,\varrho}_j$ with initial state $g^{\eta,\rho}_{j,0}$ and shows the support of a region where $sig(k)=\inp$.
While the freezers $F^{\sigma_1}_T = F^{\sigma_2}_T = U()$ are empty, this template creates the freezers $F^{\sigma_3}_T = U(G^{\_,x}_0, \dots, G^{\_,x}_{m-1})$ and $F^{\sigma_4}_T = U(G^{x,\_}_0, \dots, G^{x,\_}_{m-1})$.
The TS $\mathcal{B}_j$ with initial state $b_{j,4}$ in Figure~\ref{fig:translators}.3 builds the freezer $F^{\sigma_6}_T = U(\mathcal{B}_{0}, \dots, \mathcal{B}_{m-1}) $ and TS $\mathcal{B}'_j$ with initial state $b'_{j,6}$ in Figure~\ref{fig:translators}.4 is for $F^{\sigma_5}_T = U(\mathcal{B}'_{0}, \dots, \mathcal{B}'_{m-1})$.
In $\mathcal{B}_j$ and $\mathcal{B}'_j$ the red areas mark the support of an indicator region as defined for Lemma~\ref{lem:indicator_regions}.
Figure~\ref{fig:translators}.6 defines the TS $F_2$ with initial state $f_{2,0}$ for the freezer $F^{\sigma_3}_K$ of $K^{\sigma_3}_m$ to be introduced in Section~\ref{sec:key_unions}.
The red area shows the respective fraction of a key region support.
\begin{figure}[h]
\centering
\begin{tikzpicture}[new set = import nodes,scale=0.9]
\begin{scope}[nodes={set=import nodes}]% make all nodes part of this set
 	%Translatorpart T_{i,\alpha}
		\node (P) at (-0.75, 2) {$1)$};
		\coordinate (0) at (1,2.5);
		\coordinate(1) at (3,1.5);
		\foreach \i in {2,...,5} { \pgfmathparse{\i-2} \coordinate (\i) at (\pgfmathresult*2cm,0) ;}
		\foreach \i in {2} {\fill[blue!30, rounded corners] (\i) +(-0.7,-0.7) rectangle +(4.6,0.8);}
		\foreach \i in {2} {\fill[green!50, rounded corners] (\i) +(-0.6,-0.6) rectangle +(2.6,0.7);}
		\foreach \i in {2} {\fill[red!40, rounded corners] (\i) +(-0.5,-0.25) rectangle +(0.6,0.35);}
		\foreach \i in {0,...,5} { \node (\i) at (\i) {\scalebox{\nodeScale}{$t_{i,\alpha,\i}$}};}
\end{scope}
\graph {
	(import nodes);
			0 ->["\scalebox{\edgeScale}{$k$}"]1;
			1 ->[swap, bend right =30, "\scalebox{\edgeScale}{$\xi^\sigma_{3i+\alpha}$}"]2;
			1 ->[bend left=30,"\scalebox{\edgeScale}{$\theta^\sigma_{3i+\alpha}$}"]5;
			2->[bend left=15 , "\scalebox{\edgeScale}{$X_{i,\alpha}$}"]3;
			3 ->[bend left=15 , "\scalebox{\edgeScale}{$X_{i,\beta}$}"]4;
			4 ->[ bend left=15 , "\scalebox{\edgeScale}{$X_{i,\gamma}$}"]5;

			3 ->[bend left=15 , "\scalebox{\edgeScale}{$x_{i,\alpha}$}"]2;
			4 ->[bend left=15 , "\scalebox{\edgeScale}{$x_{i,\beta}$}"]3;
			5 ->[bend left=15 ,  "\scalebox{\edgeScale}{$x_{i,\gamma}$}"]4;
	
	};

\begin{scope}[nodes={set=import nodes},xshift=7.5cm, yshift=2cm ]
 	%freezer for X, if sigma_2
		\node (G) at (-1,0) {$2)$};
		\foreach \i in {0,1} {  \coordinate (\i) at (\i*1.5cm,0) ;}
		\foreach \i in {1} {\fill[red!40, rounded corners] (\i) +(-1.9,-0.3) rectangle +(0.4,0.4);}
		\node (0) at (0,0) {\scalebox{\nodeScale}{$g^{\eta,\varrho}_{j,0}$}};
		\node (1) at (1.5,0) {\scalebox{\nodeScale}{$g^{\eta,\varrho}_{j,1}$}};	
		\node (2) at (0,-1.2) {\scalebox{\nodeScale}{$g^{\eta,\varrho}_{j,2}$}};
		\node (3) at (1.5,-1.2) {\scalebox{\nodeScale}{$g^{\eta,\varrho}_{j,3}$}};
\end{scope}
\graph {
	(import nodes);
			0 ->["\scalebox{\edgeScale}{$\eta_j$}"]1;
			0 ->[swap, "\scalebox{\edgeScale}{$k$}"]2;
			2 ->[swap,"\scalebox{\edgeScale}{$\varrho_j$}"]3;
			1 ->["\scalebox{\edgeScale}{$k$}"]3;
			};

\begin{scope}[nodes={set=import nodes},yshift=2cm, xshift=11cm]% make all nodes part of this set
 	%freezer if set \not \in \tau which implies res in \tau
		\node (G) at (-1,0) {$6)$};
		\coordinate (0) at (0,0); \coordinate (3) at (1.5,-1);
		\foreach \i in {0,3} {\fill[red!40, rounded corners] (\i) +(-0.4,-0.3) rectangle +(0.4,0.3);}
		\node (0) at (0,0) {\scalebox{\nodeScale}{$f_{2,0}$}};
		\node (1) at (1.5,0) {\scalebox{\nodeScale}{$f_{2,1}$}};	
		\node (2) at (0,-1) {\scalebox{\nodeScale}{$f_{2,2}$}};
		\node (3) at (1.5,-1) {\scalebox{\nodeScale}{$f_{2,3}$}};
\end{scope}
\graph {
	(import nodes);
			0 ->["\scalebox{\edgeScale}{$n_0$}"]1;
			0 ->[swap, "\scalebox{\edgeScale}{$k$}"]2;
			2 ->[swap,"\scalebox{\edgeScale}{$\_$}"]3;
			3 ->[swap, "\scalebox{\edgeScale}{$k$}"]1;
			
			};
\begin{scope}[nodes={set=import nodes}, yshift=-1.5cm]
 	%freezer for X if sigma_5
		\node (G) at (-0.75,0) {$3)$};
		\foreach \i in {0,...,4} { \coordinate (\i) at (\i*2cm,0);}
		\foreach \i in {4} {\fill[red!40, rounded corners] (\i) +(-8.4,-0.3) rectangle +(0.4,0.4);}
		\foreach \i in {0,...,4} { \node (\i) at (\i*2cm,0) {\scalebox{\nodeScale}{$b_{j,\i}$}};}
\end{scope}
\graph {
	(import nodes);
			0 <->["\scalebox{\edgeScale}{$k$}"]1<->[ "\scalebox{\edgeScale}{$x_j$}"]2<->["\scalebox{\edgeScale}{$k$}"]3<->["\scalebox{\edgeScale}{$\_$}"]4;
			};
\begin{scope}[nodes={set=import nodes}, yshift=-2.5cm]
 	%freezer for X if sigma_4
		\node (G) at (-0.75,0) {$4)$};
		\foreach \i in {2,3} { \coordinate (\i) at (\i*2cm,0);}
		\foreach \i in {3} {\fill[red!40, rounded corners] (\i) +(-2.6,-0.3) rectangle +(0.6,0.4);}
		\foreach \i in {0,...,6} { \node (\i) at (\i*2cm,0) {\scalebox{\nodeScale}{$b'_{j,\i}$}};}
\end{scope}
\graph {
	(import nodes);
			0 <->["\scalebox{\edgeScale}{$k$}"]1<->[ "\scalebox{\edgeScale}{$q_2$}"]2<->["\scalebox{\edgeScale}{$x_j$}"]3<->["\scalebox{\edgeScale}{$q_3$}"]4 <->["\scalebox{\edgeScale}{$k$}"]5<->["\scalebox{\edgeScale}{$\_$}"]6;
			};
\node at (-0.75,-4.2) {5)};
\begin{scope}[nodes={set=import nodes},yshift=-7cm, scale =1.2]
%translator T_{i,\alpha} for sigma_4/sigma_5
\coordinate (21) at (3.4,2.8);
\node (t21) at (21) {\scalebox{\nodeScale}{$t'_{i,\alpha,s}$}};
\coordinate (0) at (5.4,2.8);
\coordinate (1) at (5.4,1.8);
\coordinate (help_1) at (0,1.8);
\coordinate (help_2) at (10.8,1.8);
\foreach \i in {2,5,8,11} {\pgfmathparse{int(\i-2} , \coordinate (\i) at (\pgfmathresult*1.2cm,0);}
\foreach \i in {3,4,6,7,9,10} {\pgfmathparse{int(\i-2} , \coordinate (\i) at (\pgfmathresult*1.2cm,1);}
\coordinate (12) at (1.2cm,-1);
\coordinate (13) at (1.8cm,-0);
\coordinate (14) at (2.4cm,-1);
\coordinate (15) at (4.8cm,-1);
\coordinate (16) at (5.4cm,-0);
\coordinate (17) at (6cm,-1);
\coordinate (18) at (8.4cm,-1);
\coordinate (19) at (9cm,-0);
\coordinate (20) at (9.6cm,-1);
\fill[blue!30, rounded corners] (-0.6,0) --(-0.6,-1.4)--(9,-1.4)--(9,-0.5) -- (9.6,-0.5) -- (9.6,0.25)-- (9,0.25)-- (9,1.6)-- (-0.6,1.6)--(-0.6,0);
\fill[green!50, rounded corners] (-0.5,0) --(-0.5,-1.3)--(5.4,-1.3)--(5.4,-0.5) -- (6,-0.5) -- (6,0.25)-- (5.4,0.25)-- (5.4,1.5)-- (-0.5,1.5)--(-0.5,0);
\fill[red!40, rounded corners] (-0.4,0) --(-0.4,-1.2)--(1.8,-1.2)--(1.8,-0.5) -- (2.4,-0.5) -- (2.4,0.25)-- (1.8,0.25)-- (1.8,1.4)-- (-0.4,1.4)--(-0.4,0);
%\fill[red!40, ] (-0.4,0) --(-0.4,-1.3)--(5.4,-1.3)--(5.4,-0.5) -- (6,-0.5) -- (6,0.25)-- (5.4,0.25)-- (5.4,1.4)-- (-0.4,1.4)--(-0.4,0);
\foreach \i in {0,...,20} {\node (t\i) at (\i) {\scalebox{\nodeScale}{$t'_{i,\alpha,\i}$}};}
\graph{ 
(t21)<->["\_"](t0);
(t1)<-["\scalebox{\edgeScale}{$\theta^\sigma_{3i+\alpha}$}"](help_2)->(t11);
(t0) <->["\scalebox{\edgeScale}{$k$}"] (t1)<-[swap, "\scalebox{\edgeScale}{$\xi^\sigma_{3i+\alpha}$}"] (help_1)->(t2) <->[ "\scalebox{\edgeScale}{$a_{0}$}"] (t3) <->["\scalebox{\edgeScale}{$X_{i,\alpha}$}"] (t4)<->["\scalebox{\edgeScale}{$a_{0}$}"] (t5)<->["\scalebox{\edgeScale}{$a_{1}$}"] (t6)<->["\scalebox{\edgeScale}{$X_{i,\beta}$}"](t7)<->["\scalebox{\edgeScale}{$a_{1}$}"] (t8)<->["\scalebox{\edgeScale}{$a_{2}$}"] (t9)<->["\scalebox{\edgeScale}{$X_{i,\gamma}$}"](t10)<->["\scalebox{\edgeScale}{$a_{2}$}"] (t11);
(t2)<->[swap, "\scalebox{\edgeScale}{$a_{3}$}"] (t12)<->["\scalebox{\edgeScale}{$x_{i,\alpha}$}"] (t13)<-["\scalebox{\edgeScale}{$x_{i,\alpha}$}"] (t14)<->[swap,"\scalebox{\edgeScale}{$a_{3}$}"] (t5)<->[swap,"\scalebox{\edgeScale}{$a_{4}$}"](t15)<->["\scalebox{\edgeScale}{$x_{i,\beta}$}"] (t16)<-["\scalebox{\edgeScale}{$x_{i,\beta}$}"] (t17)<->[swap,"\scalebox{\edgeScale}{$a_{4}$}"](t8)<->[swap,"\scalebox{\edgeScale}{$a_{5}$}"] (t18)<->["\scalebox{\edgeScale}{$x_{i,\gamma}$}"] (t19)<-["\scalebox{\edgeScale}{$x_{i,\gamma}$}"] (t20)<->[swap, "\scalebox{\edgeScale}{$a_{5}$}"](t11);
};
\end{scope}%
\end{tikzpicture}
\caption{The ingredients of translator union $T^\sigma_\varphi$ (1-5) where the three colored areas mark the supports of the three possible indicator regions. (1) $T^\sigma_{i,\alpha}$ used for $\sigma_1, \dots, \sigma_4$, (2) template $G^{\eta,\varrho}_j$, (3,4) $\mathcal{B}_j$ and $\mathcal{B}'_j$, (5) $T^\sigma_{i,\alpha}$ used for $\sigma_5, \sigma_6$, (6) TS $F_2$ used in the key union $K^{\sigma_3}_m$.}
\label{fig:translators}
\end{figure}
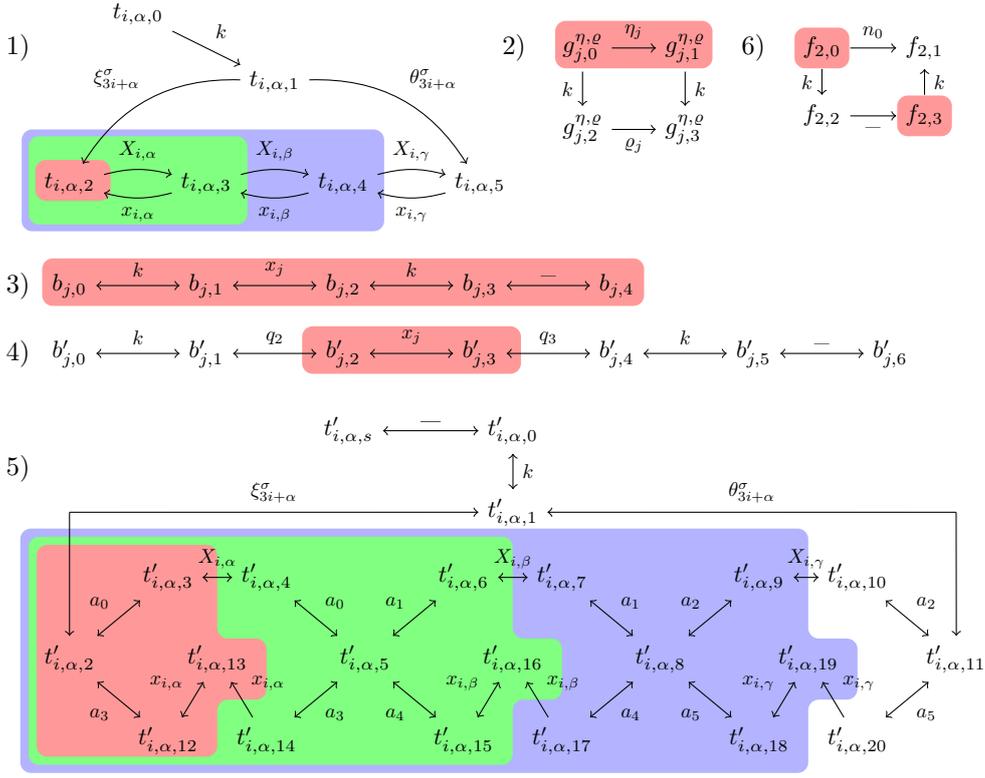

The following lemma provides the condition of the interface between $T^\sigma_\varphi$ and $K^{\sigma}_m$ that is required in an indicator region of the translator union.
Aside from $k$, the interface consist of $V=\{v_0, \dots, v_{3m-1}\}$, $W=\{w_0,\dots, w_{3m-1}\}$ and, for $\sigma_5$ and $\sigma_6$, $Acc=\{a_0, \dots,  a_{18m-1}\}$.
\begin{lemma}\label{lem:translators}
If $\varphi$ is a cubic monotone boolean $3$-CNF with $m$ clauses, $\sigma \in \{\sigma_1, \dots, \sigma_6\}$ a turn switch position for our reduction scheme, $\tau \in \sigma$ a type of nets managed by $\sigma$ and $(sup, sig)$ a $\tau$-region of $T^\sigma_\varphi$ where one of the conditions
\begin{enumerate}
\item\label{lem:translators_1}
$\sigma \in \{\sigma_1,\sigma_2, \sigma_3,\sigma_4\}$, $sig(k) = \inp$, $V \subseteq sig^{-1}(\enter)$ and $W \subseteq sig^{-1}(\keepze)$,
\item\label{lem:translators_2} 
$\sigma \in \{\sigma_1,\sigma_2, \sigma_3,\sigma_4\}$, $sig(k)=\out$, $V \subseteq sig^{-1}(\keepo)$ and $W \subseteq sig^{-1}(\exit)$, or
\item\label{lem:translators_3}  
$\sigma \in \{\sigma_5,\sigma_6\}$, $sig_K(k)\in \{\used,\free\}$, $W \cap sig^{-1}(\swap) = Acc \cap sig^{-1}(\swap) = \emptyset$, $V \subseteq sig^{-1}(\swap)$ and $sig(q_2)=\swap \Leftrightarrow sig(q_3)=\swap$
\end{enumerate}
holds, then $(sup,sig)$ is an indicator region, meaning $M=\{X \in V(\varphi) \mid sig(X)\not=\nop\}$ is a one-in-three model of $\varphi$.
\end{lemma}

The proof of Lemma~\ref{lem:translators} is rather technical and has therefore been moved to Section~\ref{sec:techproofs}.
Next, we have to be able to go the other way around, that is, we need to construct an indicator region $(sup, sig)$ for any given one-in-three model $M$ of $\varphi$.
It is important that, on the interface, $(sup, sig)$ is compatible with a key region $(sup_K,sig_K)$ such that both of them can be combined to a region of $U^\sigma_\varphi$ that inhibits $k$ at the key state.

For given $\varphi$ with $m$-clauses and one-in-three model $M$, our approach is as follows:
We first define for every clause $\zeta_i = \{X_{i,0}, X_{i,1}, X_{i,2}\}, i \in \{0, \dots, m-1\}$ the selector $\alpha_i = j$ by $M \cap \zeta_i = X_{i,j}$.
Hence, $\alpha_i$ is the index of the unique variable in $\zeta_i$ that is part of the model.
Again, $\beta_i$ and $\gamma_i$ are the $(\text{mod}\ 3)$-continuations of $\alpha_i$ as defined above.
Then, depending on the switch $\sigma$, we set up a support $sup^{\sigma}_{i, \alpha_i}$ covering only the states of translator $T^\sigma_i$ representing $\zeta_i$ and a separate support $sup^\sigma_{F}$ for the freezer $F^\sigma_T$, which is empty if $F^\sigma_T$ is empty.
Our goal for every $\tau \in \sigma$ is to extend the combined support $sup^\sigma = sup^\sigma_{0,\alpha_0} \cup \dots \cup sup^\sigma_{m-1,\alpha_{m-1}} \cup sup^\sigma_F$ with a signature $sig: E(T^\sigma_\varphi) \rightarrow \tau$ in order to obtain an indicator $\tau$-region for $T^\sigma_\varphi$. 
Consider the following state sets for our objective:

\begin{enumerate}
\setlength{\itemsep}{3pt}
\item %initial states translators
$S^0_{\sigma_1,i}=S^0_{\sigma_2, i}=S^0_{\sigma_3, i}=S^0_{\sigma_4, i}=\{t_{i,0,0},t_{i,1,0},t_{i,2,0}\}$,
\item 
$S^0_{\sigma_5, i}=\emptyset$ and $S^0_{\sigma_6, i}=\{t'_{i,j,0}, t'_{i,j,1} \mid 0\leq j\leq 2 \}$,
\item %further translator states \sigma_1,\sigma_4
$S^1_{\sigma_1,i,\alpha_i}=S^1_{\sigma_2,i,\alpha_i}=S^1_{\sigma_4,i,\alpha_i}= \{t_{i,\alpha_i,2},t_{i,\beta_i,2},t_{i,\beta_i,3},t_{i,\beta_i,4},t_{i,\gamma_i,2},t_{i,\gamma_i,3}\}$,
\item %further translator states \sigma_3
$S^1_{\sigma_3,i,\alpha_i}=\{t_{i,\alpha_i,3},t_{i,\alpha_i,4},t_{i,\alpha_i,5},t_{i,\beta_i,5},t_{i,\gamma_i,4},t_{i,\gamma_i,5}\}$
\item %translator states sigma_5,sigma_6
$S^1_{\sigma_5,i,\alpha_i}=S^1_{\sigma_6,i,\alpha_i}= N_0\cup N_1\cup N_2$ where
\begin{enumerate}
%N_0 
\item
$N_0=\{t'_{i,\alpha_i,2},t'_{i,\alpha_i,3}, t'_{i,\alpha_i,12}, t'_{i,\alpha_i,13}\}$
%N_1
\item 
$N_1= \{t'_{i,\beta_i,2},\dots,t'_{i,\beta_i,9}, t'_{i,\beta_i,12},\dots,t'_{i,\beta_i,19}\}$,
%N_2
\item
$N_2=\{t'_{i,\gamma_i,2},\dots, t'_{i,\gamma_i,6} , t'_{i,\gamma_i,12},\dots, t'_{i,\gamma_i,16}\}$,
\end{enumerate}
\item % x-freezer states for \sigma_3,\sigma_4
$sup^{\sigma_3}_F=\{g^{\_,x}_{j,0}, g^{\_,x}_{j,1}\mid j\in \{0,\dots, m-1\}\}$ and $sup^{\sigma_4}_F=\{g^{x,\_}_{j,0}, g^{x,\_}_{j,1}\mid j\in \{0,\dots, m-1\}\}$,
\item % x-freezer states for \sigma_5\sigma_6 
$sup^{\sigma_5}_F=\{b'_{j,2}, b'_{j,3}\mid j\in \{0,\dots, m-1\}\}$ and $sup^{\sigma_6}_F=\{b_{j,0},\dots, b_{j,4}\mid j\in \{0,\dots, m-1\}\}$.
\end{enumerate}

Based on this, we simply define $sup^\sigma_{i,\alpha_i}=S^0_{\sigma,i}\cup S^1_{\sigma,i,\alpha_i}$ for all $i\in \{0,\dots, m-1\}$ and all $\sigma \in \{\sigma_1, \dots, \sigma_6\}$.
For a transition $s \edge{e} s'$ of $T^\sigma_\varphi$ ,we can then set $sig(e) = \nop$ if and only if $e$ is not in $\{k, q_2, q_3\} \cup \{X_{i,\alpha_i}, x_{i,\alpha_i}, v_{3i}, \dots, v_{3i+2} \mid 0 \leq i < m\}$.
Hence, to extend the support with an appropriate signature we only have to worry about these remaining events.
Firstly, the idea is to, dependent on the turn switch position $\sigma$, assign the interaction placeholders defined in Figure~\ref{fig:operations} to these events of $T^\sigma_\varphi$.
Then, in the second step, replacing the placeholders with the interactions specified in Figure~\ref{fig:operations} leads to an indicator $\tau$-region for every $\sigma$ and every $\tau \in \sigma$.
The following Lemma~\ref{lem:indicator_regions} realizes and justifies this idea: 
\begin{lemma}[Without proof]\label{lem:indicator_regions}
For every cubic monotone boolean $3$-CNF $\varphi$ with one-in-three model $M$, every $\sigma \in \{\sigma_1, \dots, \sigma_6\}$, every $\tau \in \sigma$ and $V = \{v_0, \dots, v_{3m-1}\}$, we get an indicator $\tau$-region $(sup^\sigma, sig^\sigma)$ for $T^\sigma_\varphi$ with $sup^\sigma = sup^\sigma_{0,\alpha_0} \cup \dots \cup sup^\sigma_{m-1,\alpha_{m-1}} \cup sup^\sigma_F$ and
\[ sig^\sigma(e)=
\begin{cases}
\textsf{op}^\sigma_k, & \text{if } e=k,\\
\textsf{op}^\sigma_V, & \text{if } e \in V,\\
\textsf{op}^\sigma_M, & \text{if } e \in M,\\
\textsf{op}^\sigma_x & \text{if } e \in \{x_{i,\alpha_i} \mid 0 \leq i < m\},\\
\textsf{op}^\sigma_q & \text{if } e \in \{q_2,q_3\},\\
\nop & \text{otherwise.}
\end{cases}
\]
\end{lemma}

The lemma does not need a proof, as we only need to verify for every transition $s \edge{e} s'$ of $T^\sigma_\varphi$ that there is mapped transition $sup^\sigma(s) \edge{sig^{\sigma}(e)}sup^\sigma(s')$ in TS $\tau$ for every $\tau\in \sigma$.
\begin{figure}
\centering
\begin{tabular}{ p{1cm} p{1.5cm}p{1.5cm}p{1.5cm}p{1.5cm} p{1cm} }
$\sigma$ & $\textsf{op}^{\sigma}_k$ & $\textsf{op}^{\sigma}_M$ & $\textsf{op}^{\sigma}_V$ &  $\textsf{op}^{\sigma}_x$ &$\textsf{op}^{\sigma}_q$ \\ \hline
$\sigma_1$ & \inp & \inp & \out & \out& \\
$\sigma_2$ & \inp & \inp & \set& \set& \\
$\sigma_3$ &  \inp & \swap& \swap & \res& \\
$\sigma_4$ &  \inp & \swap & \swap & \set& \\
$\sigma_5$ &\free & \swap & \swap& \set & \swap \\
$\sigma_6$ & \used & \swap & \swap & \set &  
\end{tabular}
\caption{For $\sigma\in \{\sigma_1,\dots,\sigma_6\}$ the operations $\textsf{op}^{\sigma}_k, \textsf{op}^{\sigma}_M, \textsf{op}^{\sigma}_V, \textsf{op}^{\sigma}_x$ and $\textsf{op}^{\sigma}_q$ to be used in Lemma~\ref{lem:indicator_regions}.}
\label{fig:operations}
\end{figure}

%%%%%%%%%%%%%%%%%%%%%%%%%%%%%%%%%%%%%%%%%%%%%%%%%%%%%%%%%%%%%%%%%%%%%%%%%

\subsection{Details of the Key Union}\label{sec:key_unions}

This subsection defines the key union $K^\sigma_m$ for all numbers $m$ of clauses and every $\sigma \in \{\sigma_1,\dots, \sigma_6\}$.
In particular, $K^\sigma_m = U(H^\sigma, D^\sigma, G^\sigma, F^\sigma_K)$ consists of the \emph{head} $H^\sigma$, the \emph{duplicator} $D^\sigma$, the \emph{generator} $G^\sigma$, and the \emph{freezer} $F^\sigma_K$.
How these ingredients are constructed depends on $\sigma$.
Firstly, for $\sigma \in \{\sigma_1, \dots, \sigma_4\}$ we leave the generator empty and proceed as follows:
 
\begin{enumerate}
\item 
The head $H^\sigma=H$ is chosen as depicted in Figure~\ref{fig:key_unions}.1.
It provides key event $k$ and key state $s_{key}=h_{0,6}$ as well as the full interface $V, W$.
\item 
The duplicator $D^\sigma=U(G^{c,c}_0,\dots, G^{c,c}_{6m-2})$ is constructed from generator templates.
It provides the events $C=\{c_0,\dots,c_{6m-2}\}$ that, for a key region, receive \nop\ and therefore synchronize the head states $h_{j,6}$ and $h_{j+1,6}$ for all $j\in \{0, \dots, 6m-2\}$. 
\item 
For a key region, the freezer $F^\sigma_K$ assures that $k$ is assigned \inp\ or \out. 
Moreover, it prevents all events of $Q = \{q_0, \dots, q_{3m-1}\}$ and of $Y=\{y_0, \dots, y_{3m-1}\}$ from receiving \swap.
We firstly let $F^{\sigma_1}_K = F^{\sigma_2}_K = U(F_0, F_1)$ using $F_0$ from Figure~\ref{fig:key_unions}.2 and $F_1$ from Figure~\ref{fig:key_unions}.3.
Then, $F^{\sigma_3}_K = U(F_0, F_2, G^{\_,q}_0, \dots, G^{\_,q}_{3m-1}, G^{\_,y}_0, \dots, G^{\_,y}_{3m-1})$ is build of $F_0$ from Figure~\ref{fig:key_unions}.2, $F_2$ from Figure~\ref{fig:translators}.6, as well as generator templates from Figure~\ref{fig:translators}.2.
Similarly, $F^{\sigma_4}_K = U(F_0, G^{n,\_}_0, G^{\_,q}_0, \dots, G^{\_,q}_{3m-1}, G^{\_,y}_0, \dots, G^{\_,y}_{3m-1})$.
\end{enumerate}

At this point, one may notice that $\sigma_1$ and $\sigma_2$ actually transform input $\varphi$ into the same TS and thus, could be consolidated into one switch position.
But since there are differences in the constructed regions as defined in Figure \ref{fig:operations}, we keep the two switch positions distinguished to make our argumentation simpler.

For $\sigma \in \{\sigma_5, \sigma_6\}$ we create different key union ingredients as follows:

\begin{enumerate}
\item 
Here $H^\sigma = U(H'_0, \dots, H'_{3m-1})$ consists of multiple TSs $H'_j$ from Figure~\ref{fig:key_unions}.4.
The head again introduces $k$, but here $h'_{0,2}$ is the key state and only $V$ is provided to the interface.
\item  
The duplicator $D^\sigma = U(D_{0}, \dots, D_{18m-1})$ consists of multiple TSs $D_{j}$ from Figure~\ref{fig:key_unions}.8.
It provides $Acc$ to the interface and prevents these events from \swap\ in key regions.
\item 
The generator $G^\sigma = U(G_{0}, \dots, G_{3m-1})$, made of multiple $G_{j}$ from Figure~\ref{fig:key_unions}.9, provides $W$ for the interface and prevents the respective events from \swap\ in a key region.
\item
The freezer $F^\sigma_K = U(F'_0, F'_1, F'_2)$ consists of the TSs $F'_0,F'_1,F'_2$ from Figure~\ref{fig:key_unions}.5, Figure~\ref{fig:key_unions}.6 and Figure~\ref{fig:key_unions}.7 and provides the interface events $q_2, q_3$.
In a key region, the freezer makes sure that $q_2$ is assigned \swap\ if and only if $q_3$ gets \swap\ and, furthermore, enforces \nop\ or \swap\ onto event $z$, which synchronizes some states in other TSs.
\end{enumerate}

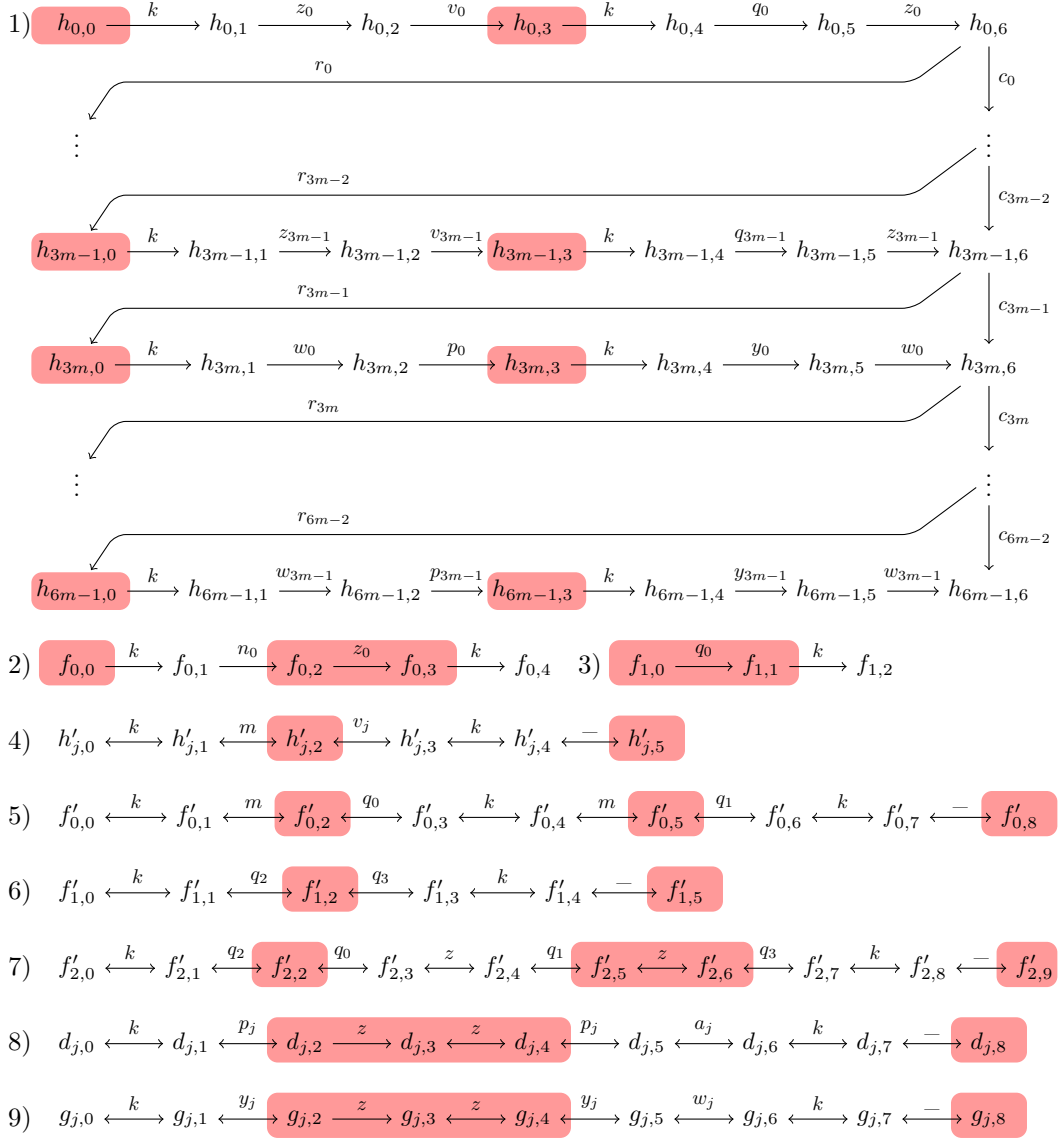
\begin{figure}[t!]
\centering
\begin{tikzpicture}[new set = import nodes]
\begin{scope}[nodes={set=import nodes}]
%headmaster 1 for sigma_1...sigma_3
\node at (-0.75,0) {1)};
%States H_0 
\foreach \i in {0,...,6} {\coordinate (h0\i) at (\i*2cm,0);}
%States H_{3m-1}
\foreach \i in {0,...,6} {\coordinate (h3m_1\i) at (\i*2cm,-3);}
%States H_{3m}
\foreach \i in {0,...,6} {\coordinate (h3m\i) at (\i*2cm,-4.5);}
%States H_{6m-1}
\foreach \i in {0,...,6} {\coordinate (h6m_1\i) at (\i*2cm,-7.5);}
\foreach \i in {h00,h03,h3m_10,h3m_13,h3m0,h3m3,h6m_10,h6m_13} {\fill[red!40, rounded corners] (\i) +(-0.6,-0.25) rectangle +(0.7,0.25);}
%using the defined state coordinates
\foreach \i in {0,...,6} {\node (h0\i) at (h0\i) {\scalebox{\nodeScale}{$h_{0,\i}$}};}
\foreach \i in {0,...,6} {\node (h3m_1\i) at (h3m_1\i) {\scalebox{\nodeScale}{$h_{3m-1,\i}$}};}
\foreach \i in {0,...,6} {\node (h3m\i) at (h3m\i) {\scalebox{\nodeScale}{$h_{3m,\i}$}};}
\foreach \i in {0,...,6} {\node (h6m_1\i) at (h6m_1\i) {\scalebox{\nodeScale}{$h_{6m-1,\i}$}};}

%REACHABILITY STUFF
\coordinate (d0) at(0,-1.5);
\coordinate (d1) at(12cm,-1.5cm);
\coordinate (r0_0) at(11,-0.75);
\coordinate (r0_1) at(6cm,-0.75cm);
\coordinate (r0_2) at(0.5cm,-0.75cm);
\node (d0) at (d0) {\scalebox{\nodeScale}{$\vdots$}};
\node (d1) at (d1) {\scalebox{\nodeScale}{$\vdots$}};
\draw[->, , rounded corners] (h06) -- (r0_0)--(r0_1)--node[above] { \scalebox{\edgeScale}{$r_{0}$}}(r0_2)-- (d0);
\coordinate (r3m_2_0) at(11cm,-2.25cm);
\coordinate (r3m_2_1) at(6cm,-2.25cm);
\coordinate (r3m_2_2) at(0.5cm,-2.25cm);
\draw[->, , rounded corners] (d1) -- (r3m_2_0)--(r3m_2_1)--node[above] { \scalebox{\edgeScale}{$r_{3m-2}$}}(r3m_2_2)-- (h3m_10);
\coordinate (r3m_1_0) at(11cm,-3.75cm);
\coordinate (r3m_1_1) at(6cm,-3.75cm);
\coordinate (r3m_1_2) at(0.5cm,-3.75cm);
\draw[->, , rounded corners] (h3m_16) -- (r3m_1_0)--(r3m_1_1)--node[above] { \scalebox{\edgeScale}{$r_{3m-1}$}}(r3m_1_2)-- (h3m0);
\coordinate (d2) at(0,-6);
\coordinate (d3) at(12cm,-6cm);
\coordinate (r3m_0) at(11,-5.25);
\coordinate (r3m_1) at(6cm,-5.25cm);
\coordinate (r3m_2) at(0.5cm,-5.25cm);
\node (d2) at (d2) {\scalebox{\nodeScale}{$\vdots$}};
\node (d3) at (d3) {\scalebox{\nodeScale}{$\vdots$}};
\draw[->, , rounded corners] (h3m6) -- (r3m_0)--(r3m_1)--node[above] { \scalebox{\edgeScale}{$r_{3m}$}}(r3m_2)-- (d2);
\coordinate (r6m_2_0) at(11cm,-6.75cm);
\coordinate (r6m_2_1) at(6cm,-6.75cm);
\coordinate (r6m_2_2) at(0.5cm,-6.75cm);
\draw[->, , rounded corners] (d3) -- (r6m_2_0)--(r6m_2_1)--node[above] { \scalebox{\edgeScale}{$r_{6m-2}$}}(r6m_2_2)-- (h6m_10);

\graph { 
%H_0
%
(h00) ->["\scalebox{\edgeScale}{$k$}"] (h01) ->["\scalebox{\edgeScale}{$z_0$}"] (h02) ->["\scalebox{\edgeScale}{$v_0$}"] (h03) ->["\scalebox{\edgeScale}{$k$}"] (h04) ->["\scalebox{\edgeScale}{$q_0$}"] (h05)->["\scalebox{\edgeScale}{$z_0$}"](h06);
%H_{3m-1}
%
(h3m_10) ->["\scalebox{\edgeScale}{$k$}"] (h3m_11) ->["\scalebox{\edgeScale}{$z_{3m-1}$}"] (h3m_12) ->["\scalebox{\edgeScale}{$v_{3m-1}$}"] (h3m_13) ->["\scalebox{\edgeScale}{$k$}"] (h3m_14) ->["\scalebox{\edgeScale}{$q_{3m-1}$}"] (h3m_15)->["\scalebox{\edgeScale}{$z_{3m-1}$}"](h3m_16);
%H_{3m}
%
(h3m0) ->["\scalebox{\edgeScale}{$k$}"] (h3m1) ->["\scalebox{\edgeScale}{$w_0$}"] (h3m2) ->["\scalebox{\edgeScale}{$p_{0}$}"] (h3m3) ->["\scalebox{\edgeScale}{$k$}"] (h3m4) ->["\scalebox{\edgeScale}{$y_{0}$}"] (h3m5)->["\scalebox{\edgeScale}{$w_0$}"](h3m6);
%H_{6m-1}
%
(h6m_10) ->["\scalebox{\edgeScale}{$k$}"] (h6m_11) ->["\scalebox{\edgeScale}{$w_{3m-1}$}"] (h6m_12) ->["\scalebox{\edgeScale}{$p_{3m-1}$}"] (h6m_13) ->["\scalebox{\edgeScale}{$k$}"] (h6m_14) ->["\scalebox{\edgeScale}{$y_{3m-1}$}"] (h6m_15)->["\scalebox{\edgeScale}{$w_{3m-1}$}"](h6m_16);

%CONSISTENCY EDGES
(h06)->["\scalebox{\edgeScale}{$c_0$}"] (d1); 
(d1)->["\scalebox{\edgeScale}{$c_{3m-2}$}"] (h3m_16); 
(h3m_16)->["\scalebox{\edgeScale}{$c_{3m-1}$}"] (h3m6); 
(h3m6)->["\scalebox{\edgeScale}{$c_{3m}$}"] (d3); 
(d3)->["\scalebox{\edgeScale}{$c_{6m-2}$}"] (h6m_16); 
};
\end{scope}
\begin{scope}[nodes={set=import nodes},yshift=-8.5cm]% make all nodes part of this set
 	%freezer F^z_0 zeros 1 for sigma_1...sigma_3
		\node (B) at (-0.75,0) {$2)$};
	\foreach \i in {0,...,4} {\coordinate (\i) at (\i*1.5cm,0);}
	\foreach \i in {0} {\fill[red!40, rounded corners] (\i) +(-0.5,-0.25) rectangle +(0.5,0.35);}
	\foreach \i in {3} {\fill[red!40, rounded corners] (\i) +(-2,-0.25) rectangle +(0.5,0.35);}
	\foreach \i in {0,...,4} {\node (n\i) at (\i) {\scalebox{\nodeScale}{$f_{0,\i}$}};}
		
\end{scope}
\graph {
	(import nodes);
			n0 ->["\scalebox{\edgeScale}{$k$}"]n1;
			n1 ->["\scalebox{\edgeScale}{$n_0$}"]n2;
			n2 ->["\scalebox{\edgeScale}{$z_0$}"]n3;
			n3 ->["\scalebox{\edgeScale}{$k$}"]n4;
			};
\begin{scope}[nodes={set=import nodes},yshift=-8.5cm, xshift=7.5cm]% make all nodes part of this set
 	%freezer  F^b 1 for sigma_1...sigma_3
		\node (B) at (-0.75,0) {$3)$};
		\foreach \i in {0,...,2} {\coordinate (\i) at (\i*1.5cm,0);}
		\foreach \i in {1} {\fill[red!40, rounded corners] (\i) +(-2,-0.25) rectangle +(0.5,0.35);}
		\foreach \i in {0,...,2} {\node (m\i) at (\i) {\scalebox{\nodeScale}{$f_{1,\i}$}};}
\end{scope}
\graph {
	(import nodes);
			m0 ->["\scalebox{\edgeScale}{$q_0$}"]m1;
			m1 ->["\scalebox{\edgeScale}{$k$}"]m2;
		};
\begin{scope}[nodes={set=import nodes}, yshift=-9.5cm]% make all nodes part of this set
 	%head H_j for sigma_4/sigma_5
		\node (B) at (-0.75,0) {$4)$};
		\foreach \i in {0,...,5} { \coordinate (\i) at (\i*1.5cm,0) ;}
		\foreach \i in {2,5} {\fill[red!40, rounded corners] (\i) +(-0.5,-0.25) rectangle +(0.5,0.35);}
		\foreach \i in {0,...,5} { \node (\i) at (\i*1.5cm,0) {\scalebox{\nodeScale}{$h'_{j,\i}$}};}
\end{scope}
\graph {
	(import nodes);
			0 <->["\scalebox{\edgeScale}{$k$}"]1<->["\scalebox{\edgeScale}{$m$}"]2 <->["\scalebox{\edgeScale}{$v_j$}"]3 <->["\scalebox{\edgeScale}{$k$}"]4  <->["\scalebox{\edgeScale}{$\_$}"]5;

			};
\begin{scope}[nodes={set=import nodes}, yshift=-10.5cm]% make all nodes part of this set
 	%generate helper q_0
		\node (B) at (-0.75,0) {$5)$};
		\foreach \i in {0,...,8} { \coordinate(\i) at (\i*1.55cm,0);}
		\foreach \i in {2,5,8} {\fill[red!40, rounded corners] (\i) +(-0.5,-0.25) rectangle +(0.5,0.35);}
		\foreach \i in {0,...,8} { \node (\i) at (\i*1.55cm,0) {\scalebox{\nodeScale}{$f'_{0,\i}$}};}
\end{scope}
\graph {
	(import nodes);
			0 <->["\scalebox{\edgeScale}{$k$}"]1<->["\scalebox{\edgeScale}{$m$}"]2 <->["\scalebox{\edgeScale}{$q_0$}"]3 <->["\scalebox{\edgeScale}{$k$}"]4  <->["\scalebox{\edgeScale}{$m$}"]5<->["\scalebox{\edgeScale}{$q_1$}"]6<->["\scalebox{\edgeScale}{$k$}"]7<->["\scalebox{\edgeScale}{$\_$}"]8;
			};
\begin{scope}[nodes={set=import nodes}, yshift=-11.5cm]% make all nodes part of this set
 	%generate helper q_2,q_3
		\node (B) at (-0.75,0) {$6)$};
		\foreach \i in {0,...,5} { \coordinate (\i) at (\i*1.6cm,0) ;}
		\foreach \i in {2,5} {\fill[red!40, rounded corners] (\i) +(-0.5,-0.25) rectangle +(0.5,0.35);}
		\foreach \i in {0,...,5} { \node (\i) at (\i*1.6cm,0) {\scalebox{\nodeScale}{$f'_{1,\i}$}};}
\end{scope}
\graph {
	(import nodes);
			0 <->["\scalebox{\edgeScale}{$k$}"]1<->["\scalebox{\edgeScale}{$q_2$}"]2 <->["\scalebox{\edgeScale}{$q_3$}"]3 <->["\scalebox{\edgeScale}{$k$}"]4  <->["\scalebox{\edgeScale}{$\_$}"]5;
			};
\begin{scope}[nodes={set=import nodes}, yshift=-12.5cm]% make all nodes part of this set
 	%generate z
		\node (B) at (-0.75,0) {$7)$};
		\foreach \i in {0,...,9} { \coordinate (\i) at (\i*1.4cm,0);}
		\foreach \i in {2} {\fill[red!40, rounded corners] (\i) +(-0.5,-0.25) rectangle +(0.5,0.35);}
		\foreach \i in {6} {\fill[red!40, rounded corners] (\i) +(-1.9,-0.25) rectangle +(0.5,0.35);}
		\foreach \i in {9} {\fill[red!40, rounded corners] (\i) +(-0.5,-0.25) rectangle +(0.3,0.35);}
		\foreach \i in {0,...,9} { \node (\i) at (\i*1.4cm,0) {\scalebox{\nodeScale}{$f'_{2,\i}$}};}
\end{scope}
\graph {
	(import nodes);
			0 <->["\scalebox{\edgeScale}{$k$}"]1<->["\scalebox{\edgeScale}{$q_2$}"]2 <->["\scalebox{\edgeScale}{$q_0$}"]3 <->["\scalebox{\edgeScale}{$z$}"]4  <->["\scalebox{\edgeScale}{$q_1$}"]5 <->["\scalebox{\edgeScale}{$z$}"]6 <->["\scalebox{\edgeScale}{$q_3$}"]7 <->["\scalebox{\edgeScale}{$k$}"]8<->["\scalebox{\edgeScale}{$\_$}"]9;
			};

\begin{scope}[nodes={set=import nodes}, yshift =-13.5cm]% make all nodes part of this set
 	%accordance duplicator for sigma_4/sigma_5
		\node (B) at (-0.75,0) {$8)$};
		\foreach \i in {0,...,8} { \coordinate (\i) at (\i*1.5cm,0) ;}
		\foreach \i in {4} {\fill[red!40, rounded corners] (\i) +(-3.5,-0.25) rectangle +(0.5,0.35);}
		\foreach \i in {8} {\fill[red!40, rounded corners] (\i) +(-0.5,-0.25) rectangle +(0.5,0.35);}
		\foreach \i in {0,...,8} { \node (d\i) at (\i) {\scalebox{\nodeScale}{$d_{j,\i}$}};}

\end{scope}
\graph {
	(import nodes);
			d0 <->["\scalebox{\edgeScale}{$k$}"]d1<->["\scalebox{\edgeScale}{$p_j$}"]d2- >["\scalebox{\edgeScale}{$z$}"]d3 <->["\scalebox{\edgeScale}{$z$}"]d4 <->["\scalebox{\edgeScale}{$p_j$}"]d5<->["\scalebox{\edgeScale}{$a_j$}"]d6<->["\scalebox{\edgeScale}{$k$}"]d7<->["\scalebox{\edgeScale}{$\_$}"]d8;
			};
\begin{scope}[nodes={set=import nodes}, yshift =-14.5cm]% make all nodes part of this set
 	%wire generator for sigma_4/sigma_5
		\node (B) at (-0.75,0) {$9)$};
		\foreach \i in {0,...,8} { \coordinate (\i) at (\i*1.5cm,0) ;}
		\foreach \i in {4} {\fill[red!40, rounded corners] (\i) +(-3.5,-0.25) rectangle +(0.5,0.35);}
		\foreach \i in {8} {\fill[red!40, rounded corners] (\i) +(-0.5,-0.25) rectangle +(0.5,0.35);}
		\foreach \i in {0,...,8} { \node (g\i) at (\i) {\scalebox{\nodeScale}{$g_{j,\i}$}};}
\end{scope}
\graph {
	(import nodes);
			g0 <->["\scalebox{\edgeScale}{$k$}"]g1<->["\scalebox{\edgeScale}{$y_j$}"]g2- >["\scalebox{\edgeScale}{$z$}"]g3 <->["\scalebox{\edgeScale}{$z$}"]g4 <->["\scalebox{\edgeScale}{$y_j$}"]g5<->["\scalebox{\edgeScale}{$w_j$}"]g6<->["\scalebox{\edgeScale}{$k$}"]g7<->["\scalebox{\edgeScale}{$\_$}"]g8;
			};	

\end{tikzpicture}
\caption{
The gadget TSs for the key union $K^\sigma_m$ where the red areas mark the support of a key region as defined for Lemma~\ref{lem:key_union}. 
(1) $H$,
(2,3) $F_0$ and $F_1$,
(4) $H'_j$,
(5-7) $F'_0$, $F'_1$ and $F'_2$,
(8) $D_j$,  
(9) $G_j$.
}
\label{fig:key_unions}
\end{figure}

The following lemma establishes the interface compatibility of all key-regions with all inhibitor regions as demanded in Lemma~\ref{lem:translators}.
Moreover, it shows the existence of a key region that is even compatible with the regions from Lemma~\ref{lem:indicator_regions}.
\begin{lemma}\label{lem:key_union}
Let $\sigma \in \{\sigma_1, \dots, \sigma_6\}$ and $\tau \in \sigma$.
If $(sup_K,sig_K)$ is a $\tau$-key region of $K^\sigma_m$, that is, where $k$ is inhibited at the key state, then
\begin{enumerate}
\item\label{lem:key_union_1}  
either $sig_K(k)=\inp$, $V \subseteq sig^{-1}_K(\enter)$ and $W \subseteq sig^{-1}_K(\keepze)$ or $sig_K(k)=\out$, $V \subseteq sig^{-1}_K(\exit)$ and $W \subseteq sig^{-1}_K(\keepo)$ in case of $\sigma_1, \dots, \sigma_4$ and 
\item\label{lem:key_union_2} 
$sig_K(k) \in \{\used,\free\}$, $Acc \cap sig^{-1}_K(\swap) = W \cap sig^{-1}_K(\swap) =\emptyset$, $V\subseteq sig^{-1}_K(\swap)$ and $sig_K(q_2)=\swap$ if and only if $sig_K(q_3)=\swap$ for $\sigma \in \{\sigma_5, \sigma_6\}$.
\end{enumerate}
Furthermore, we can always create a $\tau$-key region $(sup_K, sig_K)$ for $K^\sigma_m$ with 
\begin{enumerate}
\item\label{lem:construction_key_union_1}  
$sig_K(k) = \inp$, $V \subseteq sig^{-1}_K(\textsf{op}^\sigma_V)$ and $W \subseteq sig^{-1}_K(\nop)$ if $\sigma \in \{\sigma_1,\dots, \sigma_4\}$ or
\item\label{lem:construction_key_union_2}  
$sig_K(k) = \textsf{op}^\sigma_k$, $(Acc \cup W) \subseteq sig_K^{-1}(\nop)$ and $(V \cup \{q_2,q_3\}) \subseteq sig_K^{-1}(\swap)$, otherwise.
\end{enumerate}
\end{lemma}

The proof of Lemma~\ref{lem:key_union} is again very technical and, thus, can be found in Section~\ref{sec:techproofs}.
The following lemma connects the functionalities of $T^\sigma_\varphi$ and $K^\sigma_m$ for all $\sigma\in \{\sigma_1, \dots, \sigma_6\}$:
\begin{lemma}\label{lem:synergy}
If $\varphi$ is a cubic monotone boolean $3$-CNF with $m$ clauses, $\sigma \in \{\sigma_1, \dots, \sigma_6\}$ and $\tau\in \sigma$ then $U^\sigma_\varphi$ is $\tau$-feasible if and only if $\varphi$ has a one-in-three model.
\end{lemma}
\begin{proof}
\textit{If:} If $\varphi$ has a one-in-three model, then using the corresponding region $(sup^\sigma,sig^\sigma)$ of $T^\sigma_\varphi$ defined in Lemma~\ref{lem:indicator_regions} and the key region $(sup^\sigma_K, sig^\sigma_K)$ of $K^\sigma_\varphi$ introduced in Lemma~\ref{lem:key_union} yields a combined region $R$ of $U^\sigma_\varphi$ that inhibits the key event $k$ at $s_{key}$.
Section~\ref{sec:secondary_proofs} presents a series of lemmas that altogether prove that the existence of $R$ implies the (E)SSP of $U^\sigma_\varphi$.

\textit{Only-if:} If $U^\sigma_\varphi$ is feasible then the key event is inhibitable at the key state in $U^\sigma_\varphi$. 
Projecting the respective region to $T^\sigma_\varphi$, we get an indicator due to Lemmas~\ref{lem:translators} and~\ref{lem:key_union}.
Hence, $\varphi$ is one-in-three satisfiable.
\end{proof}

Using Lemma~\ref{lem:union_validity}, Lemma~\ref{lem:synergy}, the observation that our construction is in polynomial time, the fact that feasibility is in NP, we have shown Theorem~\ref{the:main_result}.

\section{NP-completeness of Feasibility for seven more Petri Net Classes}\label{sec:npc_cases}

This section presents and proves the following theorem:
\begin{theorem}\label{the:hardness_result} 
Deciding $\tau$-feasibility as well as $\tau$-language viability is NP-complete for 
\begin{enumerate}
\item\label{the:hardness_result_nop_inp_free}
modest TSs and $\tau = \{\nop, \inp, \free\}$ or $\tau = \{\nop, \inp, \used, \free\}$, %2
\item\label{the:hardness_result_nop_out_used}
modest TSs and $\tau = \{\nop, \out, \used\}$ or  $\tau = \{\nop, \out, \used, \free\}$, and %2
\item\label{the:hardness_result_nop_set_res}
general TSs and $\tau = \{\nop, \set, \res \} \cup \omega$ with non-empty $\omega \subseteq \{\used, \free\}$. %3
\end{enumerate}
\end{theorem}

Notice that Theorem~\ref{the:hardness_result}.\ref{the:hardness_result_nop_set_res} does not restrict input to modest TSs.
Otherwise, it easy to show that feasibility becomes tractable if $\tau$ belongs to these classes: 
Assume a modest TS $A$ contains a transition $s\edge{e}s'$.
If $\neg s' \edge{e}$, then $A$ is not feasible because $e$ is not inhibitable at $s'$.
Hence, we can assume $s' \edge{e}s''$ for some $s''\in S(A)$. 
If $s=s''$, then $s,s'$ are not separable, if $s\not=s''$ then $s',s''$ are not separable.
Consequently, a modest TS $A$ is feasible if and only if it consists of at most one state.

We again present polynomial time reductions of the NP-complete cubic monotone one-in-three $3$-SAT problem \cite{MR2001} to the corresponding feasibility problems making sure for every TS $A^{\tau}_\varphi$  constructed from a given cubic monotone $3$-CNF $\varphi$ that the ESSP implies the SSP.
By Lemma~\ref{lem:isomorphic_types} the proofs for Theorem~\ref{the:hardness_result}.\ref{the:hardness_result_nop_inp_free} and ~\ref{the:hardness_result}.\ref{the:hardness_result_nop_out_used} are the same as the respective types of nets are isomorphic.

In every case, we again install a \emph{key event} $k$ and a \emph{key state} $q$ in $A^{\tau}_\varphi$ such that $k$ is inhibitable at $q$ by a key region $(sup, sig)$ if and only if a one-in-three model $M$ exists.
Like before, the variables $V(\varphi)$ are used as events in $A^{\tau}_\varphi$ and their key signature $sig$ tells us how to find $M$ and vice versa.

This idea is put into practice by creating six directed labeled paths per clause $C_i=\{X_{i,0},X_{i,1},X_{i,2}\}$ that commonly start at state $t_{i,0}$, terminate at $t_{i,5}$ and consist of three transitions permuting the events $X_{i,0}, X_{i,1}, X_{i,2}$.
The TS $A^{\tau}_\varphi$ fulfills the following conditions:
\begin{enumerate}
\item
For every $i \in \{0, \dots, m-1\}$ and every permutation $(\alpha, \beta, \gamma)$ of $\{0,1,2\}$ there is a path $t_{i,0} \edge{X_{i,\alpha}} t \edge{X_{i,\beta}} t' \edge{X_{i,\gamma}} t_{i,5}$ in $A^{\tau}_\varphi$.
\item
If $(sup,sig)$ is a key $\tau$-region of $A^{\tau}_\varphi$, that is, one that inhibits $k$ at $q$, then  $sup(t_{0,0}) = \dots = sup(t_{m-1,0}) \not= sup(t_{0,5}) = \dots = sup(t_{m-1,5})$.
\item 
If $i \in \{0, \dots, m-1\}$ and $(sup,sig)$ is a key $\tau$-region of $A^{\tau}_\varphi$ then \emph{exactly one} of $sig(X_{i,\alpha})$, $sig(X_{i,\beta})$, $sig(X_{i,\gamma})$ is different from $\nop$.
Hence, the signature tells us how to build $M=\{X \in V(\varphi) \mid sig(X) \not=\nop\}$, a one-in-three model of $\varphi$.
\item 
If $\varphi$ has a one-in-three model then $(A^{\tau}_\varphi, k, q)$ is solvable.
\item 
If $(A^{\tau}_\varphi, k, q)$ is solvable, then $A^{\tau}_\varphi$ has the $\tau$-E(SSP).
\end{enumerate}
Clearly, having these conditions proves that there is a one-in-three model for $\varphi$ if and only if $A^{\tau}_\varphi$ has the E(SSP).

Next, we define how $A^{\tau}_\varphi$ is constructed from $\varphi$.
See Figure \ref{fig:reduction} for a visualization of the following concepts.
Firstly, we call $A_\varphi$ the basic TS with states $S = \{s_0, s_1, q\} \cup \{t_{i,0},\dots, t_{i,8} \mid 0\leq i \leq m-1\}$ and events $E = \{k, h\} \cup \{h_i, r_i \mid 0 \leq i \leq m-1\} \cup V(\varphi)$.
To omit a lengthy and complex definition of the transitions in $A_\varphi$, we use Figure~\ref{fig:reduction}, which depicts $\delta(s, e)$ with black arcs for all states $s \in \{s_0, s_1, q, t_{i,0}, \dots, t_{i,8}\}$ and all events $e \in \{k, h, h_i, r_i, X_{i,0}, X_{i,1}, X_{i,2}\}$.
If $\tau = \{\nop, \res, \free\}$ or $\tau = \{\nop, \res, \used, \free\}$ then we simply use the basic TS, that is, $A^{\tau}_\varphi = A_\varphi$.

If $\tau$ is one of $\{\set, \res, \used\}$, $\{\set, \res, \free\}$, or $\{\set, \res, \used, \free\}$, the construction of $A^{\tau}_\varphi$ is more complex.
We first require the extended TS $A^+_\varphi$ with extended states $S(A^{+}_\varphi) = S(A_\varphi) \cup \{ m_0,\dots, m_4\} \cup \{p_{i,0}, \dots, p_{i,3} \mid 0 \leq i < m\}$ and extended events $E(A^{+}_\varphi) = E(A_\varphi) \cup \{a,c,u,v\} \cup\{a_i,b_i, x_i \mid 0 \le i < m\}$.
The transitions of $A^{+}_\varphi$ are also an extension in the way that $\delta(A^{+}_\varphi)(s,e) = \delta(A_\varphi)(s,e)$ for basic states $s \in S(A_\varphi)$ and basic events $e \in E(A_\varphi)$ where $\delta(A_\varphi)(s,e)$ is defined.
Thus, the black arcs in Figure~\ref{fig:reduction} illustrate part of the extended transitions.
Aside from this, the brown arcs present the remaining transition function $\delta(A^{+}_\varphi)(s,e)$ for all $s \in \{s_0, t_{i,0}, \dots, t_{i,8}, m_0, \dots, m_4, p_{i,0}, \dots, p_{i,3}\}$ and all $e \in \{k, h, h_i, a, c, u, v, a_i, b_i, x_{i,0}, x_{i,1}, x_{i,2}\}$.

While still being depictable, $A^{+}_\varphi$ is not yet a complete TS $A^{\tau}_\varphi$ for extensions of toggle nets.
This requires the \emph{loop-enhancement} $A^{\times}_\varphi$ of $A^{+}_\varphi$ on the same states $S(A^{\times}_\varphi) = S(A^{+}_\varphi)$ and events $E(A^{\times}_\varphi) = E(A^{+}_\varphi)$ but with loop-enhanced transitions, that is, for all $s, s' \in S(A^{+}_\varphi)$ and $e \in E(A^{+}_\varphi)$ where $\delta(A^{+}_\varphi)(s,e) = s'$ we have $\delta(A^{\times}_\varphi)(s,e) = s'$ and $\delta(A^{\times}_\varphi)(s',e) = s'$.
Now, $A^{\tau}_\varphi = A^{\times}_\varphi$.
However, for understanding it is mostly better to deal with $A^{+}_\varphi$ instead of the complicated $A^{\times}_\varphi$.
Therefore, Figure~\ref{fig:reduction} desists from showing all the loops.
%%%Figure%
\begin{figure}[t]
\centering
\begin{tikzpicture}[new set = import nodes, ]
\begin{scope}[nodes={set=import nodes},xshift=-5.7cm, rotate=-90, scale =0.87 ]% make all nodes part of this set
 	%Translatorpart T_0
		
		\coordinate (c0) at (-2.5,5.5) ;	%0
		\coordinate (c1) at (-2.5,3.5) ;	%1
		\coordinate (c2) at (-2.5,1.5) ;	%2
		\coordinate(c3) at (0,5.5) ;		%3
		\coordinate (c4) at (0,3.5) ;	%4
		\coordinate (c5) at (0,1.5) ;	%5
		\coordinate (c6) at (2.5,5.5);	%6
		\coordinate(c7) at (-5,5.5) ;	%7
		\coordinate(c8) at (-3.5,8) ;	%8
		%colored regions
		\foreach \i in {c0} {\fill[red!15, rounded corners] (\i) +(-3.2,-4.5) rectangle +(0.9,2.8);}
		\foreach \i in {c7} {\fill[blue!15, rounded corners] (\i) +(-.25,-0.45) rectangle +(3,0.4);}
		\foreach \i in {c0} {\fill[blue!15, rounded corners] (\i) +(-0.35,-4.3) rectangle +(0.7,0.4);}
		
		\node (0) at (-2.5,5.5) {\nscale{$t_{i,0}$}};
		\node (1) at (-2.5,3.5) {\nscale{$t_{i,1}$}};	
		\node (2) at (-2.5,1.5) {\nscale{$t_{i,2}$}};
		\node (3) at (0,5.5) {\nscale{$t_{i,3}$}};
		\node (4) at (0,3.5) {\nscale{$t_{i,4}$}};
		\node (5) at (0,1.5) {\nscale{$t_{i,5}$}};
		\node (6) at (2.5,5.5) {\nscale{$t_{i,6}$}};
		\node (7) at (-5,5.5) {\nscale{$t_{i,7}$}};
		\node (8) at (-3.5,8) {\nscale{$t_{i,8}$}};
		%
		%\node (d0) at (1,0) {\scalebox{1.5}{$\vdots$}};
\end{scope}
\begin{scope}[nodes={set=import nodes},scale =0.95 ]
		\node (dummy1) at (-6.125, 1.5) { $\vdots$};
		\node (dummy2) at (-6.125, -1.5) { $\vdots$};
		\node (dummy3) at (-9.25,0) {\textcolor{brown}{$\dots$}};
		\node (dummy4) at (-7.5,-1.75) {\textcolor{brown}{$\vdots$}};
		\end{scope}
\begin{scope}[nodes={set=import nodes},scale =0.95 ]
		%central state
		\node (q) at (-6.125, 0) {$q$};
	
		%key and h
		\coordinate(cs0) at (-7.5,0);
		\foreach \i in {cs0} {\fill[red!15, rounded corners] (\i) +(0.5,2.25) rectangle +(-0.5,-0.5);}
		\foreach \i in {cs0} {\fill[blue!15, rounded corners] (\i) +(0.4,0.5) rectangle +(-0.4,-0.4);}
		
		\node (s0) at (-7.5,0) {\nscale{$s_0$}};
		\node (s1) at (-7.5,2) {\nscale{$s_1$}};
		%%%%%Master
		\coordinate (cm0) at (-8.5,3);
		\coordinate(cm1) at (-10,3) ;
		\coordinate(cm2) at (-10,4.5);
		\coordinate(cm3) at (-8.5,4.5);
		\coordinate (cm4) at (-7.5,3.75);
		\foreach \i in {cm0} {\fill[red!15, rounded corners] (\i) +(0.5,1.9) rectangle +(-2,-0.4);}
		%\foreach \i in {cm0} {\fill[blue!15, rounded corners] (\i) +(0.4,1.75) rectangle +(-0.5,-0.25);}
		%
		\node (m0) at (-8.5,3) {\nscale{\textcolor{brown}{$m_0$}}};
		\node (m1) at (-10,3) {\nscale{\textcolor{brown}{$m_1$}}};
		\node (m2) at (-10,4.5) {\nscale{\textcolor{brown}{$m_2$}}};
		\node (m3) at (-8.5,4.5) {\nscale{\textcolor{brown}{$m_3$}}};
		\node (m4) at (-7.5,3.75) {\nscale{\textcolor{brown}{$m_4$}}};
		%%%%%
		\node (p10) at (-9,-1) {\nscale{\textcolor{brown}{$p_{i,0}$}}};
		\node (p11) at (-10.5,-1) {\nscale{\textcolor{brown}{$p_{i,1}$}}};
		\node (p12) at (-10.5,-2.5) {\nscale{\textcolor{brown}{$p_{i,2}$}}};
		\node (p13) at (-9,-2.5) {\nscale{\textcolor{brown}{$p_{i,3}$}}};
\end{scope}
\graph {
	(import nodes);
			%%%%%%%key and reachability and h
			s0 ->[thick,"$\escale{$k$}$"]s1;
			s0 ->[thick,"$\escale{$h$}$"]q;
			s0 ->[pos=0.3 ,thick, dotted, bend left=65, "$\escale{$r_0$}$"]0;
			s1 ->[pos=0.3 , thick, dotted , bend left=65, "$\escale{$r_0$}$"]8;
			%
			%
			%s0 ->[pos=0.3 ,thick, dotted, bend right=25, "$\escale{$r_{m-1}$}$"]10;
			%s1 ->[pos=0.5 ,thick, dotted , bend right=35, "$\escale{$r_{m-1}$}$"]18;
			%%%Master%%%%
			m0 ->[ color =brown , "$\escale{$k$}$"]m1;
			m0 ->[ color =brown , "$\escale{$c$}$"]m3;
			m0 ->[ color =brown , swap, "$\escale{$u$}$"]m4;
			m1 ->[color =brown ,"$\escale{$v$}$"]m2;
			m3 ->[color =brown, swap, "$\escale{$k$}$"]m2;
			m3 ->[color =brown,  "$\escale{$h$}$"]m4;
			s0 ->[color =brown,  dotted, thick, bend left= 20, "$\escale{$a$}$"]m0;
			%%reset maker P_0
			%p0 ->[ color =brown , "$\escale{$v$}$"]p1;
			%p0 ->[ color =brown ,swap,  "$\escale{$h_0$}$"]p3;
			%p1 ->[ color =brown , "$\escale{$b_0$}$"]p2;
			%p2 ->[color =brown,"$\escale{$u$}$"]p3;
			%s0 ->[color =brown, swap, dotted, thick, bend right= 10, "$\escale{$a_0$}$"]p0;
			%%reset maker P_{m-1}
			p10 ->[ color =brown , swap, "$\escale{$v$}$"]p11;
			p10 ->[ color =brown , "$\escale{$h_{i}$}$"]p13;
			p11 ->[ color =brown , swap, "$\escale{$b_{i}$}$"]p12;
			p12 ->[color =brown, swap, "$\escale{$u$}$"]p13;
			s0 ->[color =brown,swap,  dotted, thick, "$\escale{$a_0$}$"]dummy3;
			s0 ->[color =brown, dotted, thick, "$\escale{$a_i$}$"]p10;
			s0 ->[color =brown, dotted, thick, "$\escale{$a_{m-1}$}$"]dummy4;
			%%%%%%%Clauses
			q ->[thick,swap, "$\escale{$h_0$}$"]dummy1;
			q ->[thick,"$\escale{$h_i$}$"]5;
			q ->[thick,"$\escale{$h_{m-1}$}$"]dummy2;
			%q ->[thick, "$\escale{$h_{m-1}$}$"]15;
			%%%%%T_0%%%%%%
			0 ->[thick, pos=0.3,"$\escale{$k$}$"]8;
			0 ->[thick, bend left = 20, "\escale{$X_{i,0}$}"]1;
			0->[thick, bend right = 20, swap, "\escale{$X_{i,2}$}"]3;
			0->[thick, bend left = 20, "\escale{$X_{i,1}$}"]7;
			1 ->[thick, bend left  = 20,  "\escale{$X_{i,1}$}"]2;
			1->[ thick, bend right = 20, swap, "\escale{$X_{i,2}$}"]4;
			2->[thick, bend right = 20, swap, "\escale{$X_{i,2}$}"]5;
			3 ->[thick, bend left = 20, "$\escale{$X_{i,0}$}$"]4;
			3 ->[thick, bend right = 20, swap, "$\escale{$X_{i,1}$}$"]6;
			4 ->[thick, bend left = 20, "\escale{$X_{i,1}$}"]5;
			6 ->[thick, bend left=30,"$\escale{$X_{i,0}$}$"]5;
			7->[thick, swap, bend right=15,swap ,"\escale{$X_{i,0}$}"]2;
			7->[thick, bend left =40, swap,"\escale{$X_{i,2}$}"]6;
			%%%%%%%%%backwards
			0 <-[color=brown, bend right= 20, swap, "\escale{$x_{i,0}$}"]1;
			0<-[color=brown, bend left = 20, "\escale{$x_{i,2}$}"]3;
			0<-[color=brown, bend right = 20,swap, "\escale{$x_{i,1}$}"]7;
			1 <-[color=brown, bend right= 20, swap, "\escale{$x_{i,1}$}"]2;
			1<-[color=brown,  bend left= 20, "\escale{$x_{i,2}$}"]4;
			2<-[color=brown, bend left= 20,"\escale{$x_{i,2}$}"]5;
			3 <-[color=brown, bend right= 20,swap, "$\escale{$x_{i,0}$}$"]4;
			3 <-[color=brown, bend left= 20,"$\escale{$x_{i,1}$}$"]6;
			4 <-[color=brown, bend right= 20,swap, "\escale{$x_{i,1}$}"]5;
			6 <-[color=brown, bend right= 20,bend left=15,swap, "$\escale{$x_{i,0}$}$"]5;
			7<-[color=brown, bend right = 30,swap,"\escale{$x_{i,0}$}"]2;
			7<-[color=brown, bend left =60,"\escale{$x_{i,2}$}"]6;
%%%%%%%T_{m-1}%%%%%%%%%%%%%%%%%%%%%%%%%%%%%%%%%%%%%%%%%%%%%%%%%%%%%%%%%%%%%%%%%%%%%%%%%%%%%%%%%%%%%%%%%%%%%%%%%%%%%	
};
\end{tikzpicture}
\caption{
The black arcs in isolation illustrate $A_\varphi$.
Aside from the static center on states $s_0, s_1, q$, TS $A_\varphi$ contains a compartment of states $t_{i,0}, \dots, t_{i,8}$ for every clause $C_i$ of $\varphi$.
Together with the brown arcs, we get $A^{+}_\varphi$ which adds one compartment on $m_0, \dots, m_4$ and one on $p_{i,0}, \dots, p_{i,3}$ for every clause $C_i$.
Restricted to the presented parts of $A_\varphi$, respectively $A^{+}_\varphi$, the blue areas, respectively red areas, mark the support of the region that inhibits $k$ at $q$ defined in Lemma~\ref{lem:key_region}.
The dotted arc do not add to the mechanism except for making sure that all states are reachable form the initial state.
}
\label{fig:reduction}
\end{figure}
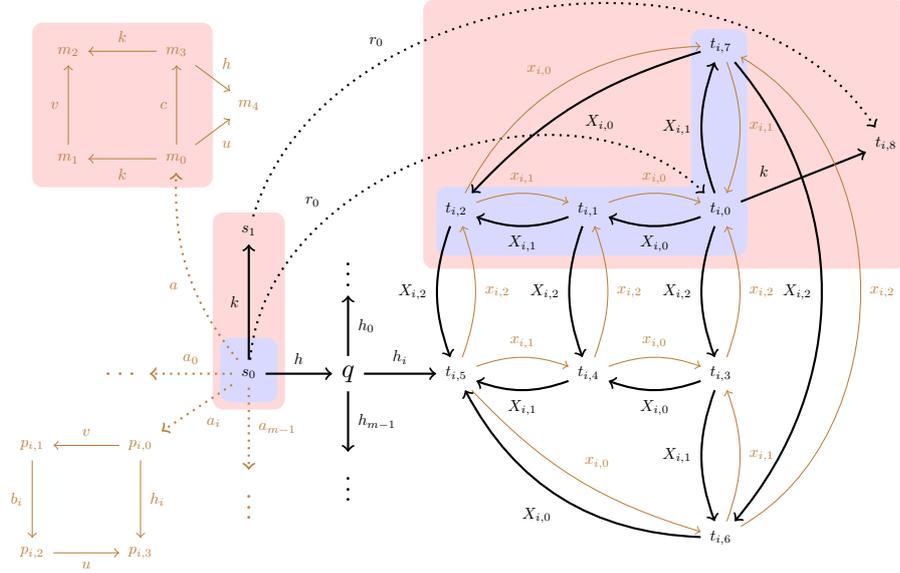

At this point, we are ready to provide the main piece of our proof.
The next lemma shows the equivalence between the one-in-three satisfiability of $\varphi$ and the inhibitability of $k$ at $q$:
\begin{lemma}\label{lem:key_region}
\label{lem:key_nop_inp_free[used]}
If $\tau$ is $\{\nop, \inp, \free\}$ or $\{\nop, \inp, \used, \free\}$ or $\{\nop, \set, \res\} \cup \omega$ with non-empty $\omega \in \{\used, \free\}$ then the key event $k$ is $\tau$-inhibitable at the key state $q$ in $A^{\tau}_\varphi$ if and only if $\varphi$ is one-in-three satisfiable.
\end{lemma}
\begin{proof}
\emph{Only-if}:
Let $(sup, sig)$ be a $\tau$-region inhibiting $k$ at $q$ in $A^{\tau}_\varphi$.
We show for every clause $C_i$ that there is exactly one variable event $X \in \{X_{i,0}, X_{i,1}, X_{i,2}\}$ with $sig(X) = \{\inp, \used, \free\}$ while the other two have \nop-signature.
Consequently, the set $M=\{X \in V(\varphi) \mid sig(X) \not = \nop\}$ will be a one-in-three model of $\varphi$.

For a start, $\tau = \{\nop, \inp, \free\}$ or $\tau = \{\nop, \inp, \used, \free\}$.
As $k$ is inhibited at $q$, assume first that $sig(k) = \free$ and $sup(q)=1$.
But then $sup(s_0) = sup(s_1) = 0$ which means that $sig(h) \not \in \tau$. 
Hence, $sig(k) \in \{\inp, \used\}$ and $sup(q)=0$.
This implies $sup(t_{i,0})=1$ and $sig(h_i) \in \{\nop,\free\}$ and, thus, $sup(t_{i,5})=0$.
By this, we get $sig(X_{i,0}), sig(X_{i,1}), sig(X_{i,2}) \in \{\nop, \inp, \free\}$.
In $A^\tau_\varphi$, there is a path $t_{i,0} \edge{X_{i,\alpha}} t_1\edge{X_{i,\beta}} t_2 \edge{X_{i,\gamma}} t_{i,5}$ for every permutation $(\alpha, \beta, \gamma)$ of $\{0,1,2\}$.
As all variable events of the clause occur exactly once on each of these six paths, there has to be exactly one $X \in \{X_{i,0},X_{i,1},X_{i,2}\}$ with $sig(X) = \inp$.
Moreover, for every $j \in \{0, 1, 2\}$ there are both, a path that starts with $X_{i,j}$ and another that ends on $X_{i,j}$.
Consequently, for $Y \in \{X_{i,0}, X_{i,1}, X_{i,2}\} \setminus X$ there are transitions $s \edge{Y} s'$ and $z \edge{Y} z'$ with $sup(s) = sup(s') \not= sup(z)=sup(z')$ implying $sig(Y)=\nop$.

Next, let $\tau = \{\nop, \set, \res\} \cup \omega$ with non-empty $\omega \in \{\used, \free\}$.
Assume $sig(k)=\used$ and $sup(q)=0$, which implies $sup(s_0) = sup(m_0) = sup(t_{i,0}) = 1$ and $sig(h) = \res$ and, thus, $sup(m_4)=0$.
We immediately get $sig(u)=\res$ implying $sup(p_{0,3})=0$ and $sig(h_i)\in \{\nop, \res,\free\}$ and, thus, $sup(t_{i,5})=0$. 
For every $j\in \{0,1,2\}$, we have $sig(X_{i,j}) \in \{\nop, \res, \free\}$ by $\edge{X_{i,j}} t_{i,5}$.
Accordingly, $\edge{x_{i,j}} t_{i,0}$ leads to $sig(x_{i,j}) \in \{\nop, \set, \used\}$.
As $sup(t_{i,0}) \not= sup(t_{i,5})$, there must be at least one event in $X_{i,0}, X_{i,1}, X_{i,2}$, respectively $x_{i,0}, x_{i,1}, x_{i,2}$, with signature not in $\{\nop, \free\}$, respectively $\{\nop, \used\}$.
We show that there must be exactly one $j$ with $sig(X_{i,j}) = \res$ and $sig(x_{i,j}) = \set$, while the other four events have \nop-signature.
For example, if $sig(X_{i,2})=\res$ then $sup(t_{i,3})=0$ and, thus, $sig(x_{i,0})=sig(x_{i,1})=\nop$ and $sig(x_{i,2}) = \set$.
This implies $sup(t_{i,1}) = sup(t_{i,2})=1, sup(t_{i,4}) =0 $ and, consequently, $sig(X_{i,0}) = sig(X_{i,1})=\nop$.
Notice that none of the events can get $\used$ or $\free$.
A similar explanation works for $j \in \{0,1\}$.

The case $sig(e)=\free$ and $sup(q)=1$, follows by the same argumentation by inverting the support and interchanging $\set$ with $\res$ and $\used$ with $\free$.

\emph{If}:
Let $M \subseteq V(\varphi)$ be a one-in-three model of $\varphi$ and, for a start, let $\tau = \{\nop, \inp, \free\}$ or $\tau = \{\nop, \inp, \used, \free\}$.
We define a $\tau$-region $(sup, sig)$ of $A^{\tau}_\varphi$ that inhibts $k$ at $q$ by $sup = \{s_0\} \cup \{s \mid s \edge{X}: X \in M\}$ and $sig(k) = \inp$, $sig(X) = \inp$ for all $X \in M$ and $sig(e) = \nop$ for all other events $e$ in $E(A^{\tau}_\varphi)$.

If $\tau$ contains, beside $\{\nop, \set, \res\}$, the interaction $\used$ then we let $M'=\{x_i \mid X_i \in M\}$ and create a $\tau$-region $(sup', sig')$ of $A^{\tau}_\varphi$ to inhibit $k$ at $q$ by the support $sup' = sup \cup \{s_1, t_{0,8},\dots, t_{m-1,8}, m_0, \dots, m_3\}$ and $sig'(k) = \used$, $sig'(X) = \res$ for all $X \in M$, $sig'(x) = \set$ for all $x \in M'$ and $sig'(e) = \nop$ for all other events $e$ in $E(A^{\tau}_\varphi)$.
If $\used$ is not in $\tau$ then $\free$ is and we get a $\tau$-region by inverting the support and interchanging $\set$ with $\res$ and $\used$ with $\free$.
\end{proof}

Using Lemma \ref{lem:isomorphic_types}, we have also covered the types of nets where $\tau$ is $\{\nop, \out, \used\}$ or $\{\nop, \out, \used, \free\}$.
Using the same lemma, we can finish our proof for Theorem \ref{the:hardness_result} by the following lemma:
\begin{lemma}\label{lem:sec_sep_atoms}
\label{lem:key_nop_inp_free[used]_ESSP}
If $\tau$ is $\{\nop, \inp, \free\}$ or $\{\nop, \inp, \used, \free\}$ or $\{\nop, \set, \res\} \cup \omega$ with non-empty $\omega \in \{\used, \free\}$ and the key event $k$ is $\tau$-inhibitable at the key state $q$ in $A^{\tau}_\varphi$ then $A^{\tau}_\varphi$ has the E(SSP).
\end{lemma}
\begin{proof}
We present for each event $e\in E(A^{\tau}_\varphi)$, respectively state $s\in S(A^{\tau}_\varphi)$), a set of regions which inhibit $e$ at, respectively separate $s$ from, all states in $A^{\tau}_\varphi$.

Assume that $\tau = \{\nop, \inp, \free\}$ or $\tau = \{\nop, \inp, \used, \free\}$.
For brevity, we use the following scheme to define $sig(e)$ based on a given support $sup$:
If $sup(s)=1$ and $sup(s')=0$ for all $s\edge{e}s'$ then $sig(e)= \inp$, if $sup(s)=sup(s')=0$ for all $s\edge{e}s'$ then $sig(e) = \free$ and, otherwise, $sig(e)=\nop$.
Using this, we can define a region simply by defining $sup$.

The inhibition of $k$ at $q$ and $h$ at $s_1$ already follows from Lemma~\ref{lem:key_region}.
Furthermore, if $(sup,sig)$ inhibits $k$ at $q$ then the region $sup \cup \{q\}$ inhibits $r_0, \dots, r_{m-1}$ at $s_1$.
For $X \in V(\varphi)$ let $sup_X = \{s \mid s\edge{k}\} \cup \{s\mid s\edge{X}\}\cup \{q\}\cup \{t_{n,0},\dots, t_{n,7}\mid 0\le n< m, X\not\in C_n\}$.
The region $sup_{X}$ inhibits $X$ at all $S_X=\{t_{n,1},\dots, t_{n,8}\mid 0\le n< m, X\in C_n\}$.
Moreover, $sup'_X = S(A^{\tau}_\varphi) \setminus S_X$ inhibits $X$ at the states of $sup'_X$, where the signature of $X$ is $\free$.
The regions $sup_{X_0}, \dots, sup_{X_m}$ also complete the inhibition of $k$ as for every $t \in \{t_{i,1},\dots, t_{i,8}\}$ there is $X \in \{X_{i,0}, X_{i,1}, X_{i,2}\}$ with $\neg (t \edge{X})$.
The regions $\{s_0, s_1\}$ and $\{s_0, s_1, q\}$ complete the inhibition of $h, r_0, \dots, r_{m-1}$.

The inhibition of $k$ at $q$ separates $s_0$ and $s_1$ and the region $\{s_0, s_1\}$ separates $s_0, s_1$ from all the other states.
The region $\{s_0, s_1, s_2\}$ completes the separation of $q$.
The regions $sup_{X_{i,0}}, sup_{X_{i,1}}, sup_{X_{i,2}}$ complete the separation of $t_{i,0},\dots, t_{i,8}$.

Now assume $\tau = \{\nop, \set, \res\} \cup \omega$ with non-empty $\omega \in \{\used, \free\}$.
For the $\tau$-ESSP of $A^{\times}_\varphi$, it is sufficient to prove the inhibition of $e \in E(A^{+}_\varphi)$ at states $s\in S(A^{+}_\varphi)$ where $\neg s\edge{e}$ and $\neg \edge{e}s$.
By definition, if $\edge{e}s$ in $A^{+}_\varphi$ then $s\edge{e}s$ in $A^{\times}_\varphi$ and $e$ does not need to be inhibitable at $s$.
We proceed like in the previous case and use the following scheme to define a signature $sig(e)$ for given support $sup$:
If $sup(s) = sup(s') = 1$ for all $s\edge{e}s'$ then $sig(e) = \used$, if $sup(s)=1$ and $sup(s')=0$ for all $s\edge{e}s'$ then $sig(e) = \res$, if $sup(s)=0$ and $sup(s')=1$ then $sig(e) = \set$ and, otherwise, $sig(e) = \nop$.

The inhibition of $k$ at $q, m_4, t_{0,5}, t_{1,5}, \dots, t_{m-1,5}, p_{0,0}, \dots, p_{m-1,3}$ is done by the key region of Lemma~\ref{lem:key_region}.
For $X \in V(\varphi)$ let $sup_X = \{s, s' \mid s\edge{k}s'\} \cup \{s\mid s\edge{X}\}\cup \{q,m_4\}\cup \{t_{n,0},\dots, t_{n,7}\mid 0\leq n< m, X\not\in C_n\}$.
The regions $sup_{X_0}, \dots, sup_{X_{m-1}}$ complete the inhibition of $k$ in $A^{\tau}_\varphi$ because for every state $t \in \{t_{i,1},\dots, t_{i,7}, t_{i,8}\}$ there is an event $X \in \{X_{i,0},X_{i,1}, X_{i,2}\}$ such that $\neg (t\edge{X})$.
Moreover, they prove $r_0, \dots, r_{m-1}$ inhibitable at all states except for $m_0,\dots, m_3$. 

For every $X \in \{X_{i,0},X_{i,1},X_{i,2}\}$ ($x\in \{x_{i,0},x_{i,1},x_{i,2}\}$) and all $t \in \{t_{i,0},\dots, t_{i,7}\}$ with $\neg t\edge{X}$ ($\neg t\edge{x}$) we immediately get $\edge{X}t$ ($\edge{x}t$).
Hence, for $X\in V(\varphi)$ the regions $sup^1_{X,x} = \{t_{n,0},\dots, t_{n,7}\mid 0\le n <  m, X \in C_n\}\cup \{p_{n,0}, \dots, p_{n,3} \mid \edge{X}s\bedge{h_n}\}$ and $sup^2_{X,x}=S(A^{\tau}_\varphi)\setminus \{m_0,\dots, m_4, p_{0,0},\dots, p_{m-1,3}\}$ inhibit $X$ and its corresponding event $x$ in $A^{\tau}_\varphi$.
Moreover, $sup^2_{X,x}$ inhibits $r_0,\dots, r_{m-1}$ at $m_0,\dots, m_4$ proving these events to be inhibitable, too.
The regions $sup^1_h=S(A^{\tau}_\varphi) \setminus (\{m_0,m_1,m_2\} \cup \{s \mid \edge{k}s\})$ and $sup^2_h=\{s_0,s_1, q, m_0,\dots, m_4,\}$ inhibit $h$.
The region $sup^1_{b_n}=S(A^{\tau}_\varphi)\setminus \{m_3, m_4,p_{0,0},p_{0,3},\dots,p_{m-1,0},p_{m-1,3} \}$ and the region $sup^1_{b_n}=\{p_{n,0},\dots, p_{n,3}\}$ inhibit $b_n$.
Furthermore, the region $\{s_0,p_{n,0}\}$ inhibits $a_n$.
The regions $sup^1_u=\{s_0,q,m_0,m_4\}\cup \{t_{n,0},\dots, t_{n,7}, p_{i,2},p_{n,3}\mid 0 \leq n < m\}$ and $sup^2_u=\{p_{n,0},\dots, p_{n,3}\mid 0 \leq n < m\}$ settle the inhibition of $u$.
The regions $sup^1_v=\{m_1,m_2\} \cup \{s \mid \edge{k}s \}\cup \{p_{i,0},\dots, p_{n,3}\mid 0 \leq n< m\}$ and $sup^2_v=\{m_0,\dots, m_4\} \cup\{p_{n,0},\dots, p_{n,1} \mid 0 \leq n < m\}$ settle the inhibition of $v$.
Finally, the region $sup_c = \{m_0, m_3\}$, respectively $sup_a = \{s_0, m_0\}$, inhibits the event $c$, respectively the event $a$.

It is easy to see that the state separating regions defined above for the consumer nets can be used here to separate the same states simply by replacing $\inp$ by $\res$.
Moreover, the states $S(A_\varphi)$ are clearly separable from the states $S(A^{\times}_\varphi)\setminus S(A_\varphi)$ as $S(A_\varphi)$ is itself a $\tau$-support. 
The remaining SSP-atoms are solved by $\{m_0\},\{m_0,m_1\}, \{m_0,m_3\}, \{m_0,m_3,m_4\},\{m_0,m_3,m_4\}$ and by $\{p_{i,0}\}$, $\{p_{i,0},p_{i,1}\}$,$\{p_{i,0},p_{i,1},p_{i,2}\}$ for $i\in \{0,\dots, m-1\}$.
If $\used$ is not available, then $\free$ is and we can modify all regions accordingly by inverting the support and interchanging $\set$ with $\res$ and $\used$ with $\free$.
\end{proof}

\section{Polynomial Time Net Synthesis for 36 Types of Nets}\label{sec:polynomial_cases}

This section proves the tractability of synthesis for the 36 types of nets given in the following theorem:
\begin{theorem}\label{the:polynomial_cases}
There is a polynomial time algorithm, that, on input TS $A$, synthesizes a $\tau$-net $N$ with state graph isomorphic to $A$ or rejects $A$ if $N$ does not exist, for every
\begin{enumerate}
\item\label{the:polynomial_cases_nop_set}
$\tau = \{\nop, \set\} \cup \omega$ with $\omega \subseteq \{\out, \used, \free\}$, %8
\item\label{the:polynomial_cases_nop_res}
$\tau = \{\nop, \res\} \cup \omega$ with $\omega \subseteq \{\inp, \used, \free\}$, %8
\item\label{the:polynomial_cases_flip_flop}
$\tau = \{\nop, \swap\} \cup \omega$ with $\omega \subseteq \{\inp, \out, \used, \free\}$, %16
\item\label{the:polynomial_cases_trivial2}
$\tau = \{\nop\} \cup \omega$ with $\omega \subseteq \{\used, \free\}$.%4
\end{enumerate}
\end{theorem}

Notice that input is not limited to modest TSs for any case of Theorem~\ref{the:polynomial_cases}.
Hence, our efficient methods are robust with respect to general input and do not depend on any restrictions.

The following subsection introduces a new polynomial time algorithm for Theorem~\ref{the:polynomial_cases}.\ref{the:polynomial_cases_nop_set} and~\ref{the:polynomial_cases}.\ref{the:polynomial_cases_nop_res}.
The basic idea is to compute a region set $\R$ solving all (E)SSP atoms of a given TS $A$, if such a set exists.
Using $\R$, the sought net $N(A, \R)$ can easily be computed.

For Theorem~\ref{the:polynomial_cases}.\ref{the:polynomial_cases_flip_flop}, Section \ref{sec:flip_flop} shows how to extend the algorithm of Schmitt \cite{S1996} in order to efficiently synthesize nets.

Before going into the two subsections, we turn towards Theorem~\ref{the:polynomial_cases}.\ref{the:polynomial_cases_trivial2} which covers types of nets with a rather trivial synthesis problem.
A related polynomial time algorithm is established in the proof of the following lemma:
\begin{lemma}\label{lem:nop[used][free]}
If $\tau=\{\nop\} \cup \omega$ is a type of nets with $\omega \subseteq \{\used, \free\}$ then $\tau$-feasibility of a given TS $A$ can be decided in $\mathcal{O}(1)$ time.
A $\tau$-feasible TS $A$ can be synthesized into a $\tau$-net $N$ with $A(N)$ isomorphic to $A$ in $\mathcal{O}(\vert E(A)\vert)$ time.
\end{lemma}
\begin{proof}
Let $s_0 \edge{e} s_1$ be any transition of $A$ with $s_0 \not= s_1$.
To separate these states, we require a $\tau$-region $(sup, sig)$ with $sup(s_0) \not= sup(s_1)$.
However, $sig(e) \in \tau$ implies that $sup(s_0) = sup(s_1)$.
Hence, $s_0$ and $s_1$ are not $\tau$-separable.
As all states have to be reachable, this implies that input TSs with more than one state cannot be $\tau$-feasible and, thus, are discarded after a constant time check.

Being reduced, $A = (\{s\}, E, \{s \edge{e} s \mid e \in E\}, s\})$ is the only single-state TS.
For this input, $N=(\{p\}, E(A), \{(p,e,\nop) \mid e \in E(A)\}, \{p\})$ has a state graph isomorphic to $A$ and is, thus, written to the output in $\mathcal{O}(\vert E(A) \vert )$ time.
\end{proof}

\subsection{Net Synthesis by Incremental  Region Growing}\label{sec:nop_res}

The result of this section is the following contribution to Theorem~\ref{the:polynomial_cases}:
\begin{lemma}\label{lem:polynomial_cases_nop_resOrSet}
If $\tau=\{\nop, \res\} \cup \omega$ is a type of nets with $\omega \subseteq \{\inp, \used, \free\}$ or $\tau=\{\nop, \set\} \cup \omega'$ with $\omega' \subseteq \{\out, \used, \free\}$ then any given TS $A$ can be synthesized into a $\tau$-net $N$ with $A(N)$ isomorphic to $A$, respectively rejected if the net $N$ does not exist, in $\mathcal{O}(\vert E(A)\vert \vert S(A)\vert^6 \max\{\vert E(A)\vert, \vert S(A)\vert\})$ time.
\end{lemma}

Before we can go into the proof of this lemma, we introduce the respective algorithm.
The core subroutine of this method is Algorithm \ref{alg:compute_support}.
According to Lemma \ref{lem:algorithm}, this method accepts a given set of states $Q \subseteq S(A)$ and returns a minimal superset $sup \supseteq Q$ that, together with a matching signature, forms a region of $A$.
Later, in Lemma \ref{lem:nop_res[inp][used][free]}, we show that the regions $\R$ derived from Algorithm \ref{alg:compute_support} solve all (E)SSP atoms of $A$.
Using that $\R$ is small enough and that $N(A, \R)$ is isomorphic to $A$ leads to an efficient synthesis method for reset nets and their combinations with \inp, \used, and \free.
The tractability for synthesis of set nets and extensions follows from type isomorphisms.
\begin{algorithm}
\KwData{TS $A$ and set of states $Q \subseteq S(A)$}
\KwResult{A support $sup \supseteq Q$ for a region of $A$.}

\While{$\exists\ s \in Q, s' \not\in Q, e \in E(A): (s' \edge{e} s) \vee (s \edge{e} s' \wedge z \edge{e} z' \text{ for } z, z' \in Q)$}
{
$Q = Q \cup \{s'\}$\;
}
\Return $sup = Q$\;
\
\caption{
Given $Q \subseteq S(A)$, the algorithm minimally extends $Q$ to the support of a $\tau$-region $(sup, sig)$ of $A$ for all reset types of nets $\tau = \{\nop, \res\} \cup \omega $ extended with $\omega \subseteq \{\inp, \used, \free\}$.
}
\label{alg:compute_support}
\end{algorithm}
\begin{lemma}\label{lem:algorithm}
If $\tau=\{\nop, \res\} \cup \omega$ is a type of nets with $\omega \subseteq \{\inp, \used, \free\}$ and $A$ is a TS and $Q \subseteq S(A)$ then the result $sup$ of Algorithm~\ref{alg:compute_support} started on $Q$ forms a $\tau$-region $(sup, sig)$ of $A$ with
\[sig(e) = 
\begin{cases}
\used, & \text{if } \used \in \tau \text{ and } \{s,s'\mid s\edge{e}s'\}\subseteq sup,\\
\free, & \text{if } \free \in \tau \text{ and } \{s,s'\mid s\edge{e}s'\} \cap sup = \emptyset,\\
\inp, & \text{if } \inp \in \tau \text{ and for all } s\edge{e}s': sup(s) = 1, sup(s') = 0,\\
\res, & \text{if } \inp \not\in \tau \text{ and for all } s\edge{e}s': sup(s) = 1, sup(s') = 0,\\
\res, & \text{if for at least one but not all } s\edge{e}s' : sup(s) = 1, sup(s') = 0,\\
\nop, &\text{otherwise,}
\end{cases}
\]
for all $e \in E(A)$.
Moreover, for all $\tau$-regions $(sup', sig')$ of $A$ with $Q \subseteq sup'$ it is true that even $sup \subseteq sup'$.
Algorithm~\ref{alg:compute_support} terminates after $\mathcal{O}(\vert E(A)\vert \vert S(A)\vert^5)$ time.
\end{lemma}
\begin{proof}
That the algorithm terminates is trivial as every iteration extends $Q$, which is possible for at most $\vert S(A)\vert$ times.
After termination, $sup$ obviously contains input $Q$.
Moreover, there are no events $e \in E(A)$ participating in a transition $s' \edge{e} s$ with $s \in sup, s' \not\in sup$.
On the other hand, if there is transition $s \edge{e} s'$ with $s \in sup, s' \not\in sup$ then no other transition $z \edge{e} z'$ can be completely inside $sup$ and hence, we can validly assign $sig(e) = \res$ or even $sig(e) = \inp$, if \inp\ is available and $z \in sup, z' \not\in sup$, all along.
Otherwise, if all transitions of $e$ are completely in- or outside $sup$, assigning $sig(e)=\nop$ is always valid.
If $e$'s transitions are consistently inside, respectively outside, $sup$ then even $sig(e)=\used$, respectively $sig(e)=\free$, is possible, given that the interaction is available.
Hence, $(sup, sig)$ is a $\tau$-region of $A$ in every case.

Now let $(sup',sig')$ be any $\tau$-region of $A$ with $Q \subseteq sup'$.
We show by induction that the set $Q_i$ resulting from $i$ while-iterations of Algorithm~\ref{alg:compute_support} fulfills $Q_i \subseteq sup'$.
For a start, $Q_0 = Q \subseteq sup'$.
Assume that $Q_i \subseteq sup'$ and $Q_{i+1} \not \subseteq sup'$ and let $\{s'\} = Q_{i+1} \setminus Q_i$ which, thus, fulfills $s' \not\in sup'$.
As $s'$ is added to $Q_{i+1}$, there are $s \in Q_i \subseteq sup'$ and $e \in E(A)$ such that either $s' \edge{e} s$ or $s \edge{e} s'$ and $z \edge{e} z'$ with $z,z' \in Q_i \subseteq sig'$.
But then $sig'(e)$ cannot be any interaction in $\{\nop, \res, \inp, \used, \free\}$, a contradiction. 
When the while loop terminates after $n$ iterations, then $sup$ becomes $Q_n$ and, thus, fulfills $sup = Q_n \subseteq sup'$.

As there are at most $\vert S(A)\vert$ while-iterations and as checking the while-condition takes $\mathcal{O}(\vert E(A)\vert\vert S(A)\vert^4)$ time, Algorithm~\ref{alg:compute_support} runs in $\mathcal{O}(\vert E(A)\vert \vert S(A)\vert^5)$ time.
\end{proof}

Having a way to reliably produce $\tau$-regions for the net classes of interest, we argue that they are versatile enough to solve all (E)SSP atoms.
Otherwise, the computed regions $\R$ would not suffice to synthesize the net $N(A, \R)$.
\begin{lemma}\label{lem:nop_res[inp][used][free]}
If $\tau=\{\nop, \res\} \cup \omega$ with $\omega \subseteq \{\inp, \used, \free\}$ and $A$ is a TS then $e \in E(A)$ is $\tau$-inhibitable at $s \in S(A)$ where $\neg (s\edge{e})$ if and only if
\begin{enumerate}
\item\label{lem:nop_res[inp][used][free]_inp}
$\inp \in \tau$ and the region $(sup, sig)$ returned by Algorithm~\ref{alg:compute_support} on input $Q = \{z \mid z\edge{e}\}$ satisfies $sig(e)=\inp$ and $sup(s)=0$, or
\item\label{lem:nop_res[inp][used][free]_used}
$\used\in \tau$ and the region $(sup, sig)$ returned by Algorithm~\ref{alg:compute_support} on input $Q=\{z, z' \mid z \edge{e} z'\}$ satisfies $sig(e)=\used$ and $sup(s)=0$, or
\item\label{lem:nop_res[inp][used][free]_free}
$\free\in \tau$ and the region $(sup, sig)$ returned by Algorithm~\ref{alg:compute_support} on input $Q = \{s\}$ satisfies $sig(e)=\free$.
\end{enumerate}
Two states $s, s' \in S(A)$ are $\tau$-separable if and only if the region $(sup, sig)$ returned by Algorithm~\ref{alg:compute_support} fulfills $sup(s') = 0$ for $Q = \{s\}$ or $sup(s)=0$ for $Q = \{s'\}$.
\end{lemma}
\begin{proof}
The if-direction for an ESSP atom $(A, e, s)$ is trivial, as $e$ is inhibited at $s$ by the pair $(sup, sig)$ from Algorithm~\ref{alg:compute_support}, a $\tau$-region according to Lemma \ref{lem:algorithm}.
Reversely, if $e$ is inhibited at $s$ by a $\tau$-region $(sup',sig')$ then interaction $sig'(e)$ is not defined on $sup'(s)$.
Hence, $sig'(e) \not\in \{\nop, \res\}$.
Let $X = \{x \mid x \edge{e}\}$, $Y = \{y \mid \edge{e} y\}$, and $Z = \{z, z' \mid z \edge{e} z'\}$.
If $sig'(e) = \inp$ ($sig'(e) = \used$, $sig'(e) = \free$) then $sup'(s) = 0$ ($sup'(s) = 0$, $sup'(s) = 1$).
Notice that $sup' \cap Y = \emptyset$ ($Z \subseteq sup'$, $Z \cap sup' = \emptyset$).
By Lemma~\ref{lem:algorithm}, the region $(sup, sig)$ returned by Algorithm~\ref{alg:compute_support} on $Q = X$ ($Q = Z$, $Q = \{s\}$) satisfies $Q \subseteq sup \subseteq sup'$.
Hence, $sup(s) = 0$ ($sup(s) = 0$, $sup(s) = 1$) and $sup \cap Y = \emptyset$ ($Z \subseteq sup$, $Z \cap sup = \emptyset$), which implies $sig(e)=\inp$, ($sig(e)=\used$, $sig(e)=\free$), too.

The if-direction for the SSP atom $(A, s, s')$ is trivial, again, as the $\tau$-region of Algorithm~\ref{alg:compute_support} separates $s, s'$.
Reversely, let $s, s'$ be separated by a $\tau$-region $(sup',sig')$ where, without loss of generality, $sup'(s)=1$ and $sup'(s')=0$.
The result $(sup, sig)$ of Algorithm~\ref{alg:compute_support} on $Q = \{s\}$ is a $\tau$-region by Lemma~\ref{lem:algorithm} that fulfills $Q \subseteq sup \subseteq sup'$.
Hence, $sup(s)=1$ and $sup(s')=0$, too.
\end{proof}

By the required versatility of the regions from Algorithm~\ref{alg:compute_support}, we can now prove Lemma~\ref{lem:polynomial_cases_nop_resOrSet}:
\begin{proof}[Proof of Lemma \ref{lem:polynomial_cases_nop_resOrSet}]
Let $S = S(A)$ and $E = E(A)$.
The idea is to firstly produce a region set $\R$ that solves all (E)SSP atoms of $A$.
If we cannot find $\R$, then we reject $A$.
There are $\mathcal{O}(\vert E\vert \vert S\vert)$ ESSP atoms $(A, e, s)$.
Depending on the availability of $\inp, \used, \free$ in $\tau$, we have to test the inhibitability of $e$ at $s$ by up to three calls of Algorithm~\ref{alg:compute_support} with inputs $Q_\inp = \{z \mid z\edge{e}\}$, $Q_\used = \{z, z' \mid z \edge{e} z'\}$, and $Q_\free = \{s\}$.
In every case, the method's running time of $\mathcal{O}(\vert E\vert \vert S\vert^5)$ heavily dominates the time for the creation of the input.
If all tests succeed, then we haved picked up enough regions to solve all ESSP atoms in $\mathcal{O}(\vert E\vert^2 \vert S\vert^6)$ time.
Otherwise, Lemma \ref{lem:nop_res[inp][used][free]} allows us to reject $A$.
Notice that, in case of $\tau=\{\nop, \res\}$ there must not be any ESSP atoms as inhibiting interactions $\inp, \used, \free$ are missing.
In this case, we would reject $A$ if it had event $e\in E$ and state $s \in S$ with $\neg (s \edge{e})$.

Next, there are $\mathcal{O}(\vert S\vert^2)$ SSP atoms $(A, s, s')$.
By Lemma \ref{lem:nop_res[inp][used][free]}, we have to call Algorithm~\ref{alg:compute_support} with $Q_s = \{s\}$ and $Q_{s'} = \{s'\}$ to decide the separability of $s, s'$.
After $\mathcal{O}(\vert E\vert \vert S\vert^7)$ time, either $\R$ solves all SSP atoms or we can reject $A$.

Hence, using $\mathcal{O}(\vert E\vert \vert S\vert^6 \max\{\vert E\vert, \vert S\vert\})$ time in total, we decide the feasibility of $A$ and, in the positive case, get $\R$.
Computing $N(A, \R)$ consumes $\mathcal{O}(\vert \R\vert \vert E\vert) = \mathcal{O}(\vert E\vert \vert S\vert \max\{\vert E\vert, \vert S\vert\})$ time, which is dominated by the previous costs.

If $\tau=\{\nop, \set\} \cup \omega'$ with $\omega' \subseteq \{\out, \used, \free\}$ our approach is to synthesize a net $N'$ for the isomorphic type $\tau'$ that replacing $\set$ with $\res$, $\out$ with $\inp$, $\used$ with $\free$, and $\free$ with $\used$.
In order to obtain a $\tau$-net $N$, we simply revert the interaction replacement in the flow function $f(N')$.
Obviously, $A(N)$ is isomorphic to $A(N')$, which is isomorphic to $A$.
\end{proof}

\subsection{Net Synthesis for Relatives of Flip-Flop-Nets}\label{sec:flip_flop}

Last step in proving Theorem~\ref{the:polynomial_cases} is to cover item \ref{the:polynomial_cases_flip_flop}, the relatives of flip-flop nets:
\begin{lemma}\label{lem:flip_flop_cases}
If $\tau=\{\nop, \swap\} \cup \omega$ with $\omega \subseteq \{\inp, \out, \used, \free\}$ then a given TS $A$ can be synthesized into a $\tau$-net $N$ with $A(N)$ isomorphic to $A$, respectively rejected if $N$ does not exist, in polynomial time.
\end{lemma}

This works simply by modifying Schmitt's algorithm \cite{S1996} which is based on the ability to efficiently solve equations over the boolean field $\mathbb{F}_2$.
As flip-flop nets have already been covered there and as types of nets $\tilde{\tau} = \{\nop,\swap\} \cup \omega$ with $\omega \subseteq  \{\out, \used, \free\}$ are isomorphic to $\tau = \{\nop,\swap\} \cup \tilde{\omega}$, where $\tilde{\omega}$ mirrors $\omega$ simply by replacing \inp\ with \out, \used\ with \free\ and vice versa, it is sufficient to show the Lemma for types of nets

\begin{enumerate}
\item\label{item:nop_swap[used][free]}
$\tau = \{\nop,\swap\} \cup \omega$ with $\omega\subseteq \{\used, \free\}$ and 
\item\label{item:nop_swap_inp[out][used][free]}
$\tau = \{\nop,\inp, \swap \} \cup \omega$ with $\omega \subseteq \{\out,\used,\free\}$ but not  $\tau = \{\nop, \inp, \out, \swap\}$.
\end{enumerate}

The following roughly summarize the main steps of Schmitts algorithm and afterwards we show how this can be modified for our cases.
If we are given a TS $A$, we firstly interpret it as a directed graph on nodes $S(A)$ and with (labeled) directed arcs given by the transitions $s \edge{e} s'$.
As a second step, we use breadth first search to compute in linear time a spanning tree $A'$ of $A$ with the initial state $s_0$ as the root node.
Notice that $A'$ spans all nodes of $A$ as we can be certain that $A$ reaches every state from $s_0$ by at least one path.
In $A'$, however, every node $s \in S(A)$ is now reached by exactly one directed path $\pi_s = s_0 \edge{e_1} \dots \edge{e_n} s$.
Moreover, every transition $s \edge{e} s'$ of $A$ that fails to be an edge of $A'$ is called \emph{chord}.

With the previous definition, our goal for every (E)SSP atom $(A, x, y)$ is to define a system $M_{x,y}$ of equations over the boolean field $\mathbb{F}_2$ using the events $E$ as variables.
If $M_{x,y}$ has a solution $\rho: E \rightarrow \mathbb{F}_2$ assigning either $0$ or $1$ to every event, it leads to a region solving the original atom.
Otherwise, the atom is shown to be unsolvable.
For every $s \in S(A)$, we define the mapping $\psi_s: E \rightarrow \mathbb{F}_2$ assigning to every $e \in E$ the parity $\psi_s(e)$ of occurrences of event $e \in E$ on the path $\pi_s$, $0$ for even (including zero occurrences) and $1$ for odd.
Then, every cord $t: s \edge{e} s'$ is associated with the linear chord equation $\psi_t: (\rho(e) + \sum_{e' \in E} (\psi_s(e') + \psi_{s'}(e')) \cdot \rho(e')) \mod 2 = 0 $.
Combining all cord equations together defines the basic equation system $\Psi$.

Any solution $\rho$ to $\Psi$ is called abstract region with the following meaning:
As long as $\tau$ contains \nop\ and \swap, we can derive real $\tau$-regions $(sup, sig)$ from $\rho$ by defining the support $sup(s) = \sum_{e \in E} \psi_s(e) \cdot \rho(e) \mod 2$ for all $s \in S(A)$.
The chord equations make sure that the event parity on any other path to $s$ is equal to $sup(s)$.
This justifies to set $sig(e) = \nop$ for all $e \in E$ with $\rho(e) = 0$ and, otherwise, if $\rho(e) = 1$ then $sig(e) = \swap$ is possible.
Another possibility would be to select the complementary support  $sup(s) = (1 + \sum_{e \in E} \psi_s(e) \cdot \rho(e)) \mod 2$.

Based on the specific type of nets $\tau$ and the atom $(A,x,y)$, we augment $\Psi$ with additional equations to obtain $M_{x,y}$.
If, beside \nop\ and \swap, $\tau$ does not contain anything beyond $\{\inp, \out, \used, \free\}$ then, to solve an SSP atom $(A, s, s')$, it is already enough to extend $\Psi$ with the equation $\sum_{e \in E} (\psi_s(e) + \psi_{s'}(e)) \cdot \rho(e) = 1 \mod 2$ to find $M_{s,s'}$.
This equation simply makes sure that $sup(s) \not= sup(s')$.
If $\rho$ does not exist, then $s$ and $s'$ cannot be separated, even though if any subset of $\{\inp, \out, \used, \free\}$ was additionally available in $\tau$.

How to solve an ESSP atom $(A, e, s)$ depends on the availability of \inp, \out, \used, and \free\ in $\tau$.
Taking flip-flop nets, this brings along only \inp\ and \out\ and we obtain $M_{e,s}$ by complementing $\Psi$ with the equation $\rho(e) = 1$ and for all transitions $z \edge{e} z'$ with the equations $\sum_{e' \in E} (\psi_{s}(e') + \psi_{z}(e')) \cdot \rho(e') = 1$.
This makes sure that all sink states of $e$-transitions have the same support as $s$ and all source states opposite support. 
If $\rho$ is a solution to $M_{e,s}$, we define the support and signature like before except for $e$ where we  let $sig(e) = \inp$ if $sup(z) = 1$ for any (that also means all) $z \edge{e}$ and, otherwise, $sig(e) = \out$.
If $\rho$ does not exist, the ESSP atom cannot be solved.

Compared to the number of transitions in $A$, the system $M_{x,y}$ has at most a linear amount of equations.
As solving $M_{x,y}$ is in polynomial time, solving (E)SSP atoms is tractable for flip-flop nets.
Having the polynomial size set $\R$ of regions for all these atoms after polynomial time, we can synthesize the net $N(A, \R)$ in polynomial time, too. 
Keeping this approach in mind, we are ready to prove the remainder of Lemma~\ref{lem:flip_flop_cases}:

\begin{proof}[Proof of Lemma \ref{lem:flip_flop_cases}]
As discussed before, we only have to show for every respective $\tau$ that the ESSP atoms $(A, e, s)$ can be expressed as extensions of $\Psi$.
The idea is to create one or more systems $M^i_{e,s}$ for every $i \in \{\inp, \out, \used,\free\} \cap \tau$ that have at least one solution $\rho$ if and only if the atom can be solved with a region $(sup, sig)$ where $sig(e) = i$.
Furthermore, as every abstract region $\rho$ induces two complementary $\tau$-supports, the atom $(A,e,s)$ is $\inp$-solvable ($\used$-solvable) if and only if it is $\out$-solvable ($\free$-solvable). 
Therefore, it is sufficient to argue only for $M^\inp_{e,s}$ and $M^\free_{e,s}$.

For $M^\inp_{e,s}$, we simply use the same system as for flip-flop nets.
Having a solution $\rho$, we first define support and signature as before where for $e$ we  let $sig(e) = \inp$ if $sup(z) = 1$ for any (that also means all) $z \edge{e}$.
Otherwise, if $sup(z) = 0$, we switch to the complementary support, which turns $sup(z)$ to $1$ and thus, $sig(e) = \inp$, again.
If $M^\inp_{e,s}$ has no solution then it is not possible to inhibit $e$ at $s$ using \inp.

To get $M^\used_{e,s}$, we take $\Psi$ and add the equation $\rho(e)=0$ plus for every transition $z \edge{e} z'$ the equations $\sum_{e' \in E} (\psi_{s}(e') + \psi_{z}(e')) \cdot \rho(e') = 1$.
These equations make sure that $e$ can be assigned \used\ or \free\ and that states incident to $e$ behave correctly.
If $\rho$ is a solution to $M^\used_{e,s}$, we again define support and signature like in all previous cases except for $e$.
In fact, we let $sig(e) = \used$ if $sup(z) = 1$ for any (that also means all) $z \edge{e}$.
Otherwise, if $sup(z) = 0$, we complement the support turning $sup(z)$ to $1$ and allowing $sig(e) = \used$.
Again, if $M^\used_{e,s}$ has no solution then it is not possible to inhibit $e$ at $s$ by \used.
\end{proof}

\section{Conclusion}

In this paper we investigate the complexity of boolean net synthesis for the 128 practically more relevant \nop-afflicted classes.
In total, we prove 84 cases NP-hard and provide polynomial time algorithms for 36 classes.
As a side product, this paper introduces a very general reduction scheme that serves well for NP-completeness proofs in this matter.

For the eight classes (\nop, \inp, [\used]), (\nop, \out, [\free]), (\nop, \set, \res) and (\nop, \swap) extended with at least one of \set\ and \res, we leave the complexity of synthesis open.
They remain for future work.

While the first four items of this list are just surprisingly difficult with respect to their very limited interactions set, the really interesting classes are the last four which are built only from \nop, \swap, \set, and \res.
None of these events can be used to solve ESSP-atoms.
Consequently, any input TS $A$ with at least one state $s$ and event $e$ with $\neg (s \edge{e})$ is instantly unfeasible and can be rejected.
Only if every event occurs at every state of $A$, we have to search for a respective net of the given class.
But then ESSP is already solved and we only have to check SSP to test feasibility, which turns out fairly difficult, too.
This is partly caused by the fact that solving SSP atoms with one of the four mentioned subclasses of (\nop, \swap, \set, \res) can be shown NP-complete.
Hence, here tractability for feasibility would depend very much on the restriction of $s \edge{e}$ for all states $s$ and all events $e$ of $A$.

\section{Technical Proofs from Section \ref{sec:main_result}}\label{sec:techproofs}

\begin{proof}[Proof of Lemma~\ref{lem:translators}]
(\ref{lem:translators_1}-\ref{lem:translators_2}): 
Firstly, if $T^\sigma_\varphi$ installs a non empty freezer then, by Lemma~\ref{lem:generator}, we have $sig(x_j)\not=\swap$.\newline
Let $i\in \{0,\dots, m-1\}$.
For abbreviation we define $S_0=\{t_{i,\alpha,2}\mid 0\le i \le m-1, 0\le \alpha \le 2\}$ and $S_1=\{t_{i,\alpha,5} \mid 0\le i \le m-1, 0\le \alpha \le 2\}$.
We now argue for $\sigma\in \{ \sigma_1,\sigma_2,\sigma_4\}$ and see later that the arguments are symmetrically true for $\sigma=\sigma_3$.\newline
If condition (\ref{lem:translators_1}) is satisfied then by $sig_K(k)=\inp$,  we have for all $\alpha \in \{0,\dots,2\}$ that $sup_K(t_{i,\alpha,1})=0$ which with $V \subseteq sig_K^{-1}(\enter)$ and $W \subseteq sig_K^{-1}(\keepze)$ implies that $S_0\subseteq sup^{-1}(1)$ and $S_1\subseteq sup^{-1}(0)$.
Symmetrically, condition (\ref{lem:translators_2}) implies that $S_0\subseteq  sup^{-1}(0)$ \emph{and} $S_1\subseteq sup^{-1}(1)$.

Let $\alpha\in \{0,1,2\}$ and $\beta=\alpha+1 \mod 3$ and $\gamma=\alpha+2 \mod 3$.
Firstly, $S_0\subseteq sup^{-1}(1), S_1\subseteq sup^{-1}(0)$ ($S_0\subseteq sup^{-1}(0), S_1\subseteq sup^{-1}(1)$) implies that at least one element of $X_{i,0},X_{i,1},X_{i,2}$ has a signature from $\{\inp,\res,\swap\}$ ($\{\out,\set,\swap\}$), cf. Figure~\ref{fig:translators}.1.
Secondly, $t_{i,\alpha,2}$ is a source of $X_{i,\alpha}$ and $t_{i,\beta,5}$ is a sink of $X_{i,\alpha}$ which implies $sig(X_{i,\alpha})\not\in \{\set,\out,\free,\used\}$ ($sig(X_{i,\alpha})\not\in  \{\res,\inp,\free,\used\}$), cf. Figure~\ref{fig:interactions}. 
By a similar argument and $sig(x_j)\not=\swap$ for all $j\in \{0,\dots, m-1\}$, we obtain that $S_0\subseteq sup^{-1}(1), S_1\subseteq sup^{-1}(0)$ ($S_0\subseteq sup^{-1}(0), S_1\subseteq sup^{-1}(1)$) implies that at least one element of $\{x_{i,0},x_{i,1},x_{i,2}\}$ has a signature from $\{\out,\set\}$ ($\{\inp,\res\}$) and all of them are prevented to have a signature from $\{\res,\inp,\swap,\free,\used\}$ ($\{\set,\out,\swap,\free,\used\}$).\newline 
We now argue that there is exactly one variable event of $T_i$ with a signature different from $\nop$.
To do so, we show that if $sig(X_{i,\alpha})\not=\nop$ then $sig(X_{i,\beta})=sig(X_{i,\gamma})=\nop$.\newline 
If $S_0\subseteq sup^{-1}(1), S_1\subseteq sup^{-1}(0)$ ($S_0\subseteq sup^{-1}(0), S_1\subseteq sup^{-1}(1)$) and $sig(X_{i,\alpha})\in \{\inp,\res,\swap\}$ ($sig(X_{i,\alpha})\in \{\out,\set,\swap\}$) implying that $sup(t_{i,\alpha,3})=0$ ($sup(t_{i,\alpha,3})=1$) then we can immediately conclude the following fact:
By $sup(t_{i,\alpha,2})=1$ ($sup(t_{i,\alpha,2})=0$) and $sup(t_{i,\alpha,3})=0$ ($sup(t_{i,\alpha,3})=1$) and Figure~\ref{fig:interactions} we have that  $x_{i,\alpha}\in \{\set, \out\} $ ($x_{i,\alpha}\in \{\inp, \res\}$) which implies $sup(t_{i,\beta,4})=1$ ($sup(t_{i,\beta,4})=0$).
Moreover, by $sig(x_{i,\beta}),sig(x_{i,\gamma})\not\in \{\res,\inp\}$ ($sig(x_{i,\beta}),sig(x_{i,\gamma})\not\in \{\set,\out\}$), we have $sup(t_{i,\beta,2})=sup(t_{i,\beta,3})=1$ ($sup(t_{i,\beta,2})=sup(t_{i,\beta,3})=0$).
The inclusion (exclusion) of $t_{i,\beta,2},t_{i,\beta,3},t_{i,\beta,4}$ additionally prevents the \swap\ signature for the events $X_{i,\beta},X_{i,\gamma}$ resulting in $sup(t_{i,\alpha,4})=sup(t_{i,\alpha,5})=0$ ($sup(t_{i,\alpha,4})=sup(t_{i,\alpha,5})=1$), cf. Figure~\ref{fig:interactions}.
Consequently, for $X\in \{X_{i,\beta},X_{i,\gamma}\}$ there are transitions $s\edge{X}s'$ and $s''\edge{X}s'''$ such that $sup(s)=sup(s')=0$ and $sup(s'')=sup(s''')=1$ assuring the \nop\ signature for $X$.\newline 
Consequently, $(sup,sig)$ satisfies for all $i\in \{0,\dots,m-1\}$ the condition $\vert M=\{X\in V(\varphi)\vert sig(X)\not=\nop\}\cap E(T_i) \vert =1 $ which makes $M$ a one-in-three model of $\varphi$.

(\ref{lem:translators_3}): 
All $a\in Acc$ satisfy that if $s\edge{a}s'$ in $T^\sigma_{\varphi}$ then $s'\edge{a}s$ in $T^\sigma_{\varphi}$. 
Consequently, $Acc\cap sig_K^{-1}(\swap)=\emptyset$ assures state synchronization for $a$-labeled transitions: $sup(s)=sup(s')$. 
For abbreviation we define $S_0=\{t'_{i,\alpha,2}\mid 0\le i \le m-1, 0\le \alpha \le 2\}$ and $S_1=\{t'_{i,\alpha,11}\mid 0\le i \le m-1, 0\le \alpha \le 2\}$.

Let $\sigma=\sigma_5$ and $i\in \{0,\dots, m-1\}$.
By definition of $\sigma_5$ we have $sig(k)=\free$.
Firstly, we note that by $sig(q_2)=\swap$ if and only if $sig(q_3)=\swap$  we have $sup(b'_{j,2})=sup(b'_{j,3})$ implying $sig(x_j)\not=\swap$ for $j\in \{0,\dots,m-1\}$.
Secondly, we show that for $T^\sigma_i$ there is exactly one variable event with a border crossing signature and both of the others are mapped to \nop.
By $sig(k)=\free$, $W \cap sig^{-1}(\swap)=\emptyset$ and $V \subseteq sig^{-1}(\swap)$ we have $S_0\subseteq sup^{-1}(1)$ and $S_1\subseteq sup^{-1}(0)$.
We get the following consequences:
Firstly, at least one variable event of $X_{i,0},X_{i,1},X_{i,2}$ must have a border crossing signature, that is, \swap.
Secondly, if more than one variable event has a \swap\ signature, then \emph{all} variable events must be mapped to \swap. 
By the state synchronization explained in Section \ref{sec:general_scheme} and defined above, the latter case would imply for $\alpha\in \{0,1,2\}$ that $sup(t'_{i,\alpha,17})=1$ and $sup(t'_{i,\alpha,15})=0$, which contradicts $\res\not\in \tau$.
Hence, exactly one variable event $X\in \{X_{i,0},X_{i,1},X_{i,2}\}$ has a \swap\ signature per translator (clause) $T^\sigma_{i}$ ($\zeta_i$).
For $Y\in \zeta_i\setminus \{X\}$ there are transitions $s\fbedge{Y}s'$ and $s''\fbedge{Y}s'''$ in $T^\sigma_i$ such that $sup(s)=sup(s')=1$ and $sup(s'')=sup(s''')=0$ which implies $sig(Y)=\nop$.
Consequently, $M=\{X\in V(\varphi)\mid sig(X)\not=\nop\}$ is a one-in-three model of $\varphi$.

Let $\sigma=\sigma_6$ and $i\in \{0,\dots, m-1\}$.
For $j\in \{0,\dots, m-1\}$ the signature $sig(k)=\used$ ($sig(k)=\free$) implies $sig(x_j)\not\in \{\inp, \res, \swap\}$ ($sig(x_j)\not\in \{\out, \set, \swap\}$), cf. Figure~\ref{fig:translators}.3.  
If $sig(k)=\used$ ($sig(k)=\free$) then, by $W\cap sig_K^{-1}(\swap)=\emptyset$ and $V\subseteq sig_K^{-1}(\swap)$, we easily have $S_0\subseteq sup^{-1}(1)$ and $S_1\subseteq sup^{-1}(0)$ ($S_0\subseteq sup^{-1}(0)$ and $S_1\subseteq sup^{-1}(1)$), c.f. Figure~\ref{fig:translators}.5.
Symmetrically to the argumentation for $\sigma=\sigma_5$, this implies that either exactly one or all variable event(s) must have a \swap\ signature, where the latter case contradicts for $sig(k)=\used$ ($sig(k)=\free$) that $sig(x_j)\not\in \{\res, \swap\}$ ($sig(x_j)\not\in \{ \set, \swap\}$).
Hence, similar to the case $\sigma=\sigma_5$, that results in $M=\{X\in V(\varphi)\mid sig(X)\not=\nop\}$ being a one-in-three model of $\varphi$.
\end{proof}

\begin{proof}[Proof of Lemma \ref{lem:key_union}]
Let $C=\{c_0,\dots, c_{6m-2}\}$ and $Z=\{z_0,\dots, z_{3m-1}\}$.

(\ref{lem:key_union_1}): Firstly, we show that $sig_K(k)\in \{\out,\inp\}$ for all $\sigma \in \{\sigma1, \dots, \sigma_4\}$.

If $\sigma\in \{\sigma_1,\sigma_2\}$ assume that $sig_K(k)=\used$ ($sig_K(k)=\free$).
As $F^\sigma_K=U(F_0,F_1)$, we have $sup_K(f_{0,3})=sup_K(f_{1,1})=sup_K(h_{0,4})=1$ ($sup_K(f_{0,3})=sup_K(f_{1,1})=sup_K(h_{0,4})=0$), cf. Figure~\ref{fig:key_unions}.2, Figure~\ref{fig:key_unions}.3 and Figure~\ref{fig:key_unions}.4.
Hence, we have $sig_K(z_0), sig_K(q_0)\not\in \{\inp,\res\}$ ($sig_K(z_0), sig_K(q_0)\not\in \{\out,\set\}$) implying that $sup_K(h_{0,6})=1$ ($sup_K(h_{0,6})=0$) and, thus, $k$ is not inhibited at the key state, a contradiction.
Consequently, $sig_K(k)\in \{\out,\inp\}$.

If $\sigma=\sigma_3$ ($\sigma=\sigma_4$), then $F^\sigma_K$ contains $F_0,F_2$ ($F_0,F_1$) and $G^{\_,q}_0,\dots, G^{\_,q}_{3m-1}$.
Assume that $sig_K(k)=\used$ ($sig_K(k)=\free$).
All states in $S(F^\sigma_K)$ incident to event $k$ have to be part (outside of) the support.
By $F_2$ ($G^{\_,q}_0$) we get that $n_0$ ($q_0$) cannot be in $\exit$ ($\enter$) and, thus, $sup(f_{0,2}) = 1$ ($sup(f_{1,0}) = 0$).
Thus, $S(F^\sigma_K)\subseteq sup_K^{-1}(1)$ ($S(F^\sigma_K)\subseteq sup_K^{-1}(0)$) which clearly implies $sup_K(h_{0,4})=sup_K(h_{0,6})=1$ ($sup_K(h_{0,4})=sup_K(h_{0,6})=0$).
Consequently, if $sig_K(k)\in \{\used,\free\}$ then $k$ is not inhibited at $h_{0,6}$, a contradiction.
Hence, $sig_K(k)\in \{\out,\inp\}$.

Secondly, assume that we inhibit $sig_K(k)=\inp$ at the key state $h_{0,6}$, which means $sup(h_{j,6})=0$.
We show that $V\subseteq sig^{-1}_K(\enter)$ and $W\subseteq sig^{-1}_K(\keepze)$.
Applying Lemma~\ref{lem:generator} to $D^\sigma$, we have $C\subseteq sig_K^{-1}(\nop)=sig_K^{-1}(\keepo)\cap sig_K^{-1}(\keepze)$.
Let $j\in \{0,\dots, 6m-1\}$.
By $sig_K(k)=\inp$, $C\subseteq sig_K^{-1}(\nop)$ and $sup(h_{0,6})=0$ we have $sup(h_{j,1})=sup(h_{j,4})=sup(h_{j,6})=0$ and $sup(h_{j,3})=1$. 
If $\swap\not\in \tau$ then $sup_K(h_{j,6})=0$  implies $Z, W\subseteq sig_K^{-1}(\keepze)$ and $sup(h_{j,2})=0$.
If $\swap\in \tau$ then by Lemma~\ref{lem:generator} we have $Q,Y\subseteq sig_K^{-1}(\keepze)$ and $sup_K(h_{j,5})=0$.
This implies $Z, W\subseteq sig_K^{-1}(\keepze)$ and $sup(h_{j,2})=0$, too. 
Hence, by $h_{j,3}=1$, we obtain $V\subseteq sig^{-1}_K(\enter)$.
Symmetrically, if $ sig_K(k)=\out $ we obtain that $V\subseteq sig_K^{-1}(\exit) $ and $ W\subseteq sig_K^{-1}(\keepo)$.

To prove the existence of an announced key region of $K^\sigma_\varphi$ for $\sigma\in \{\sigma_1,\dots, \sigma_4\}$ we, firstly, define the following subsets and operation containers and, secondly, show how they are to composed to a corresponding region:
\begin{enumerate}
%Head
\item
$S_0=\{h_{j,0},h_{j,3}\mid j\in \{0,\dots, 6m-1\}\}$,
%generators q- and y-events
\item 
$S_1=\{g^{\_,q}_{j,0}, g^{\_,q}_{j,1}, g^{\_,y}_{j,0}, g^{\_,y}_{j,1} \mid j\in \{0,\dots, 3m-1\}\}$,
%generators c-events
\item 
$S_2=\{g^{c,c}_{j,0}, g^{c,c}_{j,1} \mid j\in \{0,\dots, 6m-2\}\}$,
%freezer F_K
\item 
$S_3=\{g^{n,\_}_{0,0}, g^{n,\_}_{0,1}, f_{0,0},f_{0,2}, f_{0,3}, f_{1,0}, f_{1,1}, f_{2,0}, f_{2,3} \}$,
%support
\item 
for $\sigma\in \{\sigma_1,\dots, \sigma_4\}$: $sup^\sigma_K=S(K^\sigma_m)\cap (S_0\cup S_1\cup S_2\cup S_3)$,
%operations
\item 
$\textsf{op}^{\sigma_1}_n=\out$, $\textsf{op}^{\sigma_2}_n=\set$, $\textsf{op}^{\sigma_3}_n=\swap$, $\textsf{op}^{\sigma_4}_n=\set$.
\end{enumerate} 
For $\sigma\in \{\sigma_1,\dots, \sigma_4\}$ the set $sup^\sigma_K$ allows the signature $sig^\sigma_K$, where for $e\in E(K^\sigma_m)$ we have that
\[sig^\sigma_K(e)
\begin{cases}
\inp \text{ if } e=k, \\
\textsf{op}^\sigma_V \text{ if } e\in \{v_0,\dots, v_{3m-1}\}, \\
\textsf{op}^\sigma_n \text{ if } e=n_0, \\
\nop \text{ otherwise } 
\end{cases}
\]
Clearly, the region $(sup^\sigma_K, sig^\sigma_K)$ inhibits $k$ at $h_{0,0}$ and satisfies the condition of the lemma, cf. Figure~\ref{fig:key_unions}.
Hence, the first claim is proven.

(\ref{lem:key_union_2}): 
Let $j\in \{0,\dots, 3m-1\}$.
By definition of $\sigma$ clearly we have either $sig_K(k)=\used$, $sup_K(h'_{j,2})=0$ and $sup_K(h'_{j,1})=sup_K(h'_{j,3})=1$ or $sig_K(k)=\free$,  $sup_K(h'_{j,2})=1$ and $sup_K(h'_{j,1})=sup_K(h'_{j,3})=0$.
Both cases easily imply $sig_K(m)=sig_K(v_j)=\swap$ and, consequently, $V\subseteq sig^{-1}_K(\swap)$.
Similarly, we obtain $sig_K(q_0)=sig_K(q_1)=\swap$.
By $sup_K(f'_{1,1})=sup_K(f'_{1,3})$, we easily get $sig_K(q_2)=\swap$ if and only $sig_K(q_3)=\swap$ which by $sup_K(f'_{2,1})=sup_K(f'_{2,7})$ implies  $sup_K(f'_{2,2})=sup_K(f'_{2,6})$, too.
Now, if $sig_K(z)\not=\swap$, then, by $sig_K(q_0)=sig_K(q_1)=\swap$, we have $sup_K(f'_{2,2})\not=sup_K(f'_{2,3})=sup_K(f'_{2,4})\not=sup_K(f'_{2,5})=sup_K(f'_{2,6})$ which implies $sig_K(z)=\nop$.
By $sig_K(z)\in \{\swap, \nop\}$ we have $sup_K(d_{j,2})=sup_K(d_{j,4})$ and $sup_K(g_{j,2})=sup_K(g_{j,4})$ which, with the occurrences of $p_j$ and $y_j$, implies that $sup_K(d_{j,1})=sup_K(d_{j,5})=sup_K(d_{j,6})$ and $sup_K(g_{j,1})=sup_K(g_{j,5})=sup_K(g_{j,6})$.
Hence, we have $(W\cup Acc)\cap sig^{-1}_K(\swap)=\emptyset$.

To prove the existence of an announced key region of $K^\sigma_\varphi$ for $\sigma\in \{\sigma_5, \sigma_6\}$ we, firstly, define the following subsets and, secondly, show how they are to composed to a corresponding region:
\begin{enumerate}
%Head
\item
$S^{\sigma_5}_0=\{h'_{j,2},h'_{j,5}\mid j\in \{0,\dots, 3m-1\}\}$ and $S^{\sigma_6}_0= S(H^{\sigma_6})\setminus S^{\sigma_5}_0$,
%freezer F_K
\item 
$S^{\sigma_5}_1=\{f'_{0,2},f'_{0,5}, f'_{0,8}, f'_{1,2}, f'_{1,5}, f'_{2,2}, f'_{2,5}, f'_{2,6}, f'_{2,9} \}$ and  $S^{\sigma_6}_1= S(F^{\sigma_6}_K)\setminus S^{\sigma_5}_1$,
%duplicator
\item 
$S^{\sigma_5}_2=\{d_{j,2}, d_{j,3}, d_{j,4}, d_{j,8} \mid j\in \{0,\dots, 18m-1\}\}$ and $S^{\sigma_6}_2=S(D^{\sigma_6})\setminus S^{\sigma_5}_2$,
%generator
\item 
$S^{\sigma_5}_3=\{g_{j,2}, g_{j,3}, g_{j,4}, g_{j,8} \mid j\in \{0,\dots, 3m-1\}\}$ and $S^{\sigma_6}_3=S(G^{\sigma_6})\setminus S^{\sigma_5}_3$,
%support
\item 
for $\sigma\in \{\sigma_5, \sigma_6\}$: $sup^\sigma_K=S(K^\sigma_m)\cap (S^\sigma_0\cup S^\sigma_1\cup S^\sigma_2\cup S^\sigma_3)$.
\end{enumerate} 

For $\sigma\in \{\sigma_5,\sigma_6\}$ the set $sup^\sigma_K$ allows the signature $sig^\sigma_K$, where for $e\in E(K^\sigma_m)$ we have that $sig^\sigma_K(e)=$
\[
\begin{cases}
\textsf{op}^\sigma_k \text{ if } e=k, \\
\swap \text{ if } e\in \{v_0,\dots, v_{3m-1}\}\cup \{ m, q_0,q_1,q_2,q_3, \_\}\cup \{p_0,\dots, p_{18m-1}, y_0,\dots, y_{3m-1}\} \\
\nop \text{ otherwise } 
\end{cases}
\]
Clearly, the region $(sup^\sigma_K, sig^\sigma_K)$ inhibits $k$ at $h'_{0,2}$ and satisfies the condition of the lemma, cf. Figure~\ref{fig:key_unions}.
Hence, the lemma is proven.
\end{proof}

%%%%%%%%%%%%%%%%%%%%%%%%%%%%%%%%%%%%%%%%%%%%%%%%%%%%%%%%%%%%%%%%%%%%%%%%%%%%%%%%%%%%%%%%%%%%%%%%%%%%%%%%%%

\section{Concluding the ESSP and the SSP from a Key Region}\label{sec:secondary_proofs}
%SECONDARY PROOFS \sigma_1,...,\sigma_6

In this section we show for $\sigma\in \{\sigma_1,\dots, \sigma_6\}$ and $\tau\in \sigma$ that the inhibition of the key event at the key state in $U^\sigma_\varphi$ by a $\tau$-region implies the ESSP and the SSP for $U^\sigma_\varphi$ with respect to $\tau$.
In our reduction, events of the same kind are numbered from $0$ up to $n$, for an $n\in \mathbb{N}$, for example $c_0,\dots, c_{6m-2}$ or $y_0,\dots, y_{3m-1}$.
We occasionally refer to a subset of such numbered events in form of $\{e_{i-1},\dots, e_{j+1}\}\subseteq \{e_0,\dots, e_n\}$.
Of course, if $i=0$ or $j=n$ then the there are no events $e_{i-1}, e_{j+1}$.
However, regarding these cases separately renders a lot of case analyses and makes the proofs much more tedious without any additional insight.
Hence, for the sake of readability we refrain from such a case analyses and bid the reader to accept this little inaccuracy and consider such $e_{i-1}, e_{j+1}$ as not listed.

\subsection{Concluding the ESSP and the SSP for $\sigma_1,\sigma_2,\sigma_3,\sigma_4$}
%SECONDARY PROOFS \sigma_1,...,\sigma_4

In this section, we show for $\sigma\in \{\sigma_1,\dots, \sigma_4\}$ that $U^\sigma_\varphi$ has the (E)SSP if $k$ is inhibitable at $s_{key}$ in $U^\sigma_\varphi$.
Our approach for the ESSP is as follows:
Let $s\in \bigcup_{i=1}^{4} S(U^{\sigma_i}_\varphi)$ be a state and $e\in \bigcup_{i=1}^{4} E(U^{\sigma_i}_\varphi)$ be an event such that $\neg s\edge{e}$. 
We present a subset of states $S\subseteq \bigcup_{i=1}^{4} S(U^{\sigma_i}_\varphi)$ such that for $i\in \{1,\dots, 4\}$ the set $sup_i=S\cap S(U^{\sigma_i}_\varphi)$ is a support of $U^{\sigma_i}_\varphi$ that for $\tau\in \sigma_i$ allows a signature $sig_i$ such that $(sup_i, sig_i)$ is a $\tau$-region that inhibits $e$ at $s$ in $U^{\sigma_i}_\varphi$ if this is necessary.
Frequently, it is possible to inhibit an event or many events at different states simultaneously with the same region.
Hence, for the sake of readability we present tables with the columns '$E$', 'Support', 'Target States' with the following meanings:
\begin{enumerate}
\item
$E$: Here the set of events are presented which are inhibited at the target states by the current region.
\item
Support: 
Here the set $S\subseteq \bigcup_{i=1}^{4} S(U^{\sigma_i}_\varphi)$ is presented. 
\item
Target States: 
Here a set of states $S'\subseteq \bigcup_{i=1}^{4} S(U^{\sigma_i}_\varphi)$ is presented, such that for $\sigma\in \{\sigma_1,\dots, \sigma_4\}$ all events of $E$ are inhibited at $S'\cap S(U^\sigma_\varphi)$ by the appropriate region with the support $S\cap S(U^\sigma_\varphi)$.
Clearly, $S\cap S'=\emptyset$.
\end{enumerate}
Having a support $sup$, it remains to present an appropriate signature $sig$ allowed by $sup$.
Instead of representing the signature for each support and each $\tau$ explicitly we rather use again a general scheme that works for \emph{almost} all ESSP atoms.
More exactly, given a set $S$ defined by a certain row of the table the implied support $sup$ allows a $\tau$-signature $sig$ such that for each $\sigma\in \{\sigma_1,\dots,\sigma_4\}$, $\tau\in \sigma$ and $e\in E(U^\sigma_\varphi)$ it holds:
	\[ sig(e)=
\begin{cases}
\inp, & \text{if } e\in E \text{ or } \swap\not\in\tau \text{ and } sup(s)=1 \text{ and } sup(s')=0\\
\swap, & \text{if } sup(s)\not=sup(s') \text{ and } \swap\in \tau \\
\set\ (\out) & \text{if } sup(s)=0 \text{ and } sup(s') =1 \text{ and } \swap\not\in\tau \text{ and } \set\in\tau\ (\set\not\in\tau) \\
\nop, &\text{if } sup_i(s)=sup(s')\\
\end{cases}
\]
Note that, by $S\cap S'=\emptyset$ and $sig(e)=\inp$ for all $e\in E$ such a region actually inhibits all events of $E$ at all states of $S'$.
As already mentioned, some ESSP atoms $\{s,e\}$ requires a special treatment and need to be discussed individually.
However, these cases are very seldom and they will be discussed at the appropriate place.

If the ESSP for $U^\sigma_\varphi$ is proven then it remains to argue for the SSP.
This will be done at the very end of this section in Lemma~\ref{lem:ssp_for_sigma_1_to_sigma_4}.

\begin{lemma}\label{lem:} The key event is inhibitable. \end{lemma}
\begin{proof}
An input key region inhibits $k$ already in $K^\sigma_m$ except for the states $f_{0,2},f_{1,0}$ and at the relevant states of $F^{\sigma_3}_T$ and $F^{\sigma_4}_T$.
Therefore, it only remains to show that $k$ is inhibitable at the relevant states of $U(T^\sigma_0,\dots, T^\sigma_{m-1})$ and at $f_{0,2}$.
The first row of the next table deals with the inhibition of $k$ at the states of $U(T^\sigma_0,\dots, T^\sigma_{m-1})$ and at $f_{0,2}, f_{1,0}$.
Here, for $\sigma_3$, we deviate slightly from the construction rules for the signature, that is, for this case the event $n$ is mapped by $sig$ to $\res$: $sig(n)=\res$.
The second row of the table is dedicated to the inhibition of $k$ at $f_{1,0}$ and, therefore, concerns exclusively $U^{\sigma_1}_\varphi$ and $U^{\sigma_2}_\varphi$. 
That $sig(k)=\inp$ implies for all generators $G^{\eta,\varrho}_j$ installed by the respective union $U^\sigma_\varphi$ that exactly the source states of the  $k$-labeled transition has to be included by the support.
For readability, the table does not enumerate these states explicitly, but there are assumed to be included.

\linespread{1.4}
\begin{longtable}{ p{0.5cm} p{9.7cm}   p{2cm}  }
$E$ & Support & Target States   \\ \hline
%row 1
$\{k\}$ & 
\raggedright{
$ \bigcup_{i=0}^{6m-1}\{h_{i,0},h_{i,3},h_{i,6}\}, \bigcup_{i=0}^{3m-1}\{h_{i,2}\}, \bigcup_{i=3m}^{6m-1}\{h_{i,5}\}$,
 \newline $\{t_{i,\alpha,0}\mid 0\le i\le m-1, 0\le \alpha\le 2\}, \{f_{1,0}, f_{1,1}\}$, $\{f_{0,0}, f_{0,3}, f_{2,0}, f_{2,3}\}$, 
}
& $S(T^\sigma_\varphi), f_{0,2}$\\ \hline
% row 2
$\{k\}$ & 
\raggedright{
$\bigcup_{i=0}^{6m-1}\{h_{i,0},h_{i,3},h_{i,6}\}, \bigcup_{i=3m}^{6m-1}\{h_{i,2}\} , \bigcup_{i=0}^{3m-1}\{h_{i,5}\}$,
$\{f_{0,0}, f_{0,2}$, $f_{0,3}, f_{1,0}\}$, $S(U(T^{\sigma_1}_0,\dots, T^{\sigma_1}_{m-1}))\setminus\{ t_{i,\alpha,1} \mid 0\le i\le m-1, 0\le \alpha \le 2\}$ 
}
& $f_{1,0}$ \\ 
\end{longtable}
\end{proof}
%%%%the quite events
%%%%%%%%%%%%
\begin{lemma} The events $q_0,\dots, q_{3m-1}$ are inhibitable. \end{lemma}
\begin{proof}
Let $i\in \{0,\dots, m-1\}$ and $\ell\in \{0,\dots,2\}$ such that $j=3i+\ell \in \{0,\dots, 3m-1\}$.
The regions of the first two rows prove $q_j$ to be inhibitable in the TSs it occurs in, the last row is dedicated to the states of the other TSs.
\linespread{1.2}
\begin{longtable}{ p{0.5cm} p{7cm}   p{4.5cm}  }

$E$ & Support & Target States   \\ \hline
%first row
$\{q_j\}$ 
&
\raggedright{
%head 
$\bigcup_{i=0}^{6m-1}\{h_{i,1},h_{i,4}\}, \bigcup_{i=0}^{3m-1}\{h_{i,2}\}, \bigcup_{i=3m}^{6m-1}\{h_{i,5}\}$,  
$\bigcup_{i=0}^{6m-2}\{g^{c,c}_{i,2}, g^{c,c}_{i,3}\}$, $\bigcup_{i=0}^{3m-1}\{g^{\_,q}_{i,1}, g^{\_,q}_{i,2}\}$, $\bigcup_{i=0}^{3m-1}\{g^{\_,y}_{i,2}, g^{\_,y}_{i,3}\}$, $\bigcup_{i=0}^{m-1}\{g^{x,\_}_{i,2}, g^{x,\_}_{i,3}, g^{\_,x}_{i,2}, g^{\_,x}_{i,3}\}$, 
$\{f_{0,1}, f_{0,4},f_{1,0}, f_{1,2},  f_{2,0}, f_{2,3},  g^{n,\_}_{0,0},g^{n,\_}_{0,3} \}$,
$\{t_{i,0,1},t_{i,1,1},t_{i,2,1}\mid 0\le i \le m-1\}$
} &
\raggedright{$\bigcup_{i=0}^{6m-1}\{h_{i,0}, h_{i,3} \}, \bigcup_{i=3m}^{6m-1}\{h_{i,2} \}$, $\bigcup_{i=0}^{3m-1}\{g^{\_,q}_{i,0}, g^{\_,q}_{i,3}\}$
}\arraybackslash \\ \hline					
%second row
$\{q_j\}$ 
& 
%head
$ \{h_{j,3},h_{j,4},h_{j+3m,2},h_{j+3m,6} \}$, 
\raggedright{
%Q-generator%C-generator
$\{g^{\_,q}_{j,0}, g^{\_,q}_{j,2}, g^{c,c}_{j-1,1}, g^{c,c}_{j-1,3}, g^{c,c}_{j,0}, g^{c,c}_{j,2}\}$, 
% F
if $j=0:\{f_{1,0}\}$ ,
%translator
$\{t_{i,\ell,2},\dots, t_{i,\ell,5}\}$
} 
& 
%inhibit
\raggedright{remaining of $S(H)$, $\bigcup_{j=0}^{3m-1}S(G^{\_,q}_j)\setminus~\{g^{\_,q}_{j,0}\}$}\arraybackslash \\ \hline%%%%%%%%%%%%%%%%%%%%
$\{q_j\}$ &
%head 
$\{h_{j,0},\dots, h_{j,4}, g^{\_,q}_{j,0}, g^{\_,q}_{j,2}\}$, if $j=0:\{f_{1,0}\}$ &
%inhibit
remaining states \\
\end{longtable}
\end{proof}
%%%%%%%%%the zero events
\begin{lemma} The events $Z=\{z_0,\dots, z_{3m-1}\}$ are inhibitable. \end{lemma}
\begin{proof}
The first row of the following table is dedicated to the inhibition of $Z$ at certain states of $H$ and the sources/sinks of $k$ in $F_2$, $G^{n,\_}_0$ and $G^{\_,q}_0,\dots,G^{\_,q}_{3m-1}$.
Hence, for $F_2$ and each generator $G^{\eta,\varrho}_j$ installed by $U^\sigma_\varphi$ we, firstly, assume the sinks $g^{\eta,\varrho}_{j,2}, g^{\eta,\varrho}_{j,3}$ and $f_{2,2},f_{2,3}$ of $k$ to be included by sup.
Secondly, to assure the inhibition of $Z$ at all states of $F_2$, $G^{n,\_}_0$ and $G^{\_,q}_0,\dots, G^{\_,q}_{3m-1}$, too, we use almost the same region but include now for $F_2$, $G^{n,\_}_0$ and $G^{\_,q}_0,\dots, G^{\_,q}_{3m-1}$ the sources of $k$ instead of its sinks.
For the sake of readability we refrain from the explicit presentation of these states.
The second row is about the inhibition of $z_0$ at the states of  $F_1$ and concerns only the switch $\sigma_1$.
Finally, the third row proves for $j\in \{0,\dots, 3m-1\}$ the inhibition of $z_j$ at the remaining states.
\linespread{1.4}
\begin{longtable}{ p{0.5cm} p{9cm}   p{3cm}  }

$E$ & Support & Target States   \\ \hline
$Z$ & 
%head
\raggedright{%head
$\bigcup_{i=0}^{6m-1} \{h_{i,1}, h_{i,4}\}$, $\bigcup_{i=3m}^{6m-1} \{h_{i,2}\}$ , $\bigcup_{i=0}^{3m-1} \{h_{i,5}\}$, 
%freezer F_0
$\{f_{0,1},f_{0,2}, f_{1,2}\}$
}& 
%target states
\raggedright{
$\bigcup_{i=0}^{6m-1}\{h_{i,0}\}$, $\{f_{0,0},f_{0,3}\}$, $S(T^\sigma_\varphi)$
}\arraybackslash \\ \hline
$\{z_0\}$
&
\raggedright{%
$\{h_{0,0},h_{0,1}, h_{0,3}, h_{0,4}, h_{0,5}, h_{3m,2}, h_{3m,6}, g^{c,c}_{3m-1,2}, g^{c,c}_{3m-1,3}\}$, $\{g^{c,c}_{3m,0}, g^{c,c}_{3m,1}\}$, $\{t_{0,\alpha,2},\dots, t_{0,\alpha,5} \mid 0\le \alpha\le 2\}$}
&
%target states
$\{f_{1,0},f_{1,1},f_{1,2}\}$
\\ \hline
$\{z_j\}$ &
\raggedright{%
$\{h_{j,0},h_{j,1},h_{j,5}\}$, $j=0: \{f_{0,2}, f_{1,1},f_{1,2},  f_{2,1}, f_{2,3}, g^{n,\_}_{0,1}, g^{n,\_}_{0,3}\}$}
& 
 %target states
\text{ remaining states }
\end{longtable}
\end{proof}
%%%the Y events 
\begin{lemma} The events $y_0,\dots, y_{3m-1}$ are inhibitable. \end{lemma}
\begin{proof}
The first row of the following table is, firstly, dedicated to the inhibition of $y_j$ at $h_{j,3}$ and $g^{\_,y}_{j,0}, g^{\_,y}_{j,1}$.
To do so, for $F_2$ and all generators $G^{\eta,\varrho}_j$ installed by $U^\sigma_\varphi$ we assume the sinks (sources) of $k$-labeled transitions to be included (excluded) and refrain from the explicit presentation of these states. 
To inhibit $y_j$ at $g^{\_,y}_{j,2}, g^{\_,y}_{j,3}$, too, we use almost the same region but include now all the sources $g^{\_,y}_{j,0}, g^{\_,y}_{j,1}$ of $k$ and exclude its sinks $g^{\_,y}_{j,2}, g^{\_,y}_{j,3}$.
If $j\in \{0,\dots, 3m-1\}$ then the following table proves $y_j$ to be inhibitable.

\linespread{1.4}
\begin{longtable}{ p{0.5cm} p{9cm}   p{2.5cm}  }

$E$ & Support & Target States   \\ \hline
$Y$ & 
\raggedright{$\bigcup_{i=0}^{6m-1}\{h_{i,1}, h_{i,4}\}, \bigcup_{i=3m}^{6m-1}\{h_{i,2}\}, \bigcup_{i=0}^{3m-1}\{h_{i,5}\}$ $\{f_{0,1}, f_{0,2}, f_{0,4}\}$, $\{f_{1,2}\}$, $S(U(T^\sigma_0,\dots,T^\sigma_{m-1}))\setminus\{t_{i,0,1}, t_{i,1,1}, t_{i,2,1}\mid 0\le i \le m-1\}$} 
& $\bigcup_{i=3m}^{6m-1}\{h_{i,3}\}$ \\ \hline
$\{ y_j \}$ & $\{ h_{j+3m, 3}, h_{j+3m, 4}\}, \{g^{\_,y}_{0,0}, g^{\_,y}_{0,2}\}$ & remaining states
\end{longtable}
\end{proof}
%%%%%the consistency and reachability and placeholder events
\begin{lemma} The events $c_0,\dots, c_{6m-2}$ and $r_0,\dots, r_{6m-2}$ and $p_0,\dots, p_{3m-1}$ are inhibitable. \end{lemma}
\begin{proof}

If $j\in \{0,\dots, 3m-1\}$, $i \in \{0,\dots, m-1\}$ and $\ell \in \{0,\dots, 2\}$ such that $j=3i+\ell$ then the first row of the following table, firstly,  proves $c_j, r_j$ and $c_{j+3m}, r_{j+3m}$ at $h_{j,5},h_{j+3m,5}$, that is, $c_k, r_k$ for all $k\in \{0,\dots, 6m-2\}$ to be inhibitable at $h_{k,5}$.
Moreover, these region inhibits the event $p_j$ at the states $h_{j+3m,0}, h_{j+3m,1}$, too.

The second row, is dedicated to the inhibition of $c_j, r_j$ at the states $S(H)\setminus\{h_{j,5}\}$.
Here, if $j\in \{0,\dots, 3m-1\}$ we assume $g^{\_,q}_{j,1}, g^{\_,q}_{j,3}$ to be included and if $j\in \{3m,\dots, 6m-2\}$ we include $g^{\_,y}_{j,1}, g^{\_,y}_{j,3}$.
For simplicity, we refrain from presenting these states explicitly.

The third row, inhibits the events  $c_j,r_j$ ($p_j$) for $j\in \{0,\dots, 6m-2\}$ ($j\in \{0,\dots, 3m-1\}$) at the remaining states of $U^\sigma_\varphi$.

\linespread{1.4}
\begin{longtable}{ p{1.6cm} p{8.3cm}   p{2.4cm}  }

$E$ & Support & Target States   \\ \hline
%first row
\raggedright{$\{c_j,r_j, c_{j+3m},$ $r_{j+3m}, p_j\}$}  
& 
\raggedright{%head
$\{h_{j,2}, h_{j,6}, h_{j+3m,0},h_{j+3m,1}, h_{j+3m,5}\}$, 
%generator
$\{g^{c,c}_{j-1,1}, g^{c,c}_{j-1,3}\}$, $\{g^{c,c}_{j+3m-1,1}, g^{c,c}_{j+3m-1,3}\}$, $\{g^{c,c}_{j,0}, g^{c,c}_{j,2}\}$, $\{g^{c,c}_{j+3m,0}, g^{c,c}_{j+3m,2}\}$,
%freezer
 $\text{if }j=0: \{f_{0,3}, f_{0,4}\}$
%translator
$\{t_{i,\ell,0}, t_{i,\ell,1}\}$,
}  
& 
%target states
\raggedright{$\{h_{j,5}, h_{j+3m,5}\}$, $\{h_{j+3m,0}, h_{j+3m,1}\}$}\arraybackslash  \\ \hline%%%%%%%%%%%%%%%%%%%%%%%%%%%%%%%%%%%%%%%%%%%%%%%%%%%%%%%%%%%%%%%%%%%
%second row
$\{c_j,r_j\}$  
& 
\raggedright{%head
$\{ h_{j,5},h_{j,6}\}$, 
%generator
$ \{g^{c,c}_{j-1,1}, g^{c,c}_{j-1,3}\}, \{g^{c,c}_{j,0}, g^{c,c}_{j,2}\}$,  $\text{if }j=0: \{f_{1,2}, f_{1,3}\}$}
& 
%target states 
$S(H)\setminus\{h_{j,5}\}$\\ \hline%%%%%%%%%%%%%%%%%%%%%%%%%%%%%%%%%%%%%%%%%%%%%%%%%%%%%%%%%%%%%%
%third row
$\{c_j,r_j\}$  
& 
$\bigcup_{i=0}^{j} \{h_{j,0},\dots, h_{j,6}\}, \{g^{c,c}_{j,0}, g^{c,c}_{j,2}\}$ 
& 
%target states
remaining states\\ \hline%%%%%%%%%%%%%%%%%%%%%%%%%%%%%%%%%%%%%%%%%%%%%%%%%%%%%%%%%%%%%%
%fourth row
$\{p_j\}$  
& 
$\{h_{j+3m,0},\dots, h_{j+3m,2}\}$
& 
%target states
remaining states%
%%%%%%%%%%%%%%%%%%%%%%%%%%%%%%%%%%%%%%%%%%%%%%%%%%%%%%%%%%%%%
\end{longtable}
\end{proof}
%%%%%%%%%%the event n_0
\begin{lemma}\label{lem:} The event $n_0$ is inhibitable. \end{lemma}
\begin{proof}
The following table proves $n_0$ to be inhibitable in $U^\sigma_\varphi$:
\linespread{1.4}
\begin{longtable}{ p{0.5cm} p{8cm}   p{3cm}  }
$E$ & Support & Target States   \\ \hline
%row 1
$\{n_0\}$ & 
\raggedright{%
$\bigcup_{i=0}^{6m-1}\{h_{i,1},h_{i,4}\}, \bigcup_{i=0}^{3m-1}\{h_{i,2}\}, \bigcup_{i=3m}^{6m-1}\{h_{i,5}\}$,
 \newline $\{t_{i,\alpha,1}\mid 0\le i\le m-1, 0\le \alpha\le 2\}, \bigcup_{j=0}^{3m-1}\{g^{\_,q}_{0,0}, g^{\_,q}_{0,2} \}$, $\{f_{0,1}, f_{0,4}, f_{1,2}, f_{2,0}, f_{2,3}, g^{n,\_}_{0,0}, g^{n,\_}_{0,3}\}$, 
}
& $\{f_{0,0},f_{0,4}, f_{1,0}\}$, $\{f_{1,1}, g^{\_,n}_{0,2},f_{2,2}\}$\\ \hline
% row 2
$\{n_0\}$ & 
\raggedright{%
$\{f_{0,0}, f_{0,1},f_{2,0}, f_{2,2}, g^{\_,n}_{0,0}, g^{\_,n}_{0,2}\}$
}
& remaining states\\ 
\end{longtable}
\end{proof}

%%the vice and wire events are inhibitable
\begin{lemma}\label{lem:vice_and_wire} The events of $V=\{v_0,\dots, v_{3m-1}\}$ and $W=\{w_0,\dots, w_{3m-1}\}$ are inhibitable. \end{lemma}
\begin{proof} 
We proof the lemma by presenting regions for an arbitrary $i\in \{0,\dots, m-1\}$ that inhibits the events $v_{3i},\dots, v_{3i+2}, w_{3i},\dots, w_{3i+2}$ at all relevant states of $U^\sigma_\varphi$.
Let $i\in \{0,\dots, m-1\}$ and $\alpha\in \{0,1,2\}$.
The first row is dedicated to the inhibition of $x_{3i+\alpha}, w_{3i+\alpha}$ at $t_{i,\alpha,0}$ and the second row proves both of them to be inhibitable at $S(U^\sigma_\varphi )\setminus sup$, where $sup$ are the states in the corresponding table cell.
The third (fourth) row deals with the inhibition of $x_{3i},x_{3i+1},x_{3i+2}$ ($w_{3i},w_{3i+1},w_{3i+2}$) at the remaining states, that is, the rest of $\bigcup_{j=3m}^{6m-1}\{h_{j,0,\dots, h_{j,6}}\}$ ($\bigcup_{j=0}^{3m-1}\{h_{j,0,\dots, h_{j,6}}\}$),  $h_{3i+\alpha,1}$ and $G^{\_,y}_0,\dots, G^{\_,y}_{3m-1}$.
We achieve these regions as follows:
Let $j,\ell\in \{0,\dots, m-1\}\setminus\{i\}$ such that $X_{i,0}\in E(T^\sigma_j)\cap E(T^\sigma_\ell)$ and let for for $n\in \{i,j,\ell\}$ be $\alpha_n=0$ ($\alpha_n=1, \alpha_n=2$) if $X_{n,0}=X_{i,0}$ ($X_{n,1}=X_{i,0}, X_{n,2}=X_{i,0}$) and let $\beta_n=(\alpha_n+1)\text{mod}3$ and $\gamma_n=(\alpha_n+2)\text{mod}3$. 
The states $t_{n,\alpha_n,3}$, $t_{n,\beta_n,5}$ and $t_{n,\gamma_n,4}$ are the sinks and $t_{n,\alpha_n,2}$, $t_{n,\beta_n,4}$ and $t_{n,\gamma_n,3}$ the sources of $X_{i,0}$ in $T^\sigma_n$.
Having this insight, we now can define the following subsets of $S(U^\sigma_\varphi)$ to, finally, combine them to a fitting support of $U^\sigma_\varphi$:

\begin{enumerate}
%translators for $x_i$ using $X_{i,0}$
\item
$S_0=\{t_{n,\alpha_n,0},t_{n,\alpha_n,1}, t_{n,\beta_n,0},t_{n,\beta_n,1}, t_{n,\gamma_n,0},t_{n,\gamma_n,1}\mid n\in \{i,j,\ell\}\}$, 
\item
$S_1=\{t_{n,\alpha_n,3},\dots, t_{n,\alpha_n,5}, t_{n,\beta_n,5},t_{n,\gamma_n,4},\dots, t_{n,\gamma_n,5} \mid n\in \{i,j,\ell\}\}$
\item%head for x_i
 $S_2=\{h_{3n,2}, h_{3n+1,2},h_{3n+2,2},h_{3n,6}, h_{3n+1,6},h_{3n+2,6},\mid n\in \{i,j,\ell\}\}$
\item%x-freezer
$S_3=\{ g^{x,\_}_{n,0}, g^{x,\_}_{n,2}, g^{\_,x}_{n,0}, g^{\_,x}_{n,2} \mid  n\in \{0,\dots, m-1\}: x_n=x_{i,0}\}$
\item %translators for $w_i$ using $X_{i,0}$
$S_4=\{t_{n,\alpha_n,2},t_{n,\beta_n,2},\dots, t_{n,\beta_n,4}, t_{n,\gamma_n,2},t_{n,\gamma_n,3}\mid n\in \{i,j,\ell\}\}$, 
%head for w_i
\item $S_5=\{h_{3n+3m,p}, h_{3n+3m+1,p}, h_{3n+3m+2,p} \mid n\in \{i,j,\ell\}, p\in \{0,1,3,4, 5\}\}$
\end{enumerate}
Note that, for the third row, we rather have need of some further states of $S(U(G^{c,c}_0,\dots, G^{c,c}_{6m-2}))$ because of some incoming/outgoing c-labeled transitions in $H$.
Depending on $i,j,\ell$ the actual presentation of these states would render a lot of case analyses as, for example, $j=i+1$, $j=\ell+1$.
However, we can always choose the fitting states of the corresponding TSs and, therefore, we do not represent these states explicitly.
\linespread{1.4}
\begin{longtable}{ p{2cm} p{7cm}   p{3cm}  }
$E$ & Support & Target States   \\ \hline
%row 1						
$V, W$ 
& 
\raggedright{%head
$\bigcup_{i=0}^{6m-1}\{h_{i,0},h_{i,4}\}, \bigcup_{i=0}^{3m-1}\{h_{i,2}\}$, $\bigcup_{i=3m}^{6m-1}\{h_{i,5}\}$,
%freezer
$\{f_{0,1},f_{0,4},f_{1,0},f_{1,2},f_{2,1},f_{2,2}\}$,
$\{g^{n,\_}_{0,1}, g^{n,\_}_{0,2}\}$, $\bigcup_{i=0}^{6m-2}\{g^{c,c}_{i,2}, g^{c,c}_{i,3}\}$, $\bigcup_{i=0}^{3m-1}\{g^{\_,y}_{i,2}, g^{\_,y}_{i,3}, g^{\_,q}_{i,1}, g^{\_,q}_{i,2}\}$ ,   $\bigcup_{i=0}^{m-1}\{g^{\_,x}_{i,2}, g^{\_,x}_{i,3}, g^{x,\_}_{i,2}, g^{x, \_}_{i,3}\}$, 
%translators
$\{t_{i,0,1}, t_{i,1,1}, t_{i,2,1} \mid 0\leq i \leq m-1 \}$} &
%target states
\raggedright{$\{t_{i,0,0}, t_{i,1,0}, t_{i,2,0}\mid 0\le i\le m-1\}$}\arraybackslash \\ \hline
%row 2
$\{v_{3i+\alpha}, w_{3i+\alpha}\}$ 
& 
\raggedright{%head
$\{h_{{3i+\alpha},0},\dots,h_{{3i+\alpha},2}, h_{{3i+\alpha}+3m,0},h_{{3i+\alpha}+3m,1}$, $h_{{3i+\alpha}+3m,5}\}$, 
%y-generator
$\{g^{\_,y}_{3i+\alpha, 1}, g^{\_,y}_{3i+\alpha,3}, 
%translators
t_{i,\alpha,0},t_{i,\alpha,1}\}$} 
&
% target states
\raggedright{$S(U^\sigma_\varphi )\setminus sup$}\arraybackslash \\ \hline
$\{v_{3i},\dots,v_{3i+2}\}$ & 
\raggedright{$S_0,S_2, S_3$, 
 if $\{i,j,\ell\}\cap\{0\}\not=\emptyset: f_{0,3},f_{0,4}$},   
if $\sigma=\sigma_3$ then $S_4$, otherwise, 
 for $\sigma_1,\sigma_2,\sigma_4: S_1$
 &
 \raggedright{
remaining states 
 } \arraybackslash \\ \hline

$\{w_{3i},\dots,w_{3i+2}\}$ 
&
$S_0, S_3, S_5$, for $\sigma=\sigma_3: S_1$, for $\sigma_1,\sigma_2,\sigma_4: S_4$
&
remaining states
\end{longtable}
\end{proof}
%the variable and corresponding events are inhibitable
\begin{lemma} The events $X_0,\dots,X_{m-1}$ and $x_0,\dots,x_{m-1}$ are inhibitable. \end{lemma}
\begin{proof}
For arbitrary $i \in \{0,\dots, m-1\}, \alpha\in \{0,1,2\}$ we present regions of $U^\sigma_\varphi $ that inhibits the event $X_{i,\alpha}$  and $x_{i,\alpha}$ at all states of $S(U^\sigma_\varphi)\setminus \{s\in S(T_j)\mid j\not=i, X_{i,\alpha}, x_{i,\alpha} \in E(T_j)\}$.
Having this, by the arbitrariness of $i$ and $\alpha$ this implies that $X_0,\dots, X_{m-1},x_0,\dots,x_{m-1}$ are inhibitable in $U^\sigma_\varphi $.

Let $i \in \{0,\dots, m-1\}, \alpha_i\in \{0,1,2\}$ arbitrary but fixed.
Let $j,\ell\in \{0,\dots, m-1\}\setminus \{i\}$ such that $X_{i,\alpha_i}\in E(T^\sigma_{j})\cap E(T^\sigma_{\ell})$ and, for $n\in \{j,\ell\}$, let $\alpha_n=0$ ($\alpha_n=1, \alpha_n=2$) if $X_{n,0}=X_{i,\alpha_i}$ ($X_{n,1}=X_{i,\alpha_i}, X_{n,2}=X_{i,\alpha_i}$).
Finally, let $\beta_n=(\alpha_n+1)\text{mod}3$ and $\gamma_n=(\alpha_n+2)\text{mod}3$ for $n\in \{i,j,\ell\}$.
For $n\in \{i,j,\ell\}$, we get:
\begin{enumerate}
\item
The states $t_{n,\alpha_n,2}, t_{n,\beta_n,4}, t_{n,\gamma_n,3}$ are the sources of $X_{i,\alpha_i}$ and the sinks of $x_{i,\alpha_i}$ in $T_n$.
\item
The states $t_{n,\alpha_n,3}, t_{n,\beta_n,5}, t_{n,\gamma_n,4}$ are the sinks of $X_{i,\alpha_i}$ and the sources of $x_{i,\alpha_i}$ in $T_n$. 
\end{enumerate}
We now define subsets of $S(U^\sigma_\varphi)$ which will be used to combine supports of regions of $U^\sigma_\varphi$
\begin{enumerate}
\item $S_0=\{t_{n,0,0}, t_{n,1,0},t_{n,2,0}, t_{n,0,1}, t_{n,1,1},t_{n,2,1}\mid n\in \{i,j,\ell\} \}$
\item $S_1=\{t_{n,\alpha_n,2}, t_{n,\beta_n,2},\dots,t_{n,\beta_n,4}, t_{n,\gamma_n,2}, t_{n,\gamma_n,3}\mid n\in \{i,j,\ell\}\}$
\item $S_2=\{g^{\_,x}_{n,0}, g^{\_,x}_{n,2}, g^{x,\_}_{n,0}, g^{x,\_}_{n,2}\mid n\in \{0,\dots, m-1\}: x_n=x_{i,\alpha_i}\}$,
\item $S_3=\{g^{\_,x}_{n,1}, g^{\_,x}_{n,3}, g^{x,\_}_{n,1}, g^{x,\_}_{n,3}\mid n\in \{0,\dots, m-1\}: x_n=x_{i,\alpha_i}\}$,
\item $S_4=\{g^{\_,q}_{3n,0}, g^{\_,q}_{3n,2}, g^{\_,q}_{3n+1,0}, g^{\_,q}_{3n+1,2},g^{\_,q}_{3n+2,0}, g^{\_,q}_{3n+2,2} \mid n\in \{i,j,\ell\} \}$,
\item $S_5=\{h_{3n,3}, h_{3n+1,3}, h_{3n+2,3}, h_{3n,4}, h_{3n+1,4}, h_{3n+2,4}\mid n\in \{i,j,\ell\}  \}$, 
\item $S_6=\{h_{3n+3m,n'}, h_{3n+3m+1,n'}, h_{3n+3m+2,n'}\mid n'\in \{0,1,3,4,5\}, n\in \{i,j,\ell\}  \}$, 
\end{enumerate}
The following table proves $X_{i,\alpha_i}$ to be inhibitable at all states in question besides of $t_{i,\beta_i,2}, t_{i,\beta_i,3}$ and $t_{i,\gamma_i,2}$.
Observe that, if $\tau\in \sigma_3$ then $\swap\in\tau$.
\linespread{1.4}
\begin{longtable}{ p{1cm}  p{8cm}  p{3cm} }
$E$ & Support  & Target States \\ \hline
%row 1
$X_{i,\alpha_i}$ 
&
$S_1,S_2,S_4,S_5$ and for $\sigma_1$ if $0\in \{i,j,\ell\}:f_{1,0}$ and for $\sigma_3:S_0$
& 
$S(U^\sigma_\varphi)\setminus sup$ \\ \hline
% row 2
$X_{i,\alpha_i}$ 
&
$S_1,S_3,S_6$ and for $\sigma_1,\sigma_2,\sigma_4: S_0$ 
& 
%target states 
$S(U^\sigma_\varphi)\setminus sup$ \\\hline
\end{longtable}
To prove that $X_{i,\alpha_1}$ is inhibitable at the remaining states $t_{i,\beta_i,2}, t_{i,\beta_i,3}$ and $t_{i,\gamma_i,2}$, too, we need to additionally involve the variable event $X_{i,\gamma_i}$, which, by the further occurrences of $X_{i,\alpha_i}, X_{i,\gamma_i}$ in other translators, renders a lot of case analysis.
However, going through all the cases is tedious and renders no considerable insight.
Hence, instead of presenting all the cases explicitly, we give a rather general instruction how a inhibiting region can be derived:
\begin{enumerate}
\item 
For $T^\sigma_i$ we put exactly the states $t_{i,\alpha_i,2}, t_{i,\alpha_i,5}, t_{i,\beta_i,4}, t_{i,\gamma_i,3}$ into the support which makes the arcs labeled with $v_{3i+\alpha_i}, w_{3i+\alpha_i}$ and $X_{i,\gamma_i},x_{i,\alpha_i}$ incoming and $x_{i,\gamma_i}$ outgoing.
\item 
If $j\in \{0,\dots, m-1\}\setminus \{i\}$ such that $X_{i,\alpha_i}\in E(T^\sigma_j)$ and $X_{i,\gamma_i}\not\in E(T^\sigma_j)$ then if for $\delta\in \{0,1,2\}$ and $\varepsilon\in \{2,3,4\}$ the state $t_{j,\delta,\varepsilon}$ is a source of $X_{i,\alpha_i}$ in $T^\sigma_j$ then put the states $t_{j,\delta,2},\dots, t_{j,\delta,\varepsilon}$ into the support. 
For $\sigma_1,\sigma_2,\sigma_4$ ($\sigma_3$) this makes $v_{3j}, v_{3j+1}, v_{3j+2}$ ($w_{3j}, w_{3j+1}, w_{3j+2}$) incoming  events.
\item 
If $j\in \{0,\dots, m-1\}\setminus \{i\}$ such that $X_{i,\gamma_i}\in E(T^\sigma_j)$ and $X_{i,\alpha_i}\not\in E(T^\sigma_j)$ then if for $\delta\in \{0,1,2\}$ and $\varepsilon\in \{3,4,5\}$ the state $t_{j,\delta,\varepsilon}$ is a sink of $X_{i,\gamma_i}$ in $T^\sigma_j$ then put the states $t_{j,\delta,\varepsilon},\dots, t_{j,\delta,5}$ into the support. 
For $\sigma_1,\sigma_2,\sigma_4$ ($\sigma_3$) this makes $w_{3j}, w_{3j+1}, w_{3j+2}$ ($v_{3j}, v_{3j+1}, v_{3j+2}$) incoming events.
\item 
If $j\in \{0,\dots, m-1\}\setminus \{i\}$ such that $X_{i,\gamma_i}\in E(T^\sigma_j)$ and $X_{i,\alpha_i}\in E(T^\sigma_j)$ then choose the $T^\sigma_j$-part of the support corresponding to as which of $X_{j,0},X_{j,1},X_{j,2}$ the events $X_{i,\alpha_i}$ and $X_{i,\gamma_i}$ occur, cf. Figure~\ref{fig:possibilities}.
This condemns at most some additional $v$-events or $w$-events to be incoming events.
\item
Choose, in dependence of the border crossing $w$-, $v$- or $x$- events the necessary additional states of $U^\sigma_\varphi$:
If $v_j$ is incoming then choose $h_{j,3},h_{j,4}$, making $q_j$ exiting, and $g^{\_,q}_{j,0}, g^{\_,q}_{j,2}$ and $f_{1,0}$ if $j=0$ to be included.
If $w_j$ is incoming then include $h_{j+3m,2,3}, h_{j+3m,2,6}$, making possibly $c_{j+3m-1}$ entering and $c_{j+3m}$ exiting, and $g^{c,c}_{j+3m-1,1}, g^{c,c}_{j+3m-1,3}, g^{c,c}_{j+3m,0}, g^{c,c}_{j+3m,2}$, respectively.
Finally, for $x_{i,\alpha_i}$ and $x_{i,\gamma_i}$ choose the support with their respective generator-\emph{sources} which, by the presence of \swap\ for the relevant cases is always fitting.
\end{enumerate}
\begin{figure}[t!]%verified possibilities for key region behavior of variable events in other translators
\centering
\begin{tikzpicture}[scale = 0.75]
\begin{scope}
\node at (-0.75,0) {a)};
\foreach \i in {0,...,11} {\coordinate (\i) at (\i*1.5cm,0);}
\foreach \i in {0, 3, 6, 9} {\fill[black!15, rounded corners] (\i) +(-0.5,-0.25) rectangle +(0.5,0.25);}
\foreach \i in {0,...,3} {\node (t\i) at (\i) {\scalebox{\nodeScale}{$t_{j,0,\i}$}};}
\foreach \i in {4,...,7} {\pgfmathparse{int(\i-4}  ,  \node (t\i) at (\i) {\scalebox{\nodeScale}{$t_{j,1,\pgfmathresult}$}};}
\foreach \i in {8,...,11} {\pgfmathparse{int(\i-8)} , \node (t\i) at (\i) {\scalebox{\nodeScale}{$t_{j,2,\pgfmathresult}$}};}
\graph{ 
(t0) ->[bend left=30, "\scalebox{\edgeScale}{$X_{i,\alpha_i}$}"] (t1) ->[bend left=30, "\scalebox{\edgeScale}{$X_{j,1}$}"] (t2) ->[bend left=30, "\scalebox{\edgeScale}{$X_{i,\gamma_i}$}"] (t3);
(t4) ->[bend left=30, "\scalebox{\edgeScale}{$X_{j,1}$}"] (t5) ->[bend left=30, "\scalebox{\edgeScale}{$X_{i,\gamma_i}$}"] (t6)->[bend left=30, "\scalebox{\edgeScale}{$X_{i,\alpha_i}$}"](t7);
 (t8)- >[bend left=30, "\scalebox{\edgeScale}{$X_{i,\gamma_i}$}"](t9)->[bend left=30, "\scalebox{\edgeScale}{$X_{i,\alpha_i}$}"](t10)->[bend left=30, "\scalebox{\edgeScale}{$X_{j,1}$}"](t11);
 (t0) <-[swap, bend right=30, "\scalebox{\edgeScale}{$x_{i,\alpha_i}$}"] (t1) <-[swap, bend right=30, "\scalebox{\edgeScale}{$x_{j,1}$}"] (t2) <-[swap, bend right=30, "\scalebox{\edgeScale}{$x_{i,\gamma_i}$}"] (t3);
(t4) <-[swap, bend right=30, "\scalebox{\edgeScale}{$x_{j,1}$}"] (t5) <-[swap, bend right=30, "\scalebox{\edgeScale}{$x_{i,\gamma_i}$}"] (t6)<-[swap, bend right=30, "\scalebox{\edgeScale}{$x_{i,\alpha_i}$}"](t7);
 (t8)<- [swap, bend right=30, "\scalebox{\edgeScale}{$x_{i,\gamma_i}$}"](t9)<-[swap, bend right=30, "\scalebox{\edgeScale}{$x_{i,\alpha_i}$}"](t10)<-[swap, bend right=30, "\scalebox{\edgeScale}{$x_{j,1}$}"](t11);
};
\end{scope}
\begin{scope}[yshift=-2.5cm]
\node at (-0.75,0) {b)};
\foreach \i in {0,...,11} {\coordinate (\i) at (\i*1.5cm,0);}
\foreach \i in {0,11} {\fill[black!15, rounded corners] (\i) +(-0.5,-0.25) rectangle +(0.5,0.25);}
\foreach \i in {2,5,8} {\fill[black!15, rounded corners] (\i) +(-0.5,-0.25) rectangle +(2,0.25);}
\foreach \i in {0,...,3} {\node (t\i) at (\i) {\scalebox{\nodeScale}{$t_{j,0,\i}$}};}
\foreach \i in {4,...,7} {\pgfmathparse{int(\i-4}  ,  \node (t\i) at (\i) {\scalebox{\nodeScale}{$t_{j,1,\pgfmathresult}$}};}
\foreach \i in {8,...,11} {\pgfmathparse{int(\i-8)} , \node (t\i) at (\i) {\scalebox{\nodeScale}{$t_{j,2,\pgfmathresult}$}};}
\graph{ 
(t0) ->[bend left=30, "\scalebox{\edgeScale}{$X_{i,\alpha_i}$}"] (t1) ->[bend left=30, "\scalebox{\edgeScale}{$X_{i,\gamma_i}$}"] (t2) ->[bend left=30, "\scalebox{\edgeScale}{$X_{j,2}$}"] (t3);
(t4) ->[bend left=30, "\scalebox{\edgeScale}{$X_{i,\gamma_i}$}"] (t5) ->[bend left=30, "\scalebox{\edgeScale}{$X_{j,2}$}"] (t6)->[bend left=30, "\scalebox{\edgeScale}{$X_{i,\alpha_i}$}"](t7);
 (t8)- >[bend left=30, "\scalebox{\edgeScale}{$X_{j,2}$}"](t9)->[bend left=30, "\scalebox{\edgeScale}{$X_{i,\alpha_i}$}"](t10)->[bend left=30, "\scalebox{\edgeScale}{$X_{i,\gamma_i}$}"](t11);
(t0) <-[swap, bend right=30, "\scalebox{\edgeScale}{$x_{i,\alpha_i}$}"] (t1) <-[swap, bend right=30, "\scalebox{\edgeScale}{$x_{i,\gamma_i}$}"] (t2) <-[swap, bend right=30, "\scalebox{\edgeScale}{$x_{j,2}$}"] (t3);
(t4) <-[swap, bend right=30, "\scalebox{\edgeScale}{$x_{i,\gamma_i}$}"] (t5) <-[swap, bend right=30, "\scalebox{\edgeScale}{$x_{j,2}$}"] (t6)<-[swap, bend right=30, "\scalebox{\edgeScale}{$x_{i,\alpha_i}$}"](t7);
 (t8)<- [swap, bend right=30, "\scalebox{\edgeScale}{$x_{j,2}$}"](t9)<-[swap, bend right=30, "\scalebox{\edgeScale}{$x_{i,\alpha_i}$}"](t10)<-[swap, bend right=30, "\scalebox{\edgeScale}{$x_{i,\gamma_i}$}"](t11);
};
\end{scope}
\begin{scope}[yshift=-5cm]
\node at (-0.75,0) {c)};
\foreach \i in {0,...,11} {\coordinate (\i) at (\i*1.5cm,0);}
\foreach \i in {3,4} {\fill[black!15, rounded corners] (\i) +(-0.4,-0.25) rectangle +(0.5,0.25);}
\foreach \i in {0,6,9} {\fill[black!15, rounded corners] (\i) +(-0.5,-0.25) rectangle +(2,0.25);}
\foreach \i in {0,...,3} {\node (t\i) at (\i) {\scalebox{\nodeScale}{$t_{j,0,\i}$}};}
\foreach \i in {4,...,7} {\pgfmathparse{int(\i-4}  ,  \node (t\i) at (\i) {\scalebox{\nodeScale}{$t_{j,1,\pgfmathresult}$}};}
\foreach \i in {8,...,11} {\pgfmathparse{int(\i-8)} , \node (t\i) at (\i) {\scalebox{\nodeScale}{$t_{j,2,\pgfmathresult}$}};}
\graph{ 
(t0) ->[bend left=30,"\scalebox{\edgeScale}{$X_{j,0}$}"] (t1) ->[bend left=30,"\scalebox{\edgeScale}{$X_{i,\alpha_i}$}"] (t2) ->[bend left=30,"\scalebox{\edgeScale}{$X_{i,\gamma_i}$}"] (t3);
(t4) ->[bend left=30,"\scalebox{\edgeScale}{$X_{i,\alpha_i}$}"] (t5) ->[bend left=30,"\scalebox{\edgeScale}{$X_{i,\gamma_i}$}"] (t6)->[bend left=30,"\scalebox{\edgeScale}{$X_{j,0}$}"](t7);
 (t8)- >[bend left=30,"\scalebox{\edgeScale}{$X_{i,\gamma_i}$}"](t9)->[bend left=30,"\scalebox{\edgeScale}{$X_{j,0}$}"](t10)->[bend left=30,"\scalebox{\edgeScale}{$X_{i,\alpha_i}$}"](t11);
(t0) <-[swap, bend right=30,"\scalebox{\edgeScale}{$x_{j,0}$}"] (t1) <-[swap, bend right=30,"\scalebox{\edgeScale}{$x_{i,\alpha_i}$}"] (t2) <-[swap, bend right=30,"\scalebox{\edgeScale}{$x_{i,\gamma_i}$}"] (t3);
(t4) <-[swap, bend right=30,"\scalebox{\edgeScale}{$x_{i,\alpha_i}$}"] (t5) <-[swap, bend right=30,"\scalebox{\edgeScale}{$x_{i,\gamma_i}$}"] (t6)<-[swap, bend right=30,"\scalebox{\edgeScale}{$x_{j,0}$}"](t7);
 (t8) <- [swap, bend right=30,"\scalebox{\edgeScale}{$x_{i,\gamma_i}$}"](t9)<-[swap, bend right=30,"\scalebox{\edgeScale}{$x_{j,0}$}"](t10)<-[swap, bend right=30,"\scalebox{\edgeScale}{$x_{i,\alpha_i}$}"](t11);
};
\end{scope}
\begin{scope}[yshift=-7.5cm]
\node at (-0.75,0) {d)};
\foreach \i in {0,...,11} {\coordinate (\i) at (\i*1.5cm,0);}
\foreach \i in {1,4,7,10} {\fill[black!15, rounded corners] (\i) +(-0.5,-0.25) rectangle +(0.5,0.25);}
\foreach \i in {0,...,3} {\node (t\i) at (\i) {\scalebox{\nodeScale}{$t_{j,0,\i}$}};}
\foreach \i in {4,...,7} {\pgfmathparse{int(\i-4}  ,  \node (t\i) at (\i) {\scalebox{\nodeScale}{$t_{j,1,\pgfmathresult}$}};}
\foreach \i in {8,...,11} {\pgfmathparse{int(\i-8)} , \node (t\i) at (\i) {\scalebox{\nodeScale}{$t_{j,2,\pgfmathresult}$}};}
\graph{ 
(t0) ->[bend left=30,"\scalebox{\edgeScale}{$X_{i,\gamma_i}$}"] (t1) ->[bend left=30,"\scalebox{\edgeScale}{$X_{i,\alpha_i}$}"] (t2) ->[bend left=30,"\scalebox{\edgeScale}{$X_{j,2}$}"] (t3);
(t4) ->[bend left=30,"\scalebox{\edgeScale}{$X_{i,\alpha_i}$}"] (t5) ->[bend left=30,"\scalebox{\edgeScale}{$X_{j,2}$}"] (t6)->[bend left=30,"\scalebox{\edgeScale}{$X_{i,\gamma_i}$}"](t7);
 (t8)- >[bend left=30,"\scalebox{\edgeScale}{$X_{j,2}$}"](t9)->[bend left=30,"\scalebox{\edgeScale}{$X_{i,\gamma_i}$}"](t10)->[bend left=30,"\scalebox{\edgeScale}{$X_{i,\alpha_i}$}"](t11);
(t0) <-[swap, bend right=30,"\scalebox{\edgeScale}{$x_{i,\gamma_i}$}"] (t1) <-[swap, bend right=30,"\scalebox{\edgeScale}{$x_{i,\alpha_i}$}"] (t2) <-[swap, bend right=30,"\scalebox{\edgeScale}{$x_{j,2}$}"] (t3);
(t4) <-[swap, bend right=30,"\scalebox{\edgeScale}{$x_{i,\alpha_i}$}"] (t5) <-[swap, bend right=30,"\scalebox{\edgeScale}{$x_{j,2}$}"] (t6)<-[swap, bend right=30,"\scalebox{\edgeScale}{$x_{i,\gamma_i}$}"](t7);
 (t8)<- [swap, bend right=30,"\scalebox{\edgeScale}{$x_{j,2}$}"](t9)<-[swap, bend right=30,"\scalebox{\edgeScale}{$x_{i,\gamma_i}$}"](t10)<-[swap, bend right=30,"\scalebox{\edgeScale}{$x_{i,\alpha_i}$}"](t11);
};
\end{scope}
\begin{scope}[yshift=-10cm]
\node at (-0.75,0) {e)};
\foreach \i in {0,...,11} {\coordinate (\i) at (\i*1.5cm,0);}
\foreach \i in {2,5,7,8} {\fill[black!15, rounded corners] (\i) +(-0.5,-0.25) rectangle +(0.5,0.25);}
\foreach \i in {2,5,11} {\fill[black!15, rounded corners] (\i) +(-2,-0.25) rectangle +(0.5,0.25);}
\foreach \i in {0,...,3} {\node (t\i) at (\i) {\scalebox{\nodeScale}{$t_{j,0,\i}$}};}
\foreach \i in {4,...,7} {\pgfmathparse{int(\i-4}  ,  \node (t\i) at (\i) {\scalebox{\nodeScale}{$t_{j,1,\pgfmathresult}$}};}
\foreach \i in {8,...,11} {\pgfmathparse{int(\i-8)} , \node (t\i) at (\i) {\scalebox{\nodeScale}{$t_{j,2,\pgfmathresult}$}};}
\graph{ 
(t0) ->[bend left=30,"\scalebox{\edgeScale}{$X_{i,\gamma_i}$}"] (t1) ->[bend left=30,"\scalebox{\edgeScale}{$X_{j,1}$}"] (t2) ->[bend left=30,"\scalebox{\edgeScale}{$X_{i,\alpha_i}$}"] (t3);
(t4) ->[bend left=30,"\scalebox{\edgeScale}{$X_{j,1}$}"] (t5) ->[bend left=30,"\scalebox{\edgeScale}{$X_{i,\alpha_i}$}"] (t6)->[bend left=30,"\scalebox{\edgeScale}{$X_{i,\gamma_i}$}"](t7);
 (t8)- >[bend left=30,"\scalebox{\edgeScale}{$X_{i,\alpha_i}$}"](t9)->[bend left=30,"\scalebox{\edgeScale}{$X_{i,\gamma_i}$}"](t10)->[bend left=30,"\scalebox{\edgeScale}{$X_{j,1}$}"](t11);
(t0) <-[swap, bend right=30,"\scalebox{\edgeScale}{$x_{i,\gamma_i}$}"] (t1) <-[swap, bend right=30,"\scalebox{\edgeScale}{$x_{j,1}$}"] (t2) <-[swap, bend right=30,"\scalebox{\edgeScale}{$x_{i,\alpha_i}$}"] (t3);
(t4) <-[swap, bend right=30,"\scalebox{\edgeScale}{$x_{j,1}$}"] (t5) <-[swap, bend right=30,"\scalebox{\edgeScale}{$x_{i,\alpha_i}$}"] (t6)<-[swap, bend right=30,"\scalebox{\edgeScale}{$x_{i,\gamma_i}$}"](t7);
 (t8)<- [swap, bend right=30,"\scalebox{\edgeScale}{$x_{i,\alpha_i}$}"](t9)<-[swap, bend right=30,"\scalebox{\edgeScale}{$x_{i,\gamma_i}$}"](t10)<-[swap, bend right=30,"\scalebox{\edgeScale}{$x_{j,1}$}"](t11);
};
\end{scope}
\begin{scope}[yshift=-12.5cm]
\node at (-0.75,0) {f)};
\foreach \i in {0,...,11} {\coordinate (\i) at (\i*1.5cm,0);}
\foreach \i in {2, 5,8,11} {\fill[black!15, rounded corners] (\i) +(-0.5,-0.25) rectangle +(0.5,0.25);}
\foreach \i in {} {\fill[black!15, rounded corners] (\i) +(-0.5,-0.25) rectangle +(1.5,0.25);}
\foreach \i in {0,...,3} {\node (t\i) at (\i) {\scalebox{\nodeScale}{$t_{j,0,\i}$}};}
\foreach \i in {4,...,7} {\pgfmathparse{int(\i-4}  ,  \node (t\i) at (\i) {\scalebox{\nodeScale}{$t_{j,1,\pgfmathresult}$}};}
\foreach \i in {8,...,11} {\pgfmathparse{int(\i-8)} , \node (t\i) at (\i) {\scalebox{\nodeScale}{$t_{j,2,\pgfmathresult}$}};}
\graph{ 
(t0) ->[bend left=30,"\scalebox{\edgeScale}{$X_{j,0}$}"] (t1) ->[bend left=30,"\scalebox{\edgeScale}{$X_{i,\gamma_i}$}"] (t2) ->[bend left=30,"\scalebox{\edgeScale}{$X_{i,\alpha_i}$}"] (t3);
(t4) ->[bend left=30,"\scalebox{\edgeScale}{$X_{i,\gamma_i}$}"] (t5) ->[bend left=30,"\scalebox{\edgeScale}{$X_{i,\alpha_i}$}"] (t6)->[bend left=30,"\scalebox{\edgeScale}{$X_{j,0}$}"](t7);
(t8)- >[bend left=30,"\scalebox{\edgeScale}{$X_{i,\alpha_i}$}"](t9)->[bend left=30,"\scalebox{\edgeScale}{$X_{j,0}$}"](t10)->[bend left=30,"\scalebox{\edgeScale}{$X_{i,\gamma_i}$}"](t11);
(t0) <-[swap, bend right=30,"\scalebox{\edgeScale}{$x_{j,0}$}"] (t1) <-[swap, bend right=30,"\scalebox{\edgeScale}{$x_{i,\gamma_i}$}"] (t2) <-[swap, bend right=30,"\scalebox{\edgeScale}{$x_{i,\alpha_i}$}"] (t3);
(t4) <-[swap, bend right=30,"\scalebox{\edgeScale}{$x_{i,\gamma_i}$}"] (t5) <-[swap, bend right=30,"\scalebox{\edgeScale}{$x_{i,\alpha_i}$}"] (t6)<-[swap, bend right=30,"\scalebox{\edgeScale}{$x_{j,0}$}"](t7);
(t8)<- [swap, bend right=30,"\scalebox{\edgeScale}{$x_{i,\alpha_i}$}"](t9)<-[swap, bend right=30,"\scalebox{\edgeScale}{$x_{j,0}$}"](t10)<-[swap, bend right=30,"\scalebox{\edgeScale}{$x_{i,\gamma_i}$}"](t11);
};
\end{scope}
\end{tikzpicture}
\caption{
All cases of how two events $X_{i,\alpha_i},X_{i,\gamma_i}$ from $T_i$ can occur as events $X_{j,0},X_{j,1},X_{j,2}$ of $T_j$.
The grey regions show how to inhibit $X_{i,\alpha_i},X_{i,\gamma_i}$.
a) $X_{j,0}=X_{i,\alpha_i}, X_{j,2}=X_{i,\gamma_i}$,  b) $X_{j,0}=X_{i,\alpha_i}, X_{j,1}=X_{i,\gamma_i}$,  c) $X_{j,1}=X_{i,\alpha_i}, X_{j,2}=X_{i,\gamma_i}$,  d) $X_{j,1}=X_{i,\alpha_i}, X_{j,0}=X_{i,\gamma_i}$,  e) $X_{j,2}=X_{i,\alpha_i}, X_{j,0}=X_{i,\gamma_i}$,  f) $X_{j,2}=X_{i,\alpha_i}, X_{j,1}=X_{i,\gamma_i}$.
}
\label{fig:possibilities}
\end{figure}
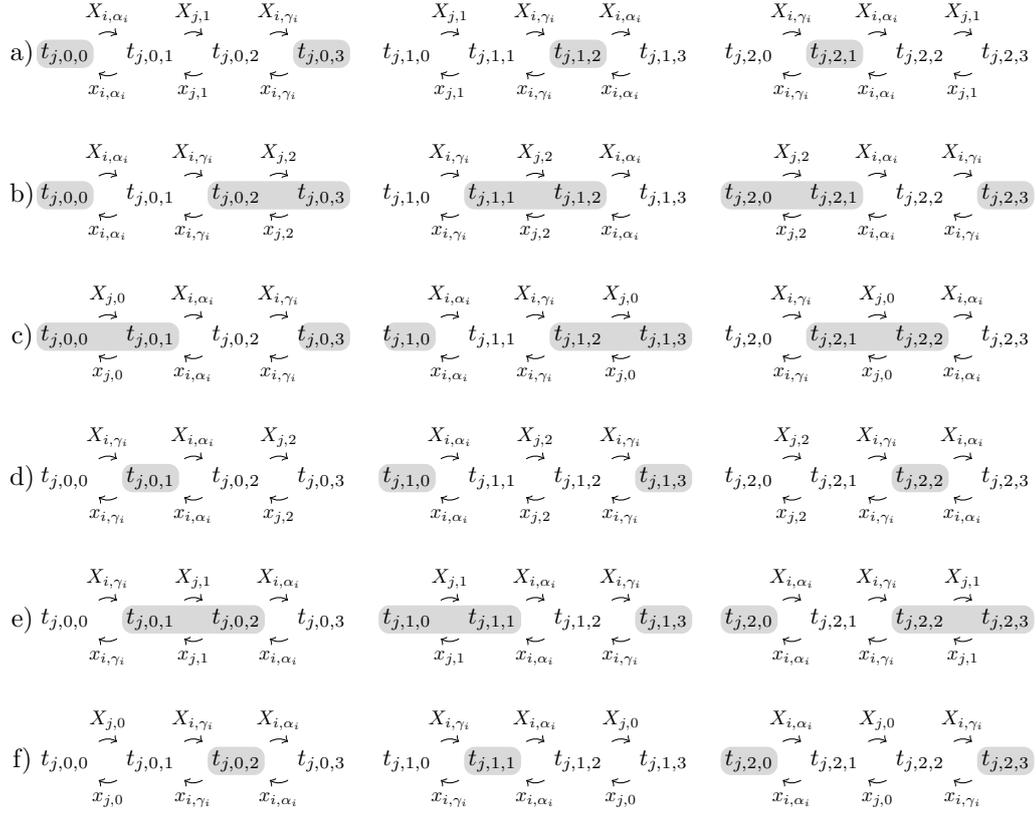
The given construction plan yields a region that inhibits $X_{i,\alpha_i}$ at the remaining states $t_{i,\beta_i,2}, t_{i,\beta_i,3}$ and $t_{i,\gamma_i,2}$ of $T^\sigma_i$.
Altogether, this proves $X_0,\dots, X_{m-1}$ to be inhibitable in $U^\sigma_\varphi$.
Moreover, by the symmetry of the occurrences of $X_j$ and $x_j$ the event $x_j$ is inhibitable at all states of $U^\sigma_\varphi$ in question besides of $g^{\_,x}_{j,0}$, respectively $g^{x,\_}_{j,2}$.
Observe, that in these cases the \swap\ operation is always available. 
Actually, the inhibition of $x_i$ at these states requires some extra effort.
With the definitions above, we define the following sets which will be used to yield a support that inhibits $x_{i,\alpha_i}$ at $g^{\_,x}_{j,0}$, respectively $g^{x,\_}_{j,2}$ where $x_j=x_{i,\alpha_i}$:
\begin{enumerate} 
\item
$S_{7}=\bigcup_{n=0}^{3m-1}\{h_{n,1},h_{n,2},h_{n,3}\}$,
\item
$S_{8}=\{h_{3n+3m,n'}, h_{3n+3m+1,n'}, h_{3n+3m+2,n'}\mid n\in \{i,j,\ell\}, n'\in \{1,4,5\}  \}$,  
\item 
$S_{9}=\bigcup_{n'=3m}^{6m-1}\{h_{n',1},h_{n',2},h_{n',3}\mid n'\not\in \{3n+3m,3n+3m+1,3n+3m+2\}: n\in \{i,j,\ell\}\}$,
\item 
$S_{10}=\{t_{n,0,0},t_{n,1,0},t_{n,2,0}\mid n\in \{0,\dots, m-1\}\}$, 
\item
$S_{11}=(S(T^\sigma_i)\cup  S(T^\sigma_j)\cup  S(T^\sigma_\ell))\setminus (S_0\cup S_1)$,
\item
$S_{12}=\{g^{\_,x}_{n,1},g^{\_,x}_{n,2}, g^{x,\_}_{n,0}, g^{x,\_}_{n,3} \mid n\in \{0,\dots, m-1\}: x_n=x_{i,\alpha_i}\}$,
\item $S_{13}=\{f_{0,0},f_{0,4},f_{2,2},f_{2,3}\}$
\item
$S_{14}=\bigcup_{n=3m}^{6m-1}\{h_{n,1},h_{n,2},h_{n,3}\}$,
\item
$S_{15}=\{h_{3n,n'}, h_{3n+1,n'}, h_{3n+2,n'}\mid n\in \{i,j,\ell\}, n'\in \{1,2,4\}  \}$,  
\item 
$S_{16}=\bigcup_{n'=0}^{3m-1}\{h_{n',1},h_{n',2},h_{n',3}\mid n'\not\in \{3n,3n+1,3n+2\}: n\in \{i,j,\ell\}\}$.
\item
$S_{17}=\{g^{\_,q}_{3n,0},g^{\_,q}_{3n,3}, g^{\_,q}_{3n+1,0},g^{\_,q}_{3n+2,3}, g^{\_,q}_{3n+2,0},g^{\_,q}_{3n+3,3} \mid n\in \{i,j,\ell\} \}$,
\end{enumerate}
Now, to enrich the set $S_7\cup \dots \cup S_{13}$ to a fitting support of $U^\sigma_\varphi$ if $\sigma = \sigma_4$, respectively the set $S_{10}\cup \dots \cup S_{17}$ if $\sigma=\sigma_3$, we have to take the remaining generators into account.
That is, for all the other generators $G^{\eta,\varrho}_n$, different from the ones affected by definition of $S_{12}$, respectively by $S_{12},S_{17}$, we put the sinks $g^{\eta,\varrho}_{n,0}, g^{\eta,\varrho}_{n,1}$ of $k$ into $sup$.
This proves the lemma.
\end{proof}
To show that the inhibition of the key event at the key state in $U^\sigma_\varphi$ implies the ESSP, it remains to show that the unique events are inhibitable, too.
Each unique event occurs exactly once in $U^\sigma_\varphi$ and, therefore, only in one single TS.
Regarding this fact and the symmetry of the gadget TSs, mainly the generators, all these events are inhibitable by an input region, too.
Hence, we state the next lemma without proof.
\begin{lemma}[Without proof] The unique events are inhibitable. \end{lemma}
Finally, the next lemma states that if $k$ is inhibitable at $h_{0,6}$ in $U^\sigma_\varphi$ then $U^\sigma_\varphi$ has the SSP.
%ssp in \sigma_1,...,\sigma_4
\begin{lemma}\label{lem:ssp_for_sigma_1_to_sigma_4} If $k$ is inhibitable at $h_{0,6}$ in $U^\sigma_\varphi$ then $U^\sigma_\varphi$ has the SSP. \end{lemma}
\begin{proof}
For simplification, in the following by a key region we mean a region of $U^\sigma_\varphi$ that inhibits $k$ at $h_{0,6}$.
Firstly, if $s\edge{e}$ and $\neg s'\edge{e}$ in $U^\sigma_\varphi$ such that $e$ is inhibitable at $s'$ by an input region, that is, there is a region of $(sup,sig)$ of $U^\sigma_\varphi$ with $sig(e)=\inp$ and $sup(s')=0$, then $s$ and $s'$ are separated.
Secondly, we note, that each inhibiting region presented in the former lemmata are input regions of $U^\sigma_\varphi$.
Thirdly, for $i\in \{0,\dots, m-1\}$ and $\ell\in \{0,\dots, 2\}$ it is true, that each event of $T_{i,\ell}$ is unique in $T_{i,\ell}$.
Hence, the states of $T_{i,\ell}$ are separable.
If $i\in \{0,\dots, 6m-1\}$ then the set of states $\{h_{i,0},\dots, h_{i,6}\}$ can be extended to a region of $U^\sigma_\varphi$.
Hence, with respect to $H$, it suffices to investigate the separability of $\{h_{i,0},\dots, h_{i,6}\}$.
By the uniqueness of $v_j,q_j,p_j,y_j$ in $H$ for $j\in \{0,\dots, 3m-1\}$, the states $h_{i,2},h_{i,4}, h_{i,6}$ are separable.
The inhibition of $k$ ($z_j$, respectively $w_j$) separates $\{h_{i,0}, h_{i,3}\}$ ($\{h_{i,1}, h_{i,5}\}$) from all other states.
The set $\bigcup_{i=0}^{6m-1}\{h_{i,5},h_{i,6}\}$ is extendable to a region of $U^\sigma_\varphi$ that separates $h_{i,1}$ from $h_{i,5}$.
Furthermore, the second row of Lemma~\ref{lem:vice_and_wire} proves that $h_{i,0}$ and $h_{i,3}$ are separable, too.
By the uniqueness of $q_0,k$ ($n,z_0$) in $F_1$ ($F_0$), the states of $F_1$ ($F_0\setminus\{f_{0,0},f_{0,3}\}$) are separable.
For $F_0$ it suffices to show that $f_{0,0}$ and $f_{0,3}$ are separable. 
Of course $\{f_{0,0},f_{0,1}\}$, possibly together with $\{g^{n,\_}_{0,0}, g^{ n, \_}_{0,2}\}$ or $\{f_{2,0}, f_{2,2}\}$, respectively, is extendable to a separating region.
This region separates $\{g^{n,\_}_{0,0}, g^{ n, \_}_{0,2}\}$ from $\{g^{n,\_}_{0,1}, g^{ n, \_}_{0,3}\}$ and  $\{f_{2,0}, f_{2,2}\}$ from  $\{f_{2,1}, f_{2,3}\}$, respectively.
An input key region solves the remaining state separation atoms for $G^{n,\_}_0$  and $F_2$.
For each applied generator $G^{\eta,\varrho}_j$, an input key region separates $\{g^{\eta,\varrho}_{j,0}, g^{\eta,\varrho}_{j,1}\}$ from $\{g^{\eta,\varrho}_{j,2}, g^{\eta,\varrho}_{j,3}\}$.
Finally, we note that for each applied generator, there is at least one region presented such that $\{g^{\eta,\varrho}_{j,0}, g^{\eta,\varrho}_{j,2}\}$ are included (excluded) are excluded (included) $\{g^{\eta,\varrho}_{j,1}, g^{\eta,\varrho}_{j,3}\}$.
Consequently, this proves the lemma.
\end{proof}

%%%%%%%%%%%%%%%%%%%%%%%%%%%%%%%%%%%%%%%%%%%%%%%%%%%%%%%%%%%%%%%%%%%%%%%%%%%%%%%%%%%%%%%%%%%%%%%%%%%%%%%%%%%%%%%%%%%%%%%%%%%%%

\subsection{Concluding the ESSP and the SSP for $\sigma_5,\sigma_6$}
Let $\tau\in\sigma\in \{\sigma_5,\sigma_6\}$ and $\test^{\sigma_5}=\free$ and $\test^{\sigma_6}=\used$.
In this section we present explicitly supports showing that the inhibition of the key event at the key state implies the $\tau$-(E)SSP for  $U^\sigma_\varphi$.
The following observation helps us to reduce the number of ESSP atoms that have to be solved explicitly:
If $A$ is a gadget TS installed by $U^\sigma_\varphi$ such that for an event $e\in E(U^\sigma_\varphi )$ holds that $e\not\in E(A)$ then $e$ can be inhibited at all states of $S(A)$ by a $\tau$-region $(S(U^\sigma_\varphi)\setminus S(A), sig)$ where for $e'\in E(U^\sigma_\varphi)$ holds that $sig(e')=\used$ if $e'=e$ and, otherwise, $sig(e')=\nop$.
Hence, for any event $e\in E(U^\sigma_\varphi)$ and any TS $A$ installed by $U^\sigma_\varphi$ we only show $e$ to be inhibitable at a state $s$ of $A$ in question, that is, where $\neg s\edge{e}$, if $e$ actually occurs in $A$, that is, $e\in E(A)$.
To prove for a every event $e$ in a given subset of $E(U^\sigma_\varphi)$ that it is inhibitable at all states $s$ in a given subset of $S(U^\sigma_\varphi)$, we will simply define a support $sup^\sigma \subseteq S(U^\sigma_\varphi)$ that allows an inhibiting region.
Note, that for the sake of simplicity we sometimes present a set $S \subseteq S(U^{\sigma_5}_\varphi)\cup S(U^{\sigma_6}_\varphi)$ such that $sup^{\sigma_5}=S\cap S(U^{\sigma_5}_\varphi)$ and $sup^{\sigma_6}=S\cap S(U^{\sigma_6}_\varphi)$.
Such a presented support $sup^\sigma$ allows always a signature $sig^\sigma$ that for $e\in E(U^\sigma_\varphi)$ is defined by the following rules: 
\[sig^\sigma(e)=
\begin{cases}
\test^\sigma, & \text{if } e\in E \\
\set, & \text{if }  s\edge{e}s'\fbedge{e}s''\in U^\sigma_\varphi \text{ such that } sup(s)\not=sup(s')=sup(s'')=1 \ (*)\\
\swap, & \text{if } e \in E(U^\sigma_\varphi )\setminus E \text{ and } sup(s)\not=sup(s') \text{ and not } * \\
\nop, & \text{if }  e \in E(U^\sigma_\varphi )\setminus E \text{ and } sup(s)=sup(s') \text{ and not } *
\end{cases}
\]
Note, that this approach to define a signature implies that the signature only depends on the given support $sup^\sigma$ and the switch $\sigma$.
Therefore, for the sake of simplicity, in the sequel we often refer to a given support $sup^\sigma$ as to the region $(sup^\sigma,sig^\sigma)$ which it allows and, e.g., say $sup^\sigma$ inhibits $e$ at $s$ instead of $(sup^\sigma,sig^\sigma)$ inhibits $e$ at $s$.

In the following proofs, for $n\in \{0,\dots, m-1\}$ if we have given one of $\alpha_n, \beta_n$ or  $\gamma_n$ as an element of  $\{0,1,2\}$ then the value of the others is assumed to be determined by the following definitions:
\begin{enumerate}
\item given $\alpha_n$: $\beta_n=(\alpha_n+1)\text{mod}3$ and $\gamma_n=(\alpha_n+2)\text{mod}3$,
\item given $\beta_n$: $\gamma_n=(\beta_n+1)\text{mod}3$ and $\alpha_n=(\beta_n+2)\text{mod}3$,
\item given $\gamma_n$: $\alpha_n=(\gamma_n+1)\text{mod}3$ and $\beta_n=(\gamma_n+2)\text{mod}3$.
\end{enumerate}

\begin{lemma}\label{lem:event_k} The key event is inhibitable. \end{lemma}
\begin{proof}

In the following for $\sigma\in \{\sigma_5,\sigma_6\}$ let $i\in \{0,\dots, m-1\}$ and $\alpha_i\in \{0,1,2\}$ arbitrary but fixed and $j,\ell\in \{0,\dots, m-1\}\setminus \{i\}$ such that $X_{i,\alpha_i}\in E(T^\sigma_j)\cap E(T^\sigma_\ell)$.
For $n\in \{j,\ell\}$ let $\alpha_n=0$ ($\alpha_n=1$,$\alpha_n=2$) if $X_{n,0}=X_{i,\alpha_i}$ ($X_{n,1}=X_{i,\alpha_i},X_{n,2}=X_{i,\alpha_i}$).

If $\sigma=\sigma_5$ ($\sigma=\sigma_6$) then a key region inhibits $k$ at all states of the key union besides $f'_{2,3},f'_{2,4}$ and $d_{0,5},\dots,d_{18m-1,5}$ and $g_{0,5},\dots, g_{3m-1,5}$.
Furthermore, with respect to $T^\sigma_\varphi$ a key union inhibits $k$ at all states of $F^\sigma_T$ and at $t_{i,\alpha_i,2}$ ($t_{i,\beta_i,11}$).

Hence, by the arbitrariness of $i$ and $\alpha_i$ and the symmetry of the translators, to prove $k$ to be inhibitable in $U^\sigma_\varphi$ it is sufficient to present regions that show the inhibition of $k$ at $f'_{2,3},f'_{2,4}$ and at $t_{i,\alpha_i,3},\dots, t_{i,\alpha_i,11}$ ($t_{i,\beta_i,2},\dots, t_{i,\beta_i,10}$) and at $g_{3i+\alpha_i,5}$ ($g_{3i+\beta_i,5}$) and at $d_{0,5},\dots,d_{18m-1,5}$.
To attack these challenge we define the following sets of states that will help us to compose corresponding supports of $U^\sigma_\varphi$:
\begin{enumerate}
%translator 1
\item
$S^{\sigma_5}_0=\emptyset$, $S^{\sigma_6}_0=\{t'_{n,0,0},t'_{n,0,1}, t'_{n,1,0},t'_{n,1,1}, t'_{n,2,0},t'_{n,2,1}\mid n\in \{i,j,\ell\}\}$
%translator T_{n,alpha_n}
\item
$S^{\sigma_5}_1=S^{\sigma_6}_1=\{t'_{n,\alpha_n,4},\dots, t'_{n,\alpha_n,11}, t'_{n,\alpha_n,12}, t'_{n,\alpha_n,13}, t'_{n,\alpha_n,15},\dots, t'_{n,\alpha_n,20} \mid n\in \{i,j,\ell\}\}$
%translator T_{n,beta_n}
\item
$S^{\sigma_5}_2=S^{\sigma_6}_2=\{ t'_{n,\beta_n,10}, t'_{n,\beta_n,11}, t'_{n,\beta_n,18}, t'_{n,\beta_n,19} \mid n\in \{i,j,\ell\}\}$
%translator T_{n,gamma_n}
\item
$S^{\sigma_5}_3=S^{\sigma_6}_3=\{t'_{n,\gamma_n,7},\dots, t'_{n,\gamma_n,11}, t'_{n,\gamma_n,15}, t'_{n,\gamma_n,16}, t'_{n,\gamma_n,18},\dots, t'_{n,\gamma_n,20} \mid n\in \{i,j,\ell\}\}$
%translator freezer F_T
\item 
$S^{\sigma_5}_4=\{b'_{n,2},b'_{n,3}\mid n\in \{0,\dots, m-1\}\}$ and $S^{\sigma_6}_4=S(F^{\sigma_6}_T)$,
%head
\item
$S^{\sigma_5}_5=\emptyset$ and $S^{\sigma_6}_5=S(H^{\sigma_6})$,
%freezer F_K
\item
$S^{\sigma_5}_6=\{f'_{1,2}, f'_{2,2}, f'_{2,3}, f'_{2,4}, f'_{2,5}, f'_{2,6}\}$ and $S^{\sigma_6}_6=S(F^{\sigma_6}_K)$,
%duplicator
\item With the index set $I_a=\{18n+6\alpha_n+3, 18n+6\beta_n+5, 18n+6\gamma_n+4\mid n\in \{i,j,\ell\}\}$ corresponding to the affected a-events we have 
\begin{enumerate}
%affected duplicators sigma_5
\item
$S^{\sigma_5}_7=\{d_{n,3}, d_{n,4}, d_{n,5}\mid n\in I_a\}$,
%not-affected duplicators sigma_5
\item
$S^{\sigma_5}_8=\{d_{n,2}, d_{n,3}, d_{n,4}\mid n\in \{0,\dots, 18m-1\}\setminus I_a\}$,
%affected duplicators sigma_6
\item
$S^{\sigma_6}_7=\{d_{n,0}, d_{n,1}, d_{n,3}, d_{n,4},d_{n,6}, d_{n,7} \mid n\in I_a  \}$,
%not-affected  duplicators sigma_6
\item
$S^{\sigma_6}_8=\{d_{n,0}, \dots, d_{n,8}\mid n\in \{0,\dots, 18m-1\}\setminus I_a\}$,
\end{enumerate}
%generators
\item With the index set $I_w=\{3n,3n+1,3n+2\mid n\in \{i,j,\ell\}\}$ corresponding to the affected w-events we have 
\begin{enumerate}
%affected generators sigma_5
\item
$S^{\sigma_5}_9=\{g_{n,3}, g_{n,4}, g_{n,5}\mid n\in I_w\}$
%not-affected generators sigma_5
\item
$S^{\sigma_5}_{10}=\{g_{n,2}, g_{n,3}, g_{n,4}\mid n\in \{0,\dots, 3m-1\}\setminus I_w\}$,
%affected generators sigma_6
\item
$S^{\sigma_6}_9=\{g_{n,0}, g_{n,1}, g_{n,3}, g_{n,4},g_{n,6}, g_{n,7} \mid n\in I_w  \}$, 
%/not-affected  generators sigma_6
\item
$S^{\sigma_6}_{10}=\{g_{n,0}, \dots, g_{n,8}\mid n\in \{0,\dots, 3m-1\}\setminus I_w\}$,
\end{enumerate}
%remaining in translators and all duplicators sigma_5
%translators
\item
$S^{\sigma_5}_{11} = S(U(T^{\sigma_5}_0,\dots,T^{\sigma_5}_{m-1} ))\setminus \{t'_{n,n',0}, t'_{n,n',1},t'_{n,n',2}, t'_{n,n',5}, t'_{n,n',8}, t'_{n,n',11}\mid n\in \{0,\dots,m-1\},n' \in \{0,1,2\} \} $,
%duplicators
\item
$S^{\sigma_5}_{12}=\{d_{n,3}, d_{n,4}, d_{n,5}\mid n\in \{0,\dots, 18m-1\}\}$, $S^{\sigma_5}_{13}=\{g_{n,2}, g_{n,3}, g_{n,4}\mid n\in \{0,\dots, 3m-1\}\}$,
%remaining in translators and all duplicators sigma_6
%translators
\item
$S^{\sigma_5}_{11}=\{t'_{n,n',0}, t'_{n,n',1},t'_{n,n',2}, t'_{n,n',5}, t'_{n,n',8}, t'_{n,n',11}\mid  n\in \{0,\dots,m-1\},n' \in \{0,1,2\} \}$,
%duplicators
\item
$S^{\sigma_6}_{12}=\{d_{n,0}, d_{n,1}, d_{n,3}, d_{n,4},d_{n,6}, d_{n,7},\mid n\in \{0,\dots, 18m-1\}\}$ and $S^{\sigma_6}_{13}=S(G^{\sigma_6})$.
\end{enumerate}
Using the just defined sets, the following table shows the inhibition of $k$ at the states in question:
\begin{longtable}{  p{0.3cm}   p{5.25cm}  p{7.25cm}}
R & Support  & Target States \\ \hline
%row 1
$R^k_2$
&
%support
\raggedright{$S^\sigma_0, S^\sigma_1, S^\sigma_2,S^\sigma_3, S^\sigma_4, S^\sigma_5,S^\sigma_7,S^\sigma_8, S^\sigma_9,S^\sigma_{10}$}
&
%target states
\raggedright{if $\sigma=\sigma_5: \{t'_{i,\alpha_i,5}, t'_{i,\alpha_i,8}, t'_{i,\alpha_i,11}, g_{3i+\alpha_i,5}, f'_{2,3}, f'_{2,4}\}$, \newline if $\sigma=\sigma_6: \{t'_{i,\beta_i,2}, t'_{i,\beta_i,5}, t'_{i,\beta_i,8}, g_{3i+\alpha_i,5}\}$}\arraybackslash\\ \hline
%row 2
$R^k_3$
&
%support
\raggedright{$S^\sigma_4,S^\sigma_5,S^\sigma_6,S^\sigma_{11}, S^\sigma_{12}, S^\sigma_{13}$}
& 
%target states
rem. of $T^{\sigma_5}_{i,\alpha_i}$/ $T^{\sigma_6}_{i,\beta_i}$, $\{d_{n,5}\mid n\in \{0,\dots, 18m-1\}\}$\\  \hline
%row 3
$R^k_4$
&
%support
\raggedright{$S(U^{\sigma_6}_\varphi)\setminus \{f'_{1,2},f'_{2,2}, \dots, f'_{2,6}\}$}
&
for $\sigma=\sigma_6:$ $\{f'_{2,3},f'_{2,3}\}$\\
\end{longtable}
\end{proof}

\begin{lemma}\label{lem:event_v} The events $v_0,\dots, v_{3m-1}$ and $w_0,\dots, w_{3m-1}$ are inhibitable.
\end{lemma}
\begin{proof}

If $\sigma=\sigma_5$ ($\sigma=\sigma_6\}$), $i\in \{0,\dots, m-1\}$ and $\alpha_i\in \{0,1,2\}$, then the lemma is justified for $v_0,\dots, v_{3m-1}$  if we show the inhibition of $v_{3i+\alpha_i}$ ($v_{3i+\beta_i}$) in $H'_{3i+\alpha_i}$ ($H'_{3i+\beta_i}$) and $T^{\sigma}_{i,\alpha_i}$ ($T^{\sigma}_{i,\beta_i}$).
For $\sigma=\sigma_5$ ($\sigma=\sigma_6$)  the event $v_{3i+\alpha_i}$ ($v_{3i+\beta_i}$) is, besides of $t'_{i,\alpha_i,0}$ ($t'_{i,\beta_i,0}$), already inhibited in $T^{\sigma_5}_{i,\alpha_i}$ ($T^{\sigma_6}_{i,\beta_i}$) by using the supports $R^k_2,R^k_3$ of Lemma~\ref{lem:event_k}.
Hence, for $v_{3i+\alpha_i}$ ($v_{3i+\beta_i}$) it remains to show the inhibition at $t'_{i,\alpha_i,0}$ ($t'_{i,\beta_i,0}$) and the states in question from $S(H'_{3i+\alpha_i})$ ($S(H'_{3i+\beta_i})$).
To do so, we define the following sets to be used for composing fitting supports:
\begin{enumerate}
%translator + freezer F_T
\item
$S_0=\{t'_{n,0,0},t'_{n,1,0},t'_{n,2,0},b'_{n,0},b'_{n,5}\mid n\in \{0,\dots,m-1\}\}$
%head and generator
\item
$S_1=\{h'_{n,0},h'_{n,4},h'_{n,5},g_{n,0},g_{n,2},g_{n,3}, g_{n,4},g_{n,7}, g_{n,8}\mid n\in \{0,\dots, 3m-1\}$
%freezer F_K
\item 
$S_2=\{f'_{0,0},f'_{0,4}, f'_{0,5}, f'_{0,6}, f'_{1,0}, f'_{1,4}, f'_{2,0},  f'_{2,8}\}$
%duplicator 
\item 
$S_3=\{d_{n,0},d_{n,7}\mid n\in \{0,\dots, 18,-1\}\}$
%support_1 sigma_5/sigma_6
\item
$S_4=S_0\cup S_1\cup S_2\cup S_3$ and $S_5=(S(U^{\sigma_6}_\varphi)\setminus S_4)\cup \{b_{n,1},b_{n,2}\mid n\in \{0,\dots,m-1\}\}$
\item
$S_6=\{h'_{n,0}, h'_{n,1}\mid n\in \{0,\dots, m-1\} \}\cup\{f'_{0,2}, f'_{0,3}, f'_{0,4}\}$ and $S_7=S(U^{\sigma_6}_\varphi)\setminus S_6$
\end{enumerate}
For $\sigma_5$ ($\sigma_6$) the set $S_4$ ($S_5$) is a support that allows a signature such that $v_{3i+\alpha_i}$ ($v_{3i+\beta_i}$) is inhibited at $t'_{i,\alpha_i,0}$ ($t'_{i,\beta_i,0}$) and $h'_{3i+\alpha_i,0}, h'_{3i+\alpha_i,4}, h'_{3i+\alpha_i,5}$ ($h'_{3i+\beta_i,0}, h'_{3i+\beta_i,4}, h'_{3i+\beta_i,5}$).
Finally, for $\sigma_5$ ($\sigma_6$) the support $S_6$ ($S_7$) can be used for the inhibition of $v_{3i+\alpha_i}$ ($v_{3i+\beta_i}$) at the last state standing $h'_{3i+\alpha_i,1}$ ($h'_{3i+\beta_i,1}$).

We now argue that $w_{3i+\alpha_i}$ is for $\sigma_5$ and $\sigma_6$ inhibitable in $T_{i,\alpha_i}$ and $G_{3i+\alpha_i}$.
Firstly, we observe that for $\sigma_5$ ($\sigma_6$) the key region inhibits $w_{3i+\alpha_i}$ at $t'_{i,\alpha_i,2}$ ($t'_{i,\alpha_i,11}$).
Secondly, the region $R^k_3$ of Lemma~\ref{lem:event_k} inhibits $w_{3i+\alpha_i}$ at all remaining states of $T^{\sigma_5}_{i,\alpha_i}$ ($T^{\sigma_6}_{i,\alpha_i}$) besides of $t'_{i,\alpha_i,0}, t'_{i,\alpha_i,5}, t'_{i,\alpha_i,8}$.
The inhibition at $t'_{i,\alpha_i,0}$ can be done by the region which is build on the support $S_4$ ($S_5$) defined above.
Thirdly, for the inhibition at $t'_{i,\alpha_i,5}, t'_{i,\alpha_i, 8}$ we define the following sets:
\begin{enumerate}
%sigma_5 translator
\item
$S_8=S(T^{\sigma_5}_{i,\alpha_i})\setminus \{t'_{i,\alpha_i,12},\dots,  t'_{i,\alpha_i,20}\}$,
%sigma_5 translator
\item
$S_9= \{d_{n,6}, d_{n,7}, d_{n,8}\mid n\in \{18i+6\alpha_i+3, 18i+6\alpha_i+4, 18i+6\alpha_i+5\}\}$ 
\end{enumerate}
Then for $\sigma_5$, respectively for $\sigma_6$, the support $(S(T^{\sigma_6}_{i,\alpha_i})\setminus S_8)\cup S_9$, respectively $S_8\cup S_9\cup S(G_{3i+\alpha_i})$, allows a signature such that $w_{3i+\alpha_i}$ is inhibited at $t'_{i,\alpha_i,5},t'_{i,\alpha_i,8}$.
That is, the inhibition of $w_{3i+\alpha_i}$ in $T_{i,\alpha_i}$ is completed.
We now argue for the states in question of $G_{3i+\alpha_i}$.
The inhibition at $g_{3i+\alpha_i,0}, g_{3i+\alpha_i,2},\dots, g_{3i+\alpha_i,4}, g_{3i+\alpha_i,7}, g_{3i+\alpha_i,8}$ is done for $\sigma_5$ ($\sigma_6$ ) with $S_4$ ($S_5$).
Finally, we can complete the set $\{g_{3i+\alpha_i,0},g_{3i+\alpha_i,1}, g_{3i+\alpha_i,2},g_{3i+\alpha_i,3},\}$, respectively the set $\{g_{3i+\alpha_i,0},g_{3i+\alpha_i,1}, g_{3i+\alpha_i,5},\dots, g_{3i+\alpha_i,8}\}$, to fitting region of $U^{\sigma_5}_\varphi$, respectively $U^{\sigma_6}_\varphi$, to inhibit $w_{3i+\alpha_i}$ at the last state standing: $g_{3i+\alpha_i,1}$.
\end{proof}

\begin{lemma}\label{lem:event_z} The event $z$ is inhibitable \end{lemma}
\begin{proof}
For $\sigma_5$ ($\sigma_6$) the support $S_0\cup S_1\cup S_2$ ($S(U^{\sigma_6}_\varphi)\setminus (S_0\cup S_1\cup S_2)$),where 
\begin{enumerate}
%freezer
\item
$S_0=\{f'_{0,0}, f'_{0,1}, f'_{0,2}, f'_{1,0}, f'_{1,1}, f'_{1,2}, f'_{2,0}, f'_{2,1}, f'_{2,2}, f'_{2,7}, f'_{2,8}, f'_{2,9}\}$,
%duplicator
\item
$S_1=\{d_{n,0},d_{n,1},d_{n,5},d_{n,6},d_{n,7},d_{n,8}\mid n\in \{0,\dots, 18m-1\}\}$,
%generator
\item
$S_2=\{g_{n,0},g_{n,1},g_{n,5},g_{n,6},g_{n,7},g_{n,8}\mid n\in \{0,\dots, 3m-1\}\}$.
\end{enumerate}
allows an inhibiting  region of $U^{\sigma_5}_\varphi$ ($U^{\sigma_6}_\varphi$),
\end{proof}

\begin{lemma}\label{lem:event_p_and_y} The events $p_0,\dots, p_{18m-1}$ and $y_0,\dots,y_{3m-1}$ are inhibitable.
\end{lemma}
\begin{proof}

Let $i\in \{0,\dots, m-1\}$ and  $\alpha_i \in \{0,\dots,2\}$ be arbitrary but fixed.
We need the following sets:
\begin{enumerate}
%y.1+p.1
\item
$S^{\sigma_5}_0=\{f'_{2,4}, f'_{2,5},g_{n,3}, d_{n',3}\mid n\in \{0,\dots, 3m-1\},n'\in \{0,\dots, 18m-1\}\}$, 
%y.2
\item $S^{\sigma_5}_1=M_0\cup M_1\cup M_2$, where
\begin{enumerate}
\item
$M_0=\{t'_{n,0,0},t'_{n,1,0},t'_{n,2,0},b_{n,0}, b_{n,3},b'_{n,0}, b'_{n,5},h'_{n,0},h'_{n,5}\mid n\in \{0,\dots, m-1\}\}$,
\item
$M_1=\{f'_{0,0}, f'_{0,4}, f'_{0,5}, f'_{0,6}, f'_{1,0}, f'_{1,4}, f'_{2,0}, f'_{2,8}\}$,
\item
$M_2=\{d_{n,0},d_{n,7}, g_{n',0},g_{n',7}\mid n\in \{0,\dots, 18m-1\}, n'\in \{0,\dots, 3m-1\}\}$,
\end{enumerate}

%y.3
\item 
$S^{\sigma_5}_2=(S(T^{\sigma_5}_{i,\alpha_i})\setminus \{t_{i,\alpha_i,0}, t_{i,\alpha_i,1}\})\cup \{h'_{3i+\alpha_i,3}, h'_{3i+\alpha_i,4},g_{3i+\alpha_i,6}, g_{3i+\alpha_i,7}, g_{3i+\alpha_i,8}\}$,
%p.4
\item $S^{\sigma_5}_3=N_0\cup N_1$, where 
\begin{enumerate}
%translator
\item
$N_0=S(T_{i,\alpha_i})\setminus \{t'_{i,\alpha_i,0}, t'_{i,\alpha_i,1}, t'_{i,\alpha_i,2}, t'_{i,\alpha_i,5}, t'_{i,\alpha_i,8}, t'_{i,\alpha_i,11}\}$,
%generators
\item
$N_1=\{d_{n,6},d_{n,7},d_{n,8}\mid n\in \{18i+6\alpha_i,\dots, 18i+6\alpha_i+5\}\}$.
\end{enumerate}
\end{enumerate}
In the sequel for $n\in \{0,\dots,3\}$ let $S^{\sigma_6}_n=S(U^{\sigma_6}_\varphi)\setminus S^{\sigma_5}_n$.
Let $n\in \{0,\dots, 3m-1\}$ and $n'\in \{0,\dots, 18m-1\}$ and $\sigma\in \{\sigma_5,\sigma_6\}$. 
The inhibition of $y_n$  and $p_n$ at $g_{n,3}$ and $d_{n',3}$ is allowed by the support $S^{\sigma}_0$. 
The support $S^{\sigma}_1$ allows the inhibition of $y_n$  at $g_{n,0}$  and $p_n$ and $d_{n',0}$.
Moreover, $y_{3i+\alpha_i}$ can be inhibited at $g_{3i+\alpha_i,6}, \dots, g_{3i+\alpha_i,8}$ by $S^{\sigma}_2$.
Finally, $S^{\sigma}_3$ allows for $n\in \{18i+6\alpha_i,\dots, 18i+6\alpha_i+5\}$ the inhibition of $p_n$ at $d_{n,6},d_{n,7},d_{n,8}$.
By the arbitrariness of $i$ and $\alpha_i$ this proves the lemma.

\end{proof}

\begin{lemma}\label{lem:event_X} The variable events $X_0,\dots,X_{m-1}$ and $x_0,\dots,x_{m-1}$ are inhibitable. \end{lemma}
\begin{proof}
Let $\sigma\in \{\sigma_5,\sigma_6\}$, $i\in \{0,\dots, m-1\}$, $\alpha_i\in \{0,\dots, 2\}$ and $j,\ell\in \{0,\dots, m-1\}\setminus \{i\}$ such that $X_{i,\alpha_i}\in E(T^\sigma_j)\cap E(T^\sigma_\ell)$. 
For $n\in \{j,\ell\}$ let $\alpha_n=0$ ($\alpha_n=1$, $\alpha_n=2$) if $X_{n,0}=X_{i,\alpha_i}$ ($X_{n,1}=X_{i,\alpha_i}$, $X_{n,2}=X_{i,\alpha_i}$).
Using the following sets:
\begin{enumerate}
%translators
\item
$S_0=\{t'_{n,\alpha_n,3}, t'_{n,\alpha_n,4}, t'_{n,\beta_n,9}, t'_{n,\beta_n,10}, t'_{n,\gamma_n,6}, t'_{n,\gamma_n,7}\mid n\in \{i,j,\ell\}\}$
%duplicators
\item
$S_1=\{d_{18n+6\alpha_n+3,n'}, d_{18n+6\beta_n+2,n'} , d_{18n+6\gamma_n+1,n'}, \mid n\in \{i,j,\ell\}, n'\in \{6,7\}\}$
\end{enumerate}
we have that $S^{\sigma_5}_0=S_0\cup S_1$ ($S(U^{\sigma_6}_\varphi)\setminus S^{\sigma_5}_0$) is a support that allows the inhibition of $X_{i,\alpha_i}$ in $T^{\sigma_5}_\varphi$ ($T^{\sigma_6}_\varphi$) and, therefore, in $U^{\sigma_5}_\varphi$ ($U^{\sigma_6}_\varphi$).
Similarly, we obtain the inhibition of $x_{i,\alpha_i}$ in $T^{\sigma}_\varphi$.
Finally, the set $\{b'_{n,0}, b'_{n,0}, b'_{n,4}, b'_{n,5}, b'_{n,6}\}$ ($\{b_{n,1},b_{n,2}\}$) can be enhanced to a support of $U^{\sigma_5}_\varphi$ ($U^{\sigma_6}_\varphi$) allowing the inhibition of $x_n=x_{i,\alpha_i}$ in $B'_n$ ($B_n$).

\end{proof}

\begin{lemma}\label{lem:event_m} The event $m$ is inhibitable. \end{lemma}
\begin{proof}
For $\sigma\in \{\sigma_5,\sigma_6\}$ the event $m$ occurs only in $U(F^\sigma_K,H^\sigma)$.
Using the sets
\begin{enumerate}
%F^\sigma_K all but h_{j,3} and f_{0,7}
\item
$S_0=\{f'_{0,0}, f'_{0,3}, f'_{0,6},, f'_{0,8},f'_{1,0},f'_{1,4}, f'_{2,1}, f'_{2,2}, f'_{2,5},f'_{2,6}, f'_{2,7}\}$
%F^\sigma_K dedicated to  f_{0,7}
\item
$S_1=\{f'_{0,0}, f'_{0,3}, f'_{0,7},f'_{1,0},f'_{1,4}, f'_{2,1}, f'_{2,2}, f'_{2,8}, f'_{2,9}\}$
%head
\item
$S_2=\{h'_{n,0}, h'_{n,4}, h'_{n,5}\mid n\in \{0,\dots, 3m-1\}\}$
\end{enumerate}
we have $S_0\cup S_2$, respectively $S_1\cup S_2$, as supports of $U(F^{\sigma_5}_K,H^{\sigma_5})$ that, altogether, allow the inhibition of $m$ at all states of $U(F^{\sigma_5}_K,H^{\sigma_5})$ in question, besides of $h'_{0,3},\dots h'_{3m-1,3}$.
Moreover, $k$ is the only event that $S_0\cup S_2$ and $S_1\cup S_2$ require to be border crossing that occurs in TSs of $U^{\sigma_5}_\varphi \setminus U(F^{\sigma_5}_K,H^{\sigma_5})$.
Hence, $S_0\cup S_2$ and $S_1\cup S_2$ are enhanceable to fitting supports of $U^{\sigma_5}_\varphi$.
By complementation of $S_0\cup S_2$ and $S_1\cup S_2$ in $U(F^{\sigma_6}_K,H^{\sigma_6})$, similarly, we obtain corresponding supports of $U^{\sigma_6}_\varphi$.
Finally, $S_3=$
\[
\{h'_{n,3}, h'_{n,4}, g_{n,3}, g_{n,4}\mid n\in \{0,\dots, 3m-1\}\}\cup \{t'_{n,n',0}, t'_{n,n',1}\mid n\in \{0,\dots, m-1\}, n'\in \{0,1,2\}\}
\]
respectively $S(U^{\sigma_6}_\varphi)\setminus S_3$, is a support of $U^{\sigma_5}_\varphi$, respectively $U^{\sigma_6}_\varphi$, inhibiting $m$ at $h'_{0,3},\dots h'_{3m-1,3}$.
\end{proof}

\begin{lemma}\label{lem:event_q_0} The event $q_0$ is inhibitable. \end{lemma}
\begin{proof}
The needed sets are:
\begin{enumerate}
%F_T for sigma_5
\item 
$S_0=\{b'_{n,2}, b'_{n,3}\mid n\in \{0,\dots, m-1\}\}$
%head
\item
$S_1=\{h'_{n,0}, h'_{n,1}\mid n\in \{0,\dots, 3m-1\}\}$
%freezer F_Z f'_{2,5},f'_{2,6}, f'_{2,7}, f'_{2,8}, f'_{2,9}
\item
$S_2=\{f'_{0,0}, f'_{0,1}, f'_{0,5}, f'_{1,2}, f'_{2,0}, f'_{2,1}, f'_{2,5},f'_{2,6}, f'_{2,7}, f'_{2,8}, f'_{2,9}\}$ 
%freezer F_Z f'_{0,5}, f'_{0,6}, f'_{0,7}, f'_{0,8}
\item
$S_3=\{f'_{0,0}, f'_{0,1}, f'_{0,5}, f'_{0,6}, f'_{0,7}, f'_{0,8}, f'_{1,2}, f'_{2,0}, f'_{2,1}\}$ 
% freezer F_z f'_{0,4}, f'_{2,4}
\item
$S_4=\{f'_{0,0}, f'_{0,4}, f'_{0,5}, f'_{0,6}, f'_{1,0}, f'_{1,4}, f'_{2,0}, f'_{2,4}, f'_{2,5}, f'_{2,8}\}$ 
% duplicator f'_{0,4}, f'_{2,4}
\item 
$S_5=\{d_{n,0}, d_{n,3}, d_{n,7}\mid n\in \{0,\dots, 18m-1\}\}$
% generator f'_{0,4}, f'_{2,4}
\item 
$S_6=\{g_{n,0}, g_{n,3}, g_{n,7}\mid n\in \{0,\dots, 3m-1\}\}$
\end{enumerate}
For $\sigma_5$ ($\sigma_6$) the support $S_0\cup S_1\cup S_2$ ($S(U^{\sigma_6}_\varphi)\setminus (S_0\cup S_1\cup S_2)$) allows the inhibition of $q_0$ at $f'_{0,0}, f'_{0,1}, f'_{2,0}, f'_{2,1} , f'_{2,5},f'_{2,6}, f'_{2,7}, f'_{2,8}, f'_{2,9}$.
The same does  $S_0\cup S_1\cup S_3$ ($S(U^{\sigma_6}_\varphi)\setminus (S_0\cup S_1\cup S_3)$) for $f'_{0,5}, f'_{0,6}, f'_{0,7}, f'_{0,8}$.
The remaining states are $f'_{0,4}, f'_{2,4}$.
For $\sigma_5$ the set $S_4\cup S_5\cup S_6$ is a support of $U(F^{\sigma_5}_K,D^{\sigma_5}, G^{\sigma_5})$, firstly, allowing the inhibition of $q_0$ at $f'_{0,4}, f'_{2,4}$ and, secondly, assuring that $k$ is the only border crossing event in $U^{\sigma_5}_\varphi\setminus U(F^{\sigma_5}_K,D^{\sigma_5}, G^{\sigma_5})$.
Hence, $S_4\cup S_5\cup S_6$ can be enhanced to a fitting support of $U^{\sigma_5}_\varphi$.
Finally, again by set complementation, we obtain a corresponding support for  $U^{\sigma_6}_\varphi$.
\end{proof}

\begin{lemma}\label{lem:event_q_1} The event $q_1$ is inhibitable. \end{lemma}
\begin{proof}
We need
\begin{enumerate}
% head & generator
\item 
$S_0=\{  h'_{n,0}, h'_{n,1}, g_{n,3} , d_{n',3}\mid n\in \{0,\dots, 3m-1\},  n'\in \{0,\dots, 18m-1\}   \}$
%F_K 
\item
$S_1=\{f'_{0,2}, f'_{0,3}, f'_{0,4}\}\cup (S(F'_2)\setminus \{f'_{2,4}, f'_{2,5}\})$
\item 
$S_2=\{f'_{0,1}, f'_{0,2}, f'_{0,3}, f'_{0,7}, f'_{0,8},  f'_{1,0}, f'_{1,4}, f'_{2,0}, f'_{2,8}\}$
\item 
$S_3=\{f'_{0,0}, f'_{0,3}, f'_{0,7},  f'_{1,0}, f'_{1,4}, f'_{2,1}, f'_{2,2}, f'_{2,8}\}$.
\end{enumerate}

The set $S_0\cup S_1$ is a support of $U^{\sigma_5}_\varphi$ that allows the inhibition of $q_1$ at all states in question besides of $f'_{0,0}, f'_{0,1}, f'_{0,7}, f'_{0,8}$.
The set $S_2$, respectively $S_3$, is a support of $F^{\sigma_5}_K$, firstly, allowing the inhibition of $q_1$ at $ f'_{0,1}, f'_{0,7}, f'_{0,8}$, respectively at $f'_{0,0}$, and, secondly, assuring that $k$ is the only border crossing event that occurs in $U^{\sigma_5}_\varphi \setminus F^{\sigma_5}_K$, too.
Hence, these sets can be enhanced to fitting supports of $U^{\sigma_5}_\varphi$.
Simply by set complementation, we obtain corresponding supports for $U^{\sigma_6}_\varphi$, too.
This proves the lemma.
\end{proof}

\begin{lemma}\label{lem:event_q_2_and_q_3} The events $q_2$ and $q_3$ are inhibitable. \end{lemma}
\begin{proof}
The events $q_2,q_3$ are inhibitable at $b'_{n,0},b'_{n,5}, b'_{n,6}$ for $n\in \{0,\dots, m-1\}$ and $f'_{1,0},f'_{1,4}, f'_{1,5}$ and $f'_{2,0},f'_{2,8}, f'_{2,9}$.
Hence, initially, we focus on the inhibition of $q_2$ at the states $b'_{n,3}, b'_{n,4}$ for $n\in \{0,\dots, m-1\}$ and $f'_{1,3}$ and $f'_{2,3},f'_{2,4}, f'_{2,5},f'_{2,6}, f'_{2,7}$.
The following sets are needed:
\begin{enumerate}
%inhibition in F^{\sigma_5}_T
\item 
$S_0=\{b'_{n,3}, b'_{n,4}, b'_{n,5}, b'_{n,6}, b_{n,2}, b_{n,3}, b_{n,4} \mid n\in \{0,\dots, m-1\}\}$,
\item 
$S_1=\{b'_{n,4}, b'_{n,5}, b'_{n,6} \mid n\in \{0,\dots, m-1\}\}$,
\item 
$S_2=\{t'_{n,n',13}, t'_{n,n',16}, t'_{n,n',19}\mid n\in \{0,\dots, m-1\}, n'\in \{0,1,2\} \}$
%inhibition in F^\sigma_K :  f'_{2,7}
\item
$S_3=\{f'_{1,3},f'_{1,4}, f'_{1,5}, f'_{2,7},f'_{2,8}, f'_{2,9} \}$
%inhibition in F^\sigma_K :  f'_{2,3},f'_{2,4}, f'_{2,5},f'_{2,6}
\item
$S_4=\{f'_{0,3},f'_{0,4}, f'_{0,5}, f'_{0,6},f'_{0,7}, f'_{0,8} , f'_{2,3},f'_{2,4}, f'_{2,5}, f'_{2,6},f'_{2,7}, f'_{2,8},f'_{2,9}\}$
\end{enumerate}
For $n\in \{0,\dots, m-1\}$ the set $S_0\cup S_2$ is a support that allows the inhibition of $q_2$ at $b'_{n,3}, b'_{n,4}$. 
The set $S_1\cup S_3$, respectively $S_4$, is the same for the inhibition at $f'_{2,7}$, respectively $f'_{2,3},f'_{2,4}, f'_{2,5},f'_{2,6}$.
By the symmetry of the occurrences of $q_2,q_3$ in $U^{\sigma_5}_\varphi$, the event $q_3$ can be shown to be inhibitable in a perfectly similar way.
Finally, again by set complementation, we obtain that $q_2$ and $q_3$ are inhibitable at all states in question in $U^{\sigma_6}_\varphi$ where, especially, here the inhibition at the states of $F^{\sigma_5}_T$ is clearly not necessary.
Moreover, with respect to ESSP the states $b_{n,2}, b_{n,3}, b_{n,4}$ in $S_0$ could have been removed.
But they are necessary for proving the SSP of $U^{\sigma_6}_\varphi$ in Lemma~\ref{lem:ssp_sigma_5_sigma_6}.
\end{proof}

\begin{lemma}\label{lem:event_a} The events $a_0,\dots, a_{18m-1}$ are inhibitable. \end{lemma}
\begin{proof} With

\begin{enumerate} 
\item$S^{\sigma_5}_0=\{d_{n,0}, d_{n,2}, d_{n,3}, d_{n,4}, d_{n,7}, d_{n,8}\mid n\in \{0,\dots, 18m-1\}\}$
\item$S^{\sigma_6}_0= \{d_{n,1}, d_{n,5}, d_{n,6} \mid n\in \{0,\dots, 18m-1\}\}$
\item
$S^{\sigma_5}_1=\{d_{n,1}, d_{n,3}, d_{n,4}, d_{n,7}, d_{n,8}\mid n\in \{0,\dots, 18m-1\}\}$
\item
$S^{\sigma_6}_1=\{d_{n,0}, d_{n,3},\dots, d_{n,6}\mid n\in \{0,\dots, 18m-1\}\}$
\item
$S^{\sigma_5}_2= S^{\sigma_6}_2=\{f'_{2,1},  \dots, f'_{2,7}, g_{n,1},\dots, g_{n,6} \mid n\in \{0,\dots, 3m-1\}  \}$
\end{enumerate}

the sets $S^\sigma_0\cup S^\sigma_2$ and $S^\sigma_1\cup S^\sigma_2$ are supports of $U(F'_2,D^\sigma,G^\sigma)$ that, altogether, allow the inhibition of $a_n$ in $D_n$ for $n\in \{0,\dots, 18m-1\}$ and $\sigma\in \{\sigma_5, \sigma_6\}$.
Moreover, for the corresponding regions the event $k$ is the only border crossing event that occurs in $U^\sigma_\varphi \setminus U(F'_2,D^\sigma,G^\sigma)$, too.
Hence, we can enhance $S^\sigma_0\cup S^\sigma_2$ and $S^\sigma_1\cup S^\sigma_2$ to fitting supports of $U^\sigma_\varphi$.
To complete the proof of the lemma, it remains to show the events $a_0,\dots, a_{18m-1}$ are inhibitable in $T^\sigma_\varphi$.
To do so, we prove for $i\in \{0,\dots, m-1\}$ and $\alpha_i\in \{0,1,2\}$, both arbitrary but fixed, that the events $a_{18i+6\alpha_i},\dots,a_{18i+6\alpha_i+5} $ are inhibitable in $T^\sigma_{i,\alpha_i}$ by regions of $U^\sigma_\varphi$.

Firstly, all $a_{18i+6\alpha_i},\dots,a_{18i+6\alpha_i+5}$ are inhibitable at $t'_{i,\alpha_i,13}, t'_{i,\alpha_i,16}$ and $t'_{i,\alpha_i,19}$.

Secondly, we show $a_{18i+6\alpha_i},\dots,a_{18i+6\alpha_i+5}$ to be inhibitable at all further states of $T^\sigma_{i,\alpha_1}$ in question, besides of $t'_{i,\alpha_i,2}, t'_{i,\alpha_i,5}, t'_{i,\alpha_i,8}$ and $t'_{i,\alpha_i,11}$.
We do so by presenting respective supports for the switch $\sigma_6$ and the simple set complementation yields corresponding regions for $\sigma_5$.
Finally, for the inhibition at $t'_{i,\alpha_i,2}, t'_{i,\alpha_i,5}, t'_{i,\alpha_i,8}$ and $t'_{i,\alpha_i,11}$ we present two further supports, both valuable for $\sigma_5$ and $\sigma_6$, and argue for the sufficiency of the implied regions.
To tackle the schedule, we need the following sets of states:
\begin{enumerate}
\item 
$M_0=\{t'_{i,\alpha_i,2}, t'_{i,\alpha_i,3},t'_{i,\alpha_i,4},  t'_{i,\alpha_i,5}, t'_{i,\alpha_i,8}, t'_{i,\alpha_i,11}\}$,
\item 
$M_1=\{t'_{i,\alpha_i,2}, t'_{i,\alpha_i,5}, t'_{i,\alpha_i,6},t'_{i,\alpha_i,7}  t'_{i,\alpha_i,8}, t'_{i,\alpha_i,11}\}$,
\item 
$M_2=\{t'_{i,\alpha_i,2}, t'_{i,\alpha_i,5}, t'_{i,\alpha_i,8},t'_{i,\alpha_i,9}  t'_{i,\alpha_i,10}, t'_{i,\alpha_i,11}\}$,
\item 
$M_3=\{t'_{i,\alpha_i,2}, t'_{i,\alpha_i,5}, t'_{i,\alpha_i,8}, t'_{i,\alpha_i,11}, t'_{i,\alpha_i,12}, t'_{i,\alpha_i,13} ,t'_{i,\alpha_i,14} \}$,
\item 
$M_4=\{t'_{i,\alpha_i,2}, t'_{i,\alpha_i,5}, t'_{i,\alpha_i,8}, t'_{i,\alpha_i,11}, t'_{i,\alpha_i,15}, t'_{i,\alpha_i,16},t'_{i,\alpha_i,17} \}$,
\item 
$M_5=\{t'_{i,\alpha_i,2}, t'_{i,\alpha_i,5}, t'_{i,\alpha_i,8}, t'_{i,\alpha_i,11}, t'_{i,\alpha_i,18}, t'_{i,\alpha_i,19},t'_{i,\alpha_i,20} \}$,
\item
$S_3=\{h'_{3i+\alpha_i,3}, h'_{3i+\alpha_i,4}, h'_{3i+\alpha_i,5}, g_{3i+\alpha_i,6}, g_{3i+\alpha_i,7}, g_{3i+\alpha_i,8}\}$.
\item 
If $M\subseteq S(T^{\sigma_6}_\varphi)$ then $Acc(M)=\{n\in \{0,\dots, 18m-1\} \mid \exists s\edge{a_n}s'\in T^{\sigma_6}_\varphi: s\in M, s'\not\in M \}$.
\end{enumerate}
For $n\in \{0,\dots,5\}$ the set $M_n\cup S_3\cup \{g_{n',6},g_{n',7},g_{n',8}\mid n'\in Acc(M_n)\}$ is a support that allows the inhibition of $a_{18i+6\alpha_i+n}$ at the remaining states of $T^{\sigma_6}_{i,\alpha_i}$, besides of the relevant states of $t'_{i,\alpha_i,2}, t'_{i,\alpha_i,5}, t'_{i,\alpha_i,8}$ and $t'_{i,\alpha_i,11}$.
By set complementation we obtain corresponding supports of $U^{\sigma_5}_\varphi$.
Hence, for $\sigma_5,\sigma_6$ it remains to show that the events $a_{18i+6\alpha_i},\dots,a_{18i+6\alpha_i+5}$ are inhibitable at the respective the relevant states of $t'_{i,\alpha_i,2}, t'_{i,\alpha_i,5}, t'_{i,\alpha_i,8}$ and $t'_{i,\alpha_i,11}$.
Let $j,\ell\in \{0,\dots, m-1\}\setminus \{i\}$ such that $X_{i,\beta_i}\in E(T^{\sigma_5}_j) \cap E(T^{\sigma_5}_\ell)$.
For $n\in \{j,\ell\}$ let $\alpha_n=0$ ($\alpha_n=1$, $\alpha_n=2$) if $X_{j,\beta_n}=X_{i,\beta_i}$ ($X_{j,\gamma_n}=X_{i,\beta_i}$, $X_{j,\alpha_n}=X_{i,\beta_i}$).
We need the following sets:
\begin{enumerate}
\item 
$N_0=\{t'_{n,\alpha_n,2},\dots,  t'_{n,\alpha_n,6}, t'_{n,\alpha_n,12},\dots,  t'_{n,\alpha_n,16}\mid n\in \{i,j,\ell\}\}$,
\item 
$N_1=\{t'_{n,\beta_n,2},t'_{n,\beta_n,3}, t'_{n,\beta_n,12},t'_{n,\beta_n,13}, \mid n\in \{i,j,\ell\}\}$,
\item 
$N_2=\{t'_{n,\gamma_n,2},\dots,  t'_{n,\gamma_n,9}, t'_{n,\gamma_n,12},\dots,  t'_{n,\gamma_n,19}\mid n\in \{i,j,\ell\}\}$,
\item 
$N_3=\{t'_{n,\alpha_n,7},\dots,  t'_{n,\alpha_n,11}, t'_{n,\alpha_n,15},t'_{n,\alpha_n,16},  t'_{n,\alpha_n,18},t'_{n,\alpha_n,19}, t'_{n,\alpha_n, 20} \mid n\in \{i,j,\ell\}\}$,
\item 
$N_4=\{t'_{n,\beta_n,4},\dots, t'_{n,\beta_n,11}, t'_{n,\beta_n,12},t'_{n,\beta_n,13}, t'_{n,\beta_n,15},\dots, t'_{n,\beta_n,20}, \mid n\in \{i,j,\ell\}\}$,
\item 
$N_5=\{t'_{n,\gamma_n,10},\dots,  t'_{n,\gamma_n,11}, t'_{n,\gamma_n,18}, t'_{n,\gamma_n,19} \mid n\in \{i,j,\ell\}\}$,
\item
$S^{\sigma_5}_4=\{h'_{n,3},h'_{n,4},h'_{n,5}  \mid n\in \{3n',3n'+1,3n'+2\}, n'\in \{i,j,\ell\}\}$.
\item 
$S^{\sigma_6}_4=\{g_{n,6},g_{n,7},g_{n,8}\mid n\in \{3n',3n'+1,3n'+2\}, n'\in \{i,j,\ell\}\}$.
\item
$S^{\sigma_5}_5=\{g_{n,6},g_{n,7},g_{n,8}\mid n\in \{3n',3n'+1,3n'+2\}, n'\in \{i,j,\ell\}\}$.
\item 
$S^{\sigma_6}_5= \{h'_{n,3},h'_{n,4},h'_{n,5}  \mid n\in \{3n',3n'+1,3n'+2\}, n'\in \{i,j,\ell\}\}$.
\item 
$sup^{\sigma}_0=N_0\cup N_1 \cup N_2\cup S^\sigma_4 \cup S(F^\sigma_T)$
\item 
$sup^{\sigma}_1=N_3 \cup N_4 \cup N_5 \cup S^\sigma_5 \cup S(F^\sigma_T)\cup  \{d_{n,6},d_{n,7}\mid n\in Acc(N_3\cup N_4 \cup N_5)\}$
\end{enumerate}
We can now argue as follows:
The region allowed by $sup^{\sigma_6}_0$ inhibits $a_{18i+6\alpha_i}, a_{18i+6\alpha_i+3}$ at $t'_{i,\alpha_i,8}, t'_{i,\alpha_i,11}$.
Moreover, the support $sup^{\sigma_6}_0$ ($sup^{\sigma_6}_1$) allows a region that inhibits $a_{18i+6\gamma_i+1}$, $a_{18i+6\gamma_i+4}$ ($a_{18i+6\beta_i+1}, a_{18i+6\beta_i+4}$) at $t'_{i,\gamma_i,11}$ ($t'_{i,\beta_i,2}$).
Finally, the region allowed by $sup^{\sigma_6}_1$ inhibits $a_{18i+6\alpha_i+2}, a_{18i+6\alpha_i+5}$ at $t'_{i,\alpha_i,2}, t'_{i,\alpha_i,5}$.
Hence, by the arbitrariness of $i$ and $\alpha_i$ and the symmetry of the construction, we have proven that the events $a_{18i+6\alpha_i},\dots,a_{18i+6\alpha_i+5}$ are inhibitable at the relevant states of $t'_{i,\alpha_i,2}, t'_{i,\alpha_i,5}, t'_{i,\alpha_i,8}$ and $t'_{i,\alpha_i,11}$ in $U^{\sigma_6}_\varphi$.
Similarly, we obtain this result for $U^{\sigma_5}_\varphi$.
\end{proof}

To show that the inhibition of the key event at the key state in $U^\varphi_{\sigma_5}$ implies the ESSP, it remains to show that the unique events are inhibitable, too.
As each of these events $u$ occurs exactly twice in $U^\varphi_{\sigma_5}$, that is, at a single backward and forward edge $s\fbedge{u}s'$, and, therefore, only in one single gadget TS, these events are inhibitable.
Hence, we state the next lemma without proof.
\begin{lemma}[Without proof] The unique events are inhibitable. \end{lemma}

Finally, it remains to prove that the inhibition of $k$ at $h_{0,5}$ implies the SSP for $U^\varphi_{\sigma}$:

\begin{lemma}\label{lem:ssp_sigma_5_sigma_6}  For $\sigma\in \{\sigma_5,\sigma_6\} $ the union $U^\sigma_\varphi$ has the SSP. \end{lemma}
\begin{proof}
Firstly, the initial state of any TS $A$, the state with an incoming/outgoing unique event labeled transition, installed by $U^\sigma_\varphi$ is separable from all the other states of $A$.
Secondly, if $i \in \{0,\dots,m-1\},\alpha_i\in \{0,1,2\}$ then $t_{i,\alpha_i,0}$ is separable.
Finally, for the solutions of all the other state separation atoms, we present the following table, where the entry of the first column states which TSs is investigated and the second row lists the lemmata where the separating regions can be found or, namely for the TS $D_j$, presents a region itself.

\begin{tabular}{  p{1cm}     p{10cm}}
TS & Separating Regions \\ \hline
%row 1
$T^{\sigma}_{i,\alpha_i}$
&
Lemma~\ref{lem:event_q_2_and_q_3}, Lemma~\ref{lem:event_a}\\ \hline
%row 2
$\mathcal{B}_j$
&
Lemma~\ref{lem:event_q_0}, Lemma~\ref{lem:event_q_2_and_q_3}\\ \hline
%row 3
$\mathcal{B}'_j$
&
Lemma~\ref{lem:event_q_0}, Lemma~\ref{lem:event_q_2_and_q_3}\\ \hline
%row 4
$H'_j$
&
Lemma~\ref{lem:event_k}, Lemma~\ref{lem:event_m}, Lemma~\ref{lem:event_q_0} \\ \hline
%row 5
$F'_0$
&
Lemma~\ref{lem:event_k}, Lemma~\ref{lem:event_m}, Lemma~\ref{lem:event_q_0} \\ \hline
%row 6
$F'_1$
&
Lemma~\ref{lem:event_z}, Lemma~\ref{lem:event_m}, Lemma~\ref{lem:event_q_0} \\ \hline
%row 7
$F'_2$
&
Lemma~\ref{lem:event_X}, Lemma~\ref{lem:event_q_1}, Lemma~\ref{lem:event_q_0} \\ \hline
%row 8
$D_j$
&
Lemma~\ref{lem:event_v},  Lemma~\ref{lem:event_z}, Lemma~\ref{lem:event_q_0}, Lemma~\ref{lem:event_q_1}, Lemma~\ref{lem:event_a}\\ \hline
%row 9
$G_j$
&
Lemma~\ref{lem:event_k},  Lemma~\ref{lem:event_z}, Lemma~\ref{lem:event_q_0}, Lemma~\ref{lem:event_a}, $S(U^{\sigma}_\varphi)\setminus\{g_{j,0}, g_{j,1}, g_{j,2}\}$
\end{tabular}

\end{proof}

%%%%%%%%%%%%%%%%%%%%%%%%%%%%%%%%%%%%%%%%%%%%%%%%%%%%%%%%%%%%%%%%%%%%%%%%%%%

\end{document}